\documentclass[reprint,amsfonts, amssymb, amsmath,  showkeys,aps,pra, superscriptaddress, onecolumn,longbibliography]{revtex4-2}
\usepackage{xcolor}

\usepackage{wrapfig}
\newcommand{\kket}[1]{|#1\rangle\!\rangle}

\newcommand{\bbrakket}[2]{\langle\!\langle#1 \rvert #2\rangle\!\rangle}

\usepackage{float}
\makeatletter
\let\newfloat\newfloat@ltx
\makeatother

\makeatletter
\newsavebox{\@brx}
\newcommand{\llangle}[1][]{\savebox{\@brx}{\(\m@th{#1\langle}\)}%
  \mathopen{\copy\@brx\kern-0.5\wd\@brx\usebox{\@brx}}}
\newcommand{\rrangle}[1][]{\savebox{\@brx}{\(\m@th{#1\rangle}\)}%
  \mathclose{\copy\@brx\kern-0.5\wd\@brx\usebox{\@brx}}}
\makeatother
\usepackage[english]{babel}
\usepackage[utf8]{inputenc}
\usepackage{graphics}
\usepackage{selinput}
\usepackage[normalem]{ulem}
\usepackage[shortlabels]{enumitem}

\usepackage{braket}
\usepackage{amsthm}
\usepackage{mathtools}
\usepackage{physics}
\usepackage{graphicx}
\usepackage[left=16mm,right=16mm,top=35mm,columnsep=15pt]{geometry} 
\usepackage{adjustbox}
\usepackage{placeins}
\usepackage[T1]{fontenc}
\usepackage{lipsum}
\usepackage{csquotes}
\usepackage{bm}

\usepackage[linesnumbered,ruled,vlined]{algorithm2e}
\SetKwInput{kwInit}{Init}

\newcommand{\fsnull}[1]{}
\newcommand{\old}[1]{}

\usepackage[makeroom]{cancel}
\usepackage[toc,page]{appendix}
\usepackage[colorlinks=true,citecolor=blue,linkcolor=magenta]{hyperref}

\usepackage{tikz}
\tikzset{every picture/.style=remember picture}

\usepackage[utf8]{inputenc}
\usepackage{graphicx}
\usepackage{xcolor}
\usepackage{amsmath}
\usepackage{amsthm}
\usepackage{bm}
\usepackage{bbm}
\usepackage{comment}
\usepackage{appendix}
\usepackage{mathdots}
\usepackage{lipsum}
\usepackage{verbatim}
\usepackage{natbib}
\usepackage{nccmath}
\usepackage{amsfonts} 

\providecommand{\calU}{\ensuremath{\mathcal{U}}}
\providecommand{\calV}{\ensuremath{\mathcal{V}}}
\providecommand{\calS}{\ensuremath{\mathcal{S}}}

\newcommand{\poly}{\operatorname{poly}}

\renewcommand{\geq}{\geqslant}
\renewcommand{\leq}{\leqslant}

\newcommand{\bs}{\textsf{BS}}

\def\calA{\mathcal{A}}
\def\calB{\mathcal{B}}
\def\calC{\mathcal{C}}
\def\calD{\mathcal{D}}
\def\calE{\mathcal{E}}

\def\calH{\mathcal{H}}
\def\calL{\mathcal{L}}
\def\calK{\mathcal{K}}

\def\calN{\mathcal{N}}
\def\calO{\mathcal{O}}
\def\calP{\mathcal{P}}

\def\be{\begin{equation}}
\def\ee{\end{equation}}
\def\bs{\begin{split}}
\def\e{\end{split}}
\def\ba{\begin{eqnarray}}
\def\bea{\begin{eqnarray}}

\def\tea{\end{eqnarray}}
\def\ea{\end{eqnarray}}
\def\eea{\end{eqnarray}}

\def\bbE{\mathbb{E}}

\newtheorem{theorem}{Theorem}
\newtheorem{lemma}{Lemma}
\newtheorem{corollary}{Corollary}
\newtheorem{observation}{Observation}

\newtheorem{claim}{Claim}
\newtheorem{example}{Example}
\newtheorem{definition}{Definition}

\usepackage{amssymb}
\usepackage{dsfont}

\usepackage{yhmath}

\def\be{\begin{equation}}
\def\te{\end{equation}}
\def\ee{\end{equation}}
\def\ba{\begin{eqnarray}}
\def\bea{\begin{eqnarray}}

\def\tea{\end{eqnarray}}
\def\ea{\end{eqnarray}}
\def\eea{\end{eqnarray}}

\usepackage{minitoc}
\usepackage{tocloft}

\renewcommand \partname{}
\begin{document}

\makeatletter
\makeatother
 
\doparttoc 
\faketableofcontents 

\title{Simulating quantum circuits with arbitrary local noise using Pauli Propagation}

\author{Armando Angrisani}
\email{armando.angrisani@epfl.ch}
\affiliation{Institute of Physics, Ecole Polytechnique Fédérale de Lausanne (EPFL),  Lausanne CH-1015, Switzerland}

\author{Antonio A. Mele}
\email{a.mele@fu-berlin.de}
\affiliation{Dahlem Center for Complex Quantum Systems, Freie Universit\"{a}t Berlin, 14195 Berlin, Germany}
\affiliation{Theoretical Division, Los Alamos National Laboratory, Los Alamos, NM 87545, USA}

\author{Manuel S. Rudolph}
\affiliation{Institute of Physics, Ecole Polytechnique Fédérale de Lausanne (EPFL),  Lausanne CH-1015, Switzerland}

\author{M. Cerezo}
\affiliation{Information Sciences, Los Alamos National Laboratory, Los Alamos, NM 87545, USA}

\author{Zo\"e Holmes}
\affiliation{Institute of Physics, Ecole Polytechnique Fédérale de Lausanne (EPFL),  Lausanne CH-1015, Switzerland}

\begin{abstract}
We present a polynomial-time classical algorithm for estimating expectation values of arbitrary observables on typical quantum circuits under any incoherent local noise, including non-unital or dephasing. Although previous research demonstrated that some carefully designed quantum circuits affected by non-unital noise cannot be efficiently simulated, we show that this does not apply to average-case circuits, as these can be efficiently simulated using Pauli-path methods. Specifically, we prove that, with high probability over the circuit gates choice, Pauli propagation  algorithms with tailored truncation strategies achieve an inversely polynomially small simulation error. This result holds for arbitrary circuit topologies and for any local noise, under the assumption that the distribution of each circuit layer is invariant under single-qubit random gates. Under the same minimal assumptions, we also prove that most noisy circuits can be truncated to an effective logarithmic depth for the task of {estimating} expectation values of observables, thus generalizing prior results to a significantly broader class of circuit ensembles. We further numerically validate our algorithm with simulations on a $6\times6$ lattice of qubits under the effects of amplitude damping and dephasing noise, as well as real-time dynamics on an $11\times11$ lattice of qubits affected by amplitude damping.
\end{abstract}

\maketitle

\section{Introduction}

Recent years have seen a healthy back-and-forth between experimental teams attempting to implement non-classically simulable quantum algorithms on quantum hardware and theorists attempting to demonstrate these experiments can be  efficiently classically simulated~\cite{arute2019quantum, gao2024limitations, madsen2022quantum, oh2023spoofing, kim2023evidence, rudolph2023classical, beguvsic2023fast, tindall2023efficient,kechedzhi2023effective,torre2023dissipative, liao2023simulation, beguvsic2023converged}. Central to these discussions is the role of hardware noise. In particular, while there exist quantum circuits that are widely believed to be computationally hard to simulate classically, these cannot be implemented exactly on currently available devices. Rather, in the current pre-fault-tolerant era, all circuits are invariably subject to hardware noise, and this issue forces researchers to make delicate compromises. On the one hand, circuits subject to noise are often simpler to classically simulate~\cite{ gao2024limitations, oh2023spoofing, aharonov2022polynomial, muller2016relative, fontana2023classical, schuster2024polynomial}. On the other hand, the search for circuits that are less affected by noise can push the user towards ones that are classically tractable~\cite{kim2023evidence, rudolph2023classical, beguvsic2023fast, tindall2023efficient,kechedzhi2023effective,torre2023dissipative, liao2023simulation, beguvsic2023converged}. Hence,  understanding the subtle connections between noise and efficient classical simulation is key to the search for quantum advantage.

Given that hardware errors stem from complex, undesired interactions with surrounding environments, a variety of mathematical models have been developed to investigate these effects~\cite{raginsky2003scaling, king2001minimal, king2002capacity, terhal2005fault, khatri2020information,garcia2023effects,fontana2022nontrivial}.
Nevertheless, the vast majority of works focused on noise models within the \emph{depolarizing} class.
This class includes noise channels that always increase the entropy of the system. Thus, if no measurements are performed and no fresh auxiliary qubits are supplied during the computation, the output of a quantum circuit composed of alternated layers of such noise and arbitrary unitary quantum gates is driven exponentially fast towards the maximally mixed state. This makes any non-trivial computation unfeasible beyond logarithmic depth, regardless of which gates are applied \cite{aharonov1996limitations,muller2016relative,hirche2020contraction, wang2020noise,wang2021can}.  

Another key feature of depolarizing noise models is that they have a strikingly simple representation in the basis of Pauli operators. Namely, under the action of \emph{local} depolarizing noise, the contributions of global Pauli operators are suppressed exponentially more than those of local Pauli operators.
This insight has been widely exploited in several recent works leveraging classical simulation techniques based on Pauli path summations, which are discrete Feynman path integrals written in the Pauli basis~\cite{rall2019simulation, aharonov2022polynomial, fontana2023classical, shao2023simulating, rudolph2023classical, schuster2024polynomial, gonzalez2024pauli, cirstoiu2024fourier}.
While its remarkable simplicity makes the depolarizing class an appealing model for analysis, it fails to account for many dissipation phenomena that routinely arise in quantum computation experiments~\cite{arute2019quantum, kandala2017hardware, chirolli2008decoherence, pino2021demonstration}.

Alongside the depolarizing class, realistic quantum noise can also be described by the \emph{dephasing} class and the \emph{non-unital} class, which both present qualitatively different features, as they do not necessarily increase the entropy of the quantum system.
Specifically, non-unital noise can decrease the entropy of the system, potentially driving
it towards a pure state. Moreover, as demonstrated in Ref.~\cite{ben2013quantum}, under some circumstances one can take advantage of non-unital noise by exploiting it as a form of quantum refrigerator, allowing for fault-tolerant quantum computation for exponential depth without fresh auxiliary qubits or mid-circuits measurements.
On the other hand, dephasing noise never decreases the entropy of the systems, but it leaves the entropy of computational basis states unchanged.
In both cases, investigating the properties of these noisy quantum circuits is arguably more challenging than analyzing those with depolarizing noise, as standard entropy accumulation tools~\cite{aharonov1996limitations,muller2016relative} cannot be used.

As the behavior of \emph{carefully designed} quantum circuits drastically differs under different sources of noise, it is natural to investigate the effect of arbitrary noise on \emph{typical} quantum circuits, i.e., circuits randomly sampled from suitable circuit distributions.
On one hand, random circuits affected by non-unital noise exhibit qualitatively different behavior from circuit with solely depolarizing noise, as shown in the context of random circuit sampling~\cite{fefferman2023effect} and variational quantum algorithms, where non-unital noise has been demonstrated to induce absence of barren plateaus~\cite{mele2024noise,singkanipa2024beyond}, in contrast to the depolarizing noise scenario~\cite{wang2020noise}. On the other hand, noisy random circuits exhibit an effective logarithmic depth~\cite{mele2024noise}, echoing the logarithmic depth barrier~\cite{aharonov1996limitations,muller2016relative} associated with depolarizing noise.

\begin{figure*}
    \centering
    \includegraphics[width=1\linewidth]{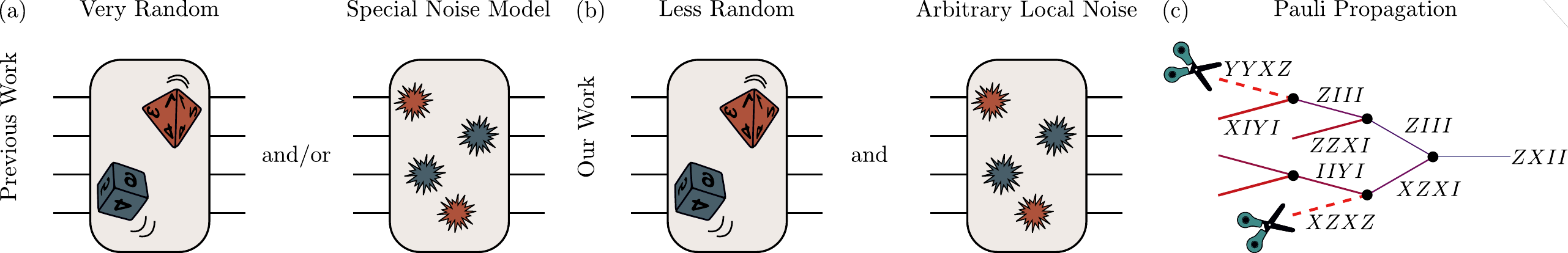}
    \caption{\textbf{Schematic summary of our main results.} a) Prior work has provided guarantees for simulating a quantum circuit either for highly random quantum circuits or assuming a special noise model. b) Here we show that the combination of a very general noise model (arbitrary local incoherent noise) with some randomness allows quantum circuits to be efficiently simulated. c) Our algorithm combines Pauli propagation with a truncation scheme based on the so-called \textit{path weight}, which is the \textit{total cumulative weight} of Paulis operators on a given branch~\cite{aharonov2022polynomial, schuster2024polynomial, gonzalez2024pauli} in contrast to \textit{current} Pauli weight~\cite{rudolph2023classical, schuster2024polynomial, angrisani2024classically}. }
    \label{fig:schematic} 
\end{figure*}

In this work we leverage the insight that circuit randomness and generic local noise behaves analogously to depolarizing noise to develop an efficient classical simulation algorithm for simulating typical quantum circuits in the presence of arbitrary local noise. Concretely, we show that expectation values of observables can be estimated with inversely polynomial precision in polynomial time on \emph{most} circuits and for \emph{any} incoherent local noise, possibly non-unital or dephasing, provided that each circuit layer is sampled from a distribution invariant under single-qubit random gates.
Our analysis combines the Pauli-path simulation algorithms also employed in Refs.\ \cite{aharonov2022polynomial, schuster2024polynomial, gonzalez2024pauli} with the normal form representation of local noise, previously employed in Refs.\ \cite{king2001minimal,ben2013quantum,mele2024noise}. 
Our results, sketched in Fig.~\ref{fig:schematic}, improve the state-of-the-art in the following ways:
\begin{itemize}
    \item \textbf{General noise model.} We make minimal assumptions about the noise model, considering circuit layers interspersed with single-qubit noise channels acting on all qubits. Rather than adopting a specific noise channel, as in Refs.\ \cite{aharonov2022polynomial, schuster2024polynomial}, we only require these single-qubit channels to be non-unitary with a constant noise rate.
    Our results also encompass dynamic non-unitary operations such as the ones considered in Ref.\ \cite{deshpande2024dynamic}, provided that they are applied with a constant rate.
    \item \textbf{Polynomial runtime for any geometry.} While Ref.~\cite{mele2024noise} previously analyzed random circuits with arbitrary local noise, it demonstrated efficient classical simulation for arbitrary circuit connectivity, but crucially only for constant precision. For inversely polynomial precision -- which, in many cases, corresponds to the desired accuracy (e.g., for estimating extensive quantities such as a molecule’s ground-state energy~\cite{peruzzo2014variational}) -- polynomial-time simulation was guaranteed only under the restriction of one-dimensional connectivity. In contrast, our results establish that any inversely polynomial error can be efficiently achieved for all circuit topologies, thereby matching the accuracy required in many physically-motivated scenarios.
    \item \textbf{Less random gate-sets.} Our assumptions about the circuit ensemble are also minimal, enabling efficient classical simulation for a broader class of circuits than those considered in previous works, such as Refs.\ \cite{aharonov2022polynomial, mele2024noise}. Specifically, we assume that each layer is independently sampled from a distribution invariant under different families of random single-qubit gates. In particular, for non-unital or depolarizing-like noise, we assume that the random single-qubit gates are sampled from unitary 1-designs. For dephasing noise, we assume the gates satisfy an approximate unitary 2-design property (e.g., they implement two random Pauli rotations occurring in orthogonal directions).
\end{itemize}

Under the same minimal assumptions on circuit structure and noise type, we demonstrate that for the task of estimating expectation values of observables, any incoherent noise reduces the circuit to an effective logarithmic depth. That is, any gates outside of this logarithmic depth window do not have a meaningful effect on expectation value estimates and thus can be ignored while classically simulating the circuit. This result significantly generalizes the main finding of Ref.\ \cite{mele2024noise}, which applies specifically to noisy circuits composed of two-qubit gates sampled from unitary 2-designs.

Finally, we complement our theoretical investigation with numerical simulations of a periodic $6 \times 6$ lattice evolving under a transverse-field Ising variational ansatz affected by amplitude-damping noise and dephasing noise. For our error estimates, we employ a Monte Carlo certification approach that can be used beyond exactly simulable regimes. The resulting analysis shows that Pauli propagation equipped with path-weight truncation substantially outperforms our analytic accuracy guarantees. Finally, we provide a large-scale example of a dynamical simulation of a noisy rotated transverse-field Ising model on an $11\times11$ square lattice with 121 qubits subject to amplitude damping noise. These structured circuits do not resemble typical random circuits and thus break the assumptions of our theoretical analysis. Despite this fact, we find that Pauli propagation remains reasonably efficient, thus showing its promise in regimes beyond our theoretical guarantees.

\section{Framework}

\noindent{\textbf{General local noise model.}} In general, any single-qubit channel $\calN$ can be written in the following \emph{normal form}\ \cite{king2001minimal,ben2013quantum, mele2024noise}
\begin{align}\label{eq:normalform}
    \calN(\cdot) = U \calN'(V(\cdot)V^\dag)U^\dag,
\end{align}
where $U,V$ are arbitrary single-qubit unitaries. The action of $\calN'$ on the single-qubit Pauli matrices is given by
\begin{align}
    &\calN'(I) = I + t_X X + t_Y Y + t_Z Z,
    \\&\calN'(X) = D_X X,
    \\&\calN'(Y) = D_Y Y,
    \\&\calN'(Z) = D_Z Z,
\end{align}
where $\bold{D} = (D_X, D_Y, D_Z) \in [-1,1]^3$ and $ \bold{t} = (t_X, t_Y, t_Z)\in [-1,1]^3$ are two 
vectors, which we refer as \emph{normal form parameters}.
We say that the channel $\calN$ has constant noise rate if $\frac{1}{3}\norm{\bold{D}}_2^2$ is a constant strictly smaller than one. 
Moreover, we say that $\calN$ is \emph{unital} if $\calN(I) = I$ and \emph{non-unital} otherwise.

\smallskip

There are in broad terms three different families of noise models, which we schematically sketch in Fig.~\ref{fig:noisemodels}. 
\begin{itemize}
    \item \textbf{Depolarizing-like noise:} This corresponds to unital channels satisfying $\norm{\bold{D}}_\infty < 1 $, i.e., to noise channels that drive \emph{any} input state towards the maximally mixed state. The most prominent example of depolarizing-like noise is the depolarizing channel, where $D_X=D_Y=D_Z$.
    \item \textbf{Dephasing-like noise:} This corresponds to unital channels satisfying $\abs{D_P} = 1$ for exactly one $P\in\{X,Y,Z\}$. Such channels leave invariant the quantum states that lie along a specific diameter of the Bloch sphere, while driving all other states towards such diameter. 
    \item \textbf{Non-unital noise:} Finally, non-unital noise channels are those that do not preserve the identity, i.e., $\calN(I) \neq I$. Repeatedly applying a non-unital channel lead the quantum state towards a \emph{fixed point} different than maximally mixed state. A common example of non-unital noise is the amplitude-damping channel, characterized by $\bold{D} = (\sqrt{1-\gamma}, \sqrt{1-\gamma}, 1- \gamma )$ and $\bold{t} = (0,0, \gamma)$ for some $\gamma \in (0,1]$. The fixed point of the amplitude damping channel is the computational zero state.
\end{itemize}
While our results encompass any local noise with a constant noise rate, i.e., $\norm{\bold{D}}_2 \in \Theta(1), \frac{1}{3}\norm{\bold{D}}_2^2 < 1$,  previous works have considered special cases. For instance, Refs.\ \cite{aharonov2022polynomial, fontana2023classical, schuster2024polynomial, gonzalez2024pauli} consider noise within the depolarizing class. Additionally, Ref.\ \cite{schuster2024polynomial} investigates randomized non-unital noise, which models spontaneous emissions occurring in random directions. This noise is characterized by normal form parameters given by $\bold{D} = ({1-\gamma/2}, {1-\gamma/2}, 1- \gamma)$ and $\bold{t} = (0,0, r\,\gamma)$, where $r$ is randomly set to either $1$ or $-1$ with equal probability.

\begin{figure*}
    \centering
\includegraphics[width=.95\linewidth]{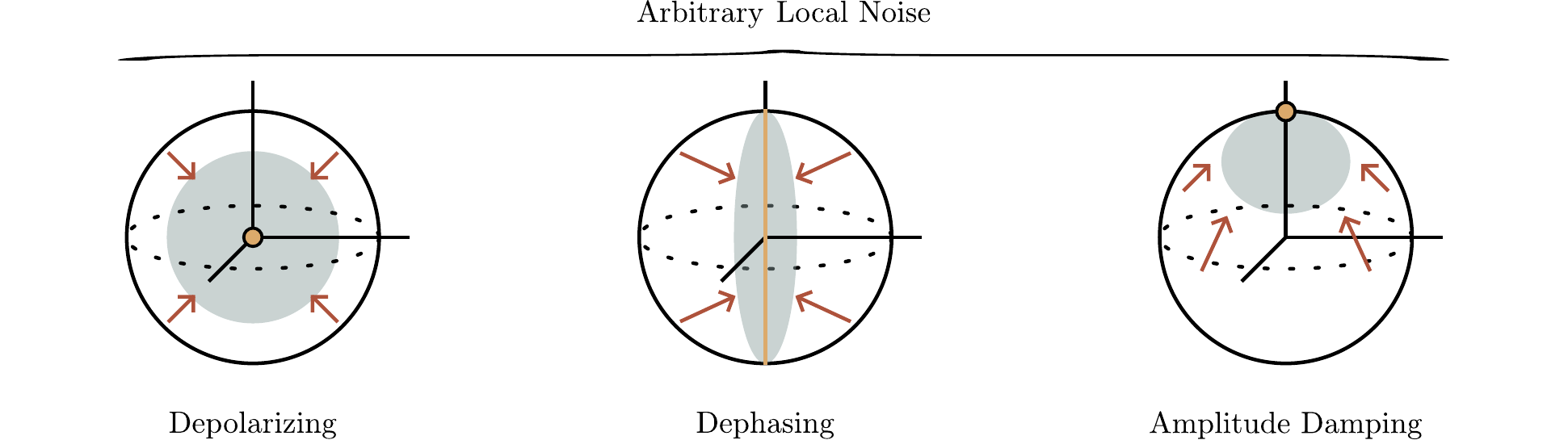}
    \caption{\textbf{Overview of noise models.} Here we provide a geometric sketch of three commonly considered families of local noise: a) Depolarizing noise, which drives any input state towards the maximally mixed state, b) Dephasing-like noise, which drives states to $z$-axis connecting poles of the Bloch sphere, and c) Non-unital noise, which drives states to a given fixed point on the Bloch sphere.  While prior work considered special cases of these families our analysis applies to any local noise channel with a constant noise rate. }
    \label{fig:noisemodels} 
\end{figure*}

\medskip

\noindent\textbf{Random local noise circuit models.}
In this work, we focus on the task of estimating expectation values such as $ \Tr[O \calC(\rho)]$. Here, $\rho$ is an arbitrary initial $n$-qubit quantum state, $O$ is a bounded Hermitian observable such that $\norm{O} \leq 1$ (where $\norm{\dots}$ denotes the spectral norm), and $\calC$ is an $L$-layered noisy random circuit of the form
\begin{align}
    \calC\coloneqq \calV^{\mathrm{single}} \circ \calN^{\otimes n} \circ \calU_{L} \circ \calN^{\otimes n} \circ \calU_{L-1} \circ \dots \circ \calN^{\otimes n} \circ \calU_{1} \, . \label{eq:noisy-circuit-main}
\end{align}
Above, the final layer $\calV^{\mathrm{single}}$ consists of single-qubit gates, the layers $\calU_j$ consist of non-overlapping gates, each acting on $\calO(1)$ qubits. Then, we assume for simplicity that each $\calN$ is an arbitrary single-qubit noise channel as defined in Eq.~\eqref{eq:normalform}. The extension to different local noise channels acting on different qubits is straightforward. Finally, we study the case when all random layers $\calU_1, \dots, \calU_L$ and $\calV^{\mathrm{single}} $  
are sampled independently from some distributions invariant (up to their second moment) 
under right-multiplication of random single-qubit gates.
In particular, we focus on two broad classes of random circuits, characterized by the distribution of these single-qubit gates.
\begin{itemize}
    \item \textbf{Locally unbiased circuits:} If the single-qubit gates are sampled from unitary 1-designs, then we say that the layers are sampled from \emph{locally unbiased} distributions. This generalizes the ensembles of circuits considered in Ref.\ \cite{aharonov2022polynomial}, which are composed by layers invariant under right-multiplication of random Pauli unitaries. 
    In this work, we demonstrate that circuits with locally unbiased layers affected by non-unital noise can be efficiently classically simulated.
    \item \textbf{Approximately locally scrambling circuits:} If the single-qubit gates are sampled from unitary 2-designs, then we say that the layers are sampled from \emph{locally scrambling} distributions\ \cite{caro2022outofdistribution, huang2023learning, angrisani2024classically}. In Appendix~\ref{sec:dist-maps} we introduce a generalization of this notion termed \emph{approximately locally scrambling} distributions. This broader framework encompasses, for example, arbitrary unitaries preceded by two consecutive random Pauli rotations along orthogonal directions. We further show that circuits composed of approximately locally scrambling layers, affected by arbitrary local noise, possibly dephasing-like, can be efficiently simulated classically.
\end{itemize}

\section{Related works}

This work combines several technical tools and insights from the previous literature. Here, we briefly review some previous works on noisy quantum circuits and on classical simulation of quantum circuits via Pauli-path methods.

\smallskip

\noindent\textbf{Classical simulation of quantum circuits with noise beyond the `depolarizing assumption'.}
Previous works have significant advanced the understanding of quantum noise beyond the depolarizing model. Under certain circumstances, different noise models can lead to qualitatively different scenarios. A prominent example is the so-called ``quantum refrigerator''~\cite{ben2013quantum}, which demonstrates that non-unital noise can be turned into a resource for fault-tolerance, removing the need for fresh auxiliary qubits
or the mid-circuit measurements required for implementing error-correcting schemes.
In particular, for a noise rate below the error-correcting threshold, quantum computation under non-unital noise remains possible for $\Theta(\exp(n))$ time~\cite{ben2013quantum, fawzi2022lower}.
An analogous, yet weaker, construction is also possible under dephasing noise, which allows for arbitrary quantum computation for polynomial time~\cite{ben2013quantum}. As there exist polynomial depth quantum circuits that are widely believed to be hard to classically simulate, these results imply that there exist at least some (carefully designed) circuits subject to non-unital noise and dephasing noise that cannot be efficiently simulated classically. 

More recently, a series of works investigated the impact of arbitrary local noise on random quantum circuits~\cite{quek2022exponentially, fefferman2023effect, mele2024noise, crognaletti2024estimates}. 
Even in this setting, some stark differences arise between unital and non-unital noise. In the latter case, noisy random circuits do not exhibit anti-concentration of the output distribution~\cite{fefferman2023effect} and exponential concentration of local observables~\cite{mele2024noise, singkanipa2024beyond}. 
Loosely speaking, these effects occur because non-unital noise effectively cools down the qubits throughout the computation.
On the other hand, when averaged over random gates, any local noise induces an \emph{effective depth}~\cite{mele2024noise}, that is any noise ``truncates''’ most quantum circuits
to effectively logarithmic depth, for the task of estimating observable expectation values.
In particular, as discussed in Ref.~\cite{mele2024noise}, the existence of this noise-induced effective depth allows one to classically estimate expectation values of  observables with inverse-polynomial precision in quasi-polynomial time if the circuit has constant geometric locality. The runtime improves to polynomial if the circuit architecture is one-dimensional.
In the present paper, we conduct a more fine-grained analysis of this phenomenon, by extending it to a more general class of random circuits on one hand, and by employing a more sophisticated simulation algorithm on the other hand, achieving inversely polynomial precision in polynomial time across all circuit architectures.

\smallskip

\noindent\textbf{Classical simulation of noisy circuits via Pauli propagation algorithms.}
Pauli propagation algorithms have gained momentum in recent years as powerful tools for simulating quantum circuits in a variety of settings. For circuits interspersed with noise within the depolarizing class, polynomial-time classical simulation with inverse-polynomial precision has been established, both for sampling from the output distribution of random circuits~\cite{aharonov2022polynomial} and for estimating expectation values~\cite{fontana2023classical, shao2023simulating}. 
Recently, it was shown that arbitrary circuits under depolarizing noise, when applied to a random input state, are also classically simulable in polynomial time with inverse-polynomial precision~\cite{schuster2024polynomial}.

While the results mentioned above rely on average-case assumptions, worst-case results for simulating noisy circuits using Pauli propagation have also been demonstrated in specific regimes~\cite{aharonov1996limitations, rall2019simulation, gonzalez2024pauli,schuster2024polynomial}. 
Namely, any circuit under depolarizing noise is classically simulable at depths larger than logarithmic in the system size, as the output state is inversely polynomially close to the maximally mixed state~\cite{aharonov1996limitations}.
Conversely, at sub-logarithmic depth, circuits dominated by Clifford gates~\cite{rall2019simulation, gonzalez2024pauli} or those with input states that are nearly maximally mixed, such as in the one clean-qubit DQC1) model of computation~\cite{schuster2024polynomial}, are also classically simulable under local depolarizing noise.

Beyond circuits with depolarizing noise, Ref.~\cite{schuster2024polynomial} investigates a \emph{specific} example of non-unital noise that effectively acts as depolarizing noise. This enables the estimation of expectation values of observables with inverse-polynomial precision in polynomial time across all circuit architectures, given a mild form of randomness in the choice of the input state. The noise model in question involves spontaneous emissions occurring in random directions. As the authors discuss, such randomization arises naturally in random circuits as an effect of the random gates applied before and after each noise channel.
In our work, we extend this analysis by explicitly identifying several classes of noisy random circuits which 
behave effectively as subject to depolarizing noise, and providing a unified treatment for \emph{any} arbitrary local noise, possibly non-unital or dephasing-like. Another extension of Pauli-path simulation to non-unital noise is presented in Ref.\ \cite{martinez2025efficient}, which focuses on circuits consisting of Clifford gates and random Pauli rotations interspersed with amplitude damping noise. This work introduces a randomized variant of the Pauli Propagation algorithm specifically designed for circuits affected by non-unital noise. We summarize the state-of-the-art results for estimating observables of noisy circuits in Table\ \ref{tab:comparison}.

\smallskip

\noindent\textbf{Classical simulation of noiseless circuits via Pauli propagation algorithms.}
Moving beyond noisy devices, it has been recently proven that Pauli-path methods allow for the estimation of expectation values in noiseless random quantum circuits in polynomial time for arbitrarily small constant precision, and in quasi-polynomial time for inverse-polynomial precision~\cite{angrisani2024classically}. Furthermore, this can be improved to almost-polynomial time if the circuit has constant geometric locality and a depth of at most poly-logarithmic scale. 

Crucially, the algorithm in Ref.~\cite{angrisani2024classically} truncated all Pauli operators above a certain weight threshold. In contrast, the algorithm considered in the present paper truncates Pauli paths based on their path weight, as sketched in Fig.~\ref{fig:schematic}. This distinction explains the differing runtimes of the two methods. Notably, for inversely polynomial precision, the runtime for noiseless circuits in Ref.~\cite{angrisani2024classically} is significantly higher than that achieved in our work for noisy circuits.

Moreover, Pauli-path methods are also provably efficient for noiseless near-Clifford circuits -- those composed of Clifford gates and Pauli rotations by sufficiently small angles~\cite{beguvsic2023fast, beguvsic2023simulating, lerch2024efficient, zhang2024clifford,mitarai2022quadratic}.
Refs.~\cite{rudolph2023classical, bermejo2024quantum} provide further numerical evidence of the performance of Pauli Propagation on noiseless circuits of particular relevance in the near-term and early-fault-tolerant era, such as Trotterized circuits for Hamiltonian simulations~\cite{rudolph2023classical} and Quantum Convolutional Neural Networks (QCNNs)~\cite{bermejo2024quantum}.

A further application of Pauli-path methods consists in classically spoofing the linear cross-entropy benchmark used in quantum supremacy experiments~\cite{gao2024limitations, aharonov2022polynomial, tanggara2024classically, angrisani2024classically}.

\begin{table}[h]
    \centering
    \begin{tabular}{|c|c|c|}
        \hline
        \textbf{Local noise model} & \textbf{Average-case assumptions} & \textbf{Runtime} \\
        \hline
         Noiseless & locally scrambling layers  & \(  n^{\calO(\log(n))} \)\ \cite{angrisani2024classically} \\
        \hline
        Depolarizing & random input state & $\poly(n)$\ \cite{schuster2024polynomial} \\
        \hline
        Spontaneous emission & emissions occurring in random  & $M\cdot \poly(n)\ $ \cite{schuster2024polynomial}         \vspace{-3pt}\\
             & directions, random input state &  \\
        \hline
        Depolarizing-like\ \cite{fontana2023classical, shao2023simulating}, & circuit with Clifford gates  & $M \cdot \poly(n)$
        \vspace{-3pt}\\
         amplitude damping\ \cite{martinez2025efficient}     & and random Pauli rotations&  \\
        \hline
        Arbitrary non-unital noise  & locally unbiased layers & $M\cdot \poly(n)$ [Thm.\ \ref{thm:non-u-inf}] \\
        \hline
        Arbitrary unital noise  & approximately locally & $\poly(n)$ [Thm.\ \ref{thm:deph-inf}] \vspace{-3pt} \\
        (including dephasing-like) & scrambling layers &  \\
        \hline
    \end{tabular}
    \caption{Comparison of runtimes for previous and our results established under different noise models and average-case assumptions. The goal is to estimate the expectation value of an observable $O$ with inversely polynomial precision and inversely polynomial failure probability for a polynomial depth circuit with an arbitrary topology. We assume that the observable $O$ satisfies $\norm{O}\leq 1$ and that it contains at most $M$ different Pauli terms.}
    \label{tab:comparison}
\end{table}

\section{Results}

\noindent \textbf{Theoretical guarantees.} 
Our first contribution establishes the efficient classical simulability of typical quantum circuits interspersed by arbitrary \textit{non-unital} noise. For inversely polynomial precision, the computational cost of our Pauli propagation algorithm (described in Section~\ref{sec:methods}) is polynomial if the measured observable contains polynomially-many Pauli operators. The runtime of our algorithm thus aligns with the polynomial scalings found in Refs~\cite{martinez2025efficient, schuster2024polynomial} (see Table~\ref{tab:comparison}) but holds for a wider class of noise models. 
Namely, our result holds for arbitrary local non-unital noise and any random noisy circuit whose layers are sampled independently from a distribution invariant under single-qubit unitary 1-designs. We stress that some form of average case assumption over circuits 
is essential due to the existence of the ``quantum refrigerator'' scheme~\cite{ben2013quantum}, as discussed in detail in Appendix\ \ref{app:necessity}.

\begin{theorem}[Non-unital noise, informal]
\label{thm:non-u-inf}
Let $O$ be an observable expressed as linear combination of $M$ Pauli operators, and let $\rho$ be an initial state. Assume the Pauli coefficients
of $O$ and $\rho$ can be efficiently computed.
Let $\calC$ be a random noisy circuit whose layers are sampled independently from a distribution invariant under single-qubit unitary 1-designs.  Assume the noise to be local and non-unital, and the noise rate to be constant.
Then the noisy expectation value $\Tr[O\calC(\rho)]$ can be estimated classically in time $M \cdot \poly(n, \epsilon^{-1}, \delta^{-1})$ within error $\epsilon\norm{O}$ and with probability $1-\delta$ over the circuit randomness.
\end{theorem}
This result also closes a gap in understanding the relationship between barren plateaus and classical simulability~\cite{cerezo2023does} in the context of generic noisy circuits. While previous works have shown that non-unital noisy random circuits do not exhibit barren plateaus~\cite{mele2024noise} (although due to only their last few layers), only a quasi-polynomial time classical average-case algorithm was known to achieve inverse-polynomial precision in estimating observable expectation values. 
Our result improves upon this by reducing the time complexity to polynomial time, demonstrating that in this physically motivated non-unital noisy setting, the absence of barren plateaus is indeed accompanied by efficient average-classical simulation. 
More specifically, this provides a concrete example where we observe absence of barren plateaus~\cite{mele2024noise}, efficient average classical simulability (our work), and provable worst-case hardness for classical simulability~\cite{ben2013quantum}.

We also establish the efficient classical simulability of typical quantum circuits interspersed by arbitrary unital noise, possibly within the \textit{dephasing} class.
In this case, the computational cost for inversely polynomial precision is polynomial for all observables.
Here, we make a slightly stronger assumption on the circuit ensemble, by assuming the circuit layers to be sampled from distributions invariant under single-qubit gates satisfying a property which we call \emph{approximate scrambling} (cf. Definition~\ref{def:mixing}). This natural property is satisfied by a broad class of gate-sets, such as Clifford gates and consecutive Pauli rotations along two orthogonal axes. 

\begin{theorem}[Unital noise, informal]
\label{thm:deph-inf}
Let $O$ be an observable and let $\rho$ be an initial state. Assume the Pauli coefficients
of $O$ and $\rho$ can be efficiently computed.
Let $\calC$ be a random noisy circuit whose layers are sampled independently from a distribution invariant under approximately scrambling single-qubit gates. Assume the noise to be local and unital, possibly dephasing-like, and that the noise rate is constant.
\begin{itemize}
    \item The noisy expectation value of any observable $O$ expressed as a linear combination of $M$ Pauli operators can be estimated classically in time $M\,\mathrm{poly}(n, \epsilon^{-1}, \delta^{-1})$  within additive error $\epsilon\norm{O}$ with probability at least $1-\delta$ over the circuit randomness.
    \item If the circuit depth is at least logarithmic in the number of qubits $n$, then the noisy expectation value of any observable $O$ can be estimated classically in time  $\mathrm{poly}(n)$ within additive error $\epsilon\norm{O}$ with probability at least $1-\delta$ over the circuit randomness,  for arbitrary $\epsilon,\delta$ decaying inverse-polynomially in $n$.
\end{itemize} 
\end{theorem}
Finally, we prove that any local noise effectively truncates typical circuits to a logarithmic depth for the task of estimating expectation values, generalizing the main result of Ref.~\cite{mele2024noise} to a broader class of circuits. This result implies that for inversely polynomial precision, the full circuit $\mathcal{C}$ can effectively be replaced by a logarithmic-depth circuit, specifically by its last logarithmic-many layers, which can potentially be simulated classically.

\begin{theorem}[Noise-induced shallow depth, informal]
\label{thm:main-eff}
Let $O$ be an observables and consider an $L$-layered circuit $\calC$ randomly sampled from $\calD_{\mathrm{circ}}$ where $\calD_{\mathrm{circ}}$ is a distribution over noisy circuits satisfying the assumptions of Theorems\ \ref{thm:non-u-inf} or\ \ref{thm:deph-inf}. 
Consider the ``truncated'' noisy circuit 
\begin{align}
        \mathcal{C}_{[L, L-j]}  \coloneqq \mathcal{V}^{\mathrm{single}} \circ \mathcal{N}^{\otimes n} \circ \mathcal{U}_{L} \circ \mathcal{N}^{\otimes n} \circ \mathcal{U}_{L-1} \circ \dots \circ \mathcal{N}^{\otimes n} \circ \mathcal{U}_{L-j}, 
\end{align}
for a suitable $j \in \mathcal{O}(\log(1/\epsilon))$.
With high probability over the choice of $\mathcal{C}$, it holds that 
\begin{align}
    \abs{\Tr[O\mathcal{C}(\rho)] - \Tr[O\mathcal{C}_{[L, L-j]}(\sigma)]} \leq \epsilon \|O\|,
\end{align}
where $\sigma$ is an arbitrary state.
\end{theorem}

It is noteworthy that combining this result with the recent findings from Ref.~\cite{angrisani2024classically}, already provides a non-trivial simulation algorithm for noisy random circuits. In particular, Ref.~\cite{angrisani2024classically} demonstrated that observables of logarithmic-depth circuits can be estimated with inversely polynomial precision in almost-polynomial time (i.e., $n^{\calO(\log\log(n))}$-time) for circuits with constant geometric dimension~\footnote{The results in Ref.~\cite{angrisani2024classically} were originally established for noiseless circuits but can be extended straightforwardly to noisy circuits.}.
However, the bounds derived in this paper are significantly stronger, as they establish a polynomial runtime for noisy circuits across all architectures.

\medskip

\begin{figure}
    \centering
    \includegraphics[width=0.6\linewidth]{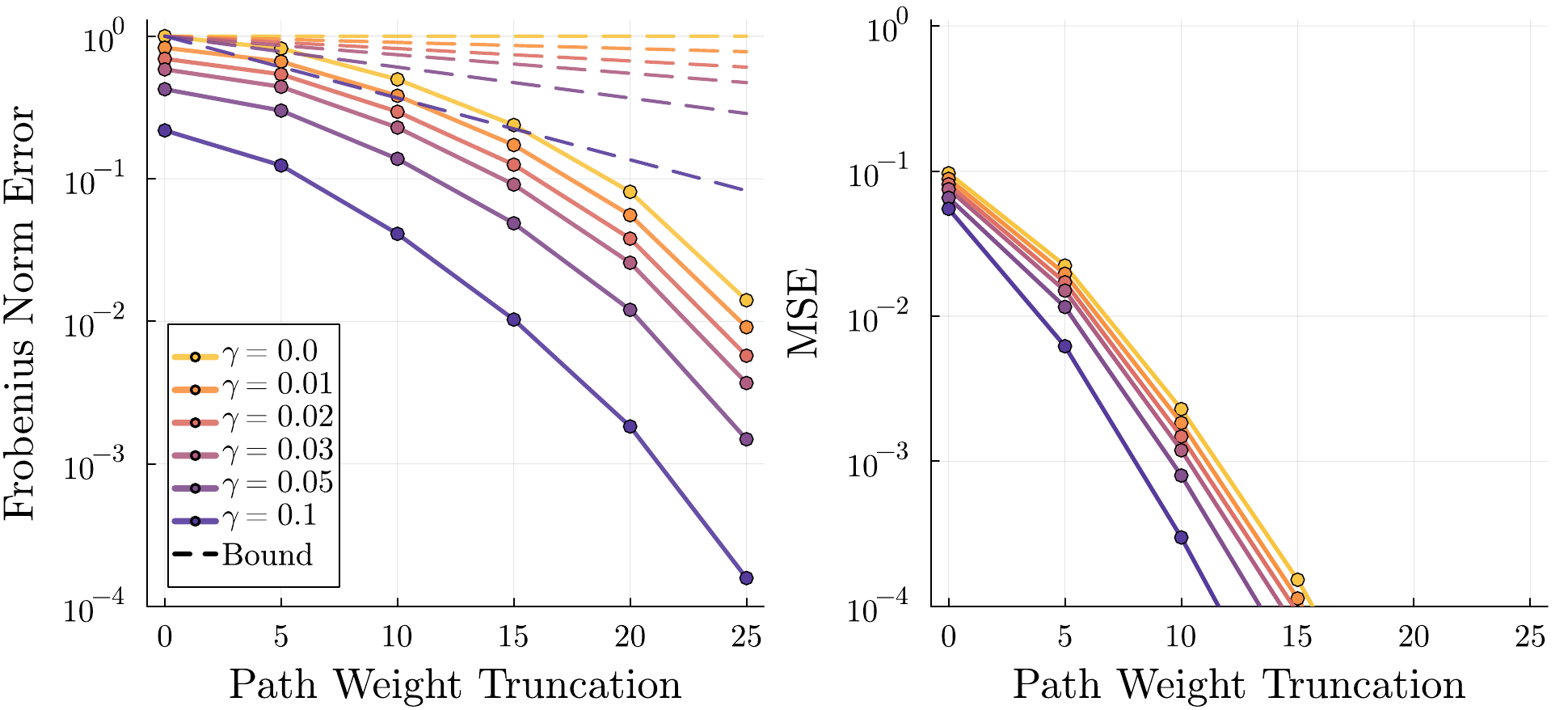}
    \caption{\textbf{Error of simulating circuits with amplitude damping noise}. a) Frobenius norm error and b) mean square error (MSE) for different path-weight truncation values and noise strengths in a periodic $6\times6$ lattice of qubits. The observable is a Pauli-Z operator in the middle of the lattice with a circuit consisting of RX, RZ single-qubit rotations and RZZ entangling gates and in b) the initial state is the all zero state. The theoretical bound, established in Thm.\ \ref{thm:mse-core}, is calculated via $(1-\gamma + \gamma^2)^k$, where $k$ is the path-weight truncation order and $1-\gamma + \gamma^2$ is the squared contraction coefficient computed in Example\ \ref{obs:coeff-noise}.}
    \label{fig:ampdamp_error}
\end{figure}

\noindent \textbf{Numerical implementation.} We complement our theoretical analysis with precise numerical estimates of the average simulation error for a Hamiltonian variational ansatz~\cite{wecker2015progress} and a large-scale simulation of noisy quantum dynamics. All numerical results were collected with the \href{https://github.com/MSRudolph/PauliPropagation.jl}{PauliPropagation.jl} package.

As shown in Figs.~\ref{fig:ampdamp_error} and~\ref{fig:dephasing_error}, the average simulation error encountered in practice is often significantly smaller than predicted by our bounds, demonstrating that Pauli-path methods are very well-suited for simulating typical noisy circuits. 
Both simulations treat the same system, but with varying degrees of amplitude damping noise (Fig.~\ref{fig:ampdamp_error}) or dephasing noise (Fig.~\ref{fig:dephasing_error}). The quantum circuits consist of parametrized RX and RZ Pauli rotation gates per qubit followed by parametrized RZZ Pauli rotation gates on a periodic $6\times6$ latice on 36 qubits. This corresponds to an HVA ansatz of the so-called tilted transverse-field Ising model. The observable here is a Pauli-Z operator in the middle of the lattice, and the initial state is the all-zero state.  
We emphasize that this circuit ensemble would not satisfy the stronger assumptions of other works, such as Refs.~\cite{mele2024noise, angrisani2024classically}, which require the presence of local 2-designs. 
Moreover, we note that simulating deep average-case circuits with amplitude damping noise is arguably more interesting than with noiseless or depolarizing noise~\cite{mcclean2018barren, angrisani2024classically, wang2020noise}, due to the fact that such circuits can experience `absence of barren plateaus'~\cite{mele2024noise}, highlighting the need for non-trivial classical algorithms to tackle such cases. In Appendix~\ref{sec:addnumer}, we further elaborate on this and demonstrate the absence of barren plateaus for our circuit ensemble, which does not meet the stronger assumptions used in Refs.~\cite{mele2024noise}.

Our bounds on the Frobenius norm error of the backpropagated observable as a function of the path-weight truncation order are well satisfied (see panels (a)). This norm also indirectly bounds the error of expectation values with random initial states. We show, however, that the mean square error (MSE) of expectation values with the ubiquitous all-zero initial state is orders of magnitude lower at the same path-weight truncation order (see panels (b)). These error estimates were obtained via the Monte Carlo approach provided in Theorem~\ref{thm:num-estimate} in the Appendices. Overall, we conclude that the simulation error tends to decay exponentially with the path-weight truncation order, and that circuits affected by stronger noise are continually easier to simulate on average.

\begin{figure}
    \centering
    \includegraphics[width=0.6\linewidth]{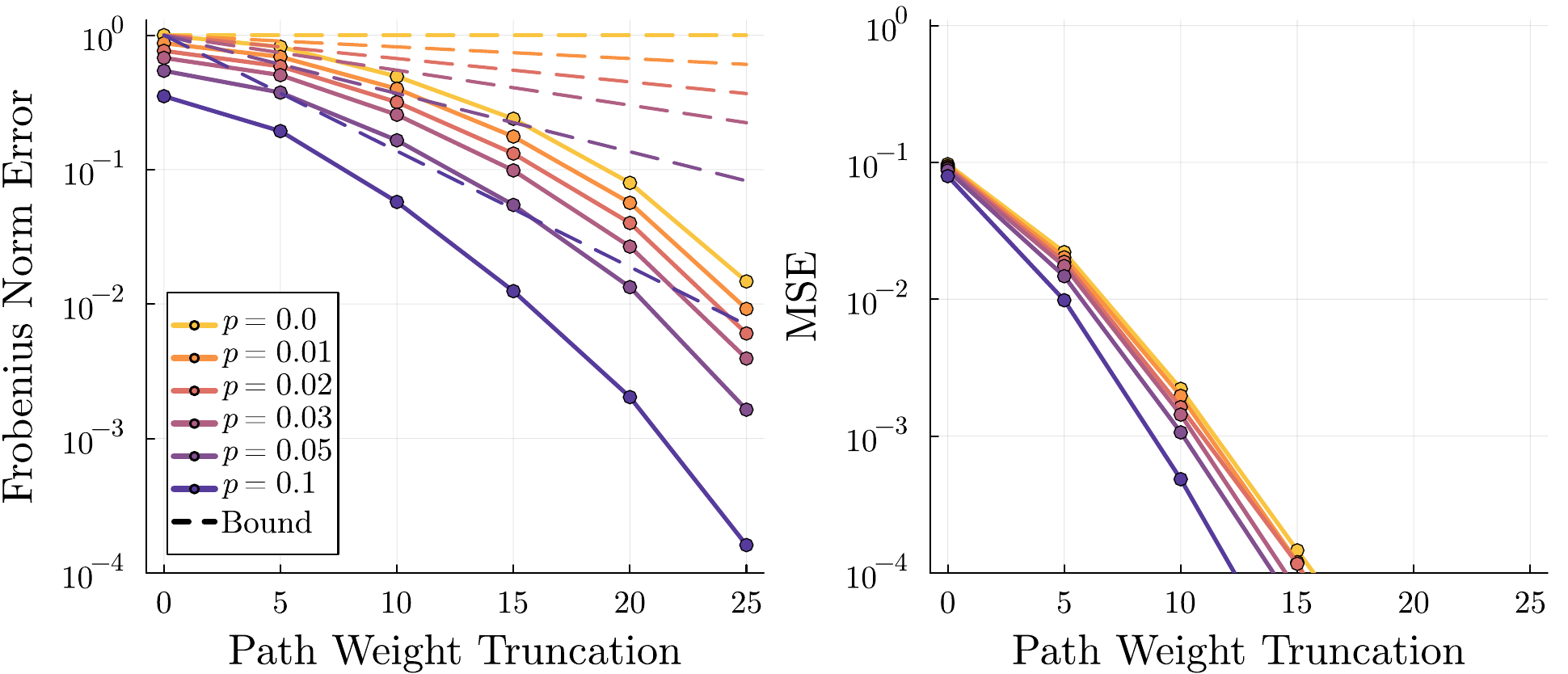}
    \caption{\textbf{Error of simulating circuits with dephasing noise}.  a) Frobenius norm error and b) mean square error (MSE) for the same 36-qubit simulation as in Fig.~\ref{fig:ampdamp_error}. The theoretical bound, established in Thm.\ \ref{thm:mse-core}, is calculated via $\left(\frac{1+(1-2p)^2}{2}\right)^k$, where $k$ is the path-weight truncation order and $\frac{1+(1-2p)^2}{2}$ is the mean squared contraction coefficient computed in Example\ \ref{obs:coeff-noise}.}
    \label{fig:dephasing_error}
\end{figure}

To unquestionably go beyond our theoretical guarantees, and to position Pauli propagation as a key tool for simulating quantum systems, we showcase a large-scale simulation of noisy real-time dynamics.
As with most numerical simulations that go beyond brute-force computation, it becomes difficult to judge the reliability of the results. 
To address this we start by using a system from Ref.~\cite{beguvsic2024real}, which compares the converging Pauli propagation results to iPEPS simulation data. 
Namely, we consider an $11\times11$ square qubit lattice with a rotated transverse-field Ising Hamiltonian,
\begin{equation}
    H = - J \sum_{\langle i,j\rangle}  X_iX_j -h \sum_i Z_i\,,
\end{equation}
where the XX interaction terms with strength $J$ act on neighbors on the 2D lattice, and $h$ is the local field strength. Importantly, we choose $h=1$ and $J=3.004438$, which puts the system at a quantum critical point~\cite{blote2002cluster}. The quantum circuit is then constructed as the second-order Trotterization of this Hamiltonian with time step $dt=0.04$ and up to 23 Trotter steps. The initial state is the all-zero state, and we estimate the expectation value of a Pauli-Z operator in the middle of the 2D lattice as a function of time. 

Ref.~\cite{beguvsic2024real} truncates propagating Pauli operators with more than 5 X or Y Paulis (this is not Pauli weight, which also counts Z Paulis) and Pauli operators with small coefficients below a magnitude of $2^{-18}$ up to $2^{-23}$. They then use their comparison with an iPEPS simulation to argue that these truncations are sufficient to ensure convergence. We here employ the same XY-Pauli truncation with the most challenging coefficient truncation threshold of $2^{-23}$, which brings the noise-free system close to apparent convergence with more than 1TB of memory usage. Then, we insert local amplitude damping noise into the quantum circuit and simulate the evolution with path-weight truncation on a laptop. 

To probe whether our substantially less resource-intensive noisy simulations are reliable we continually loosen the algorithmic truncations and observe convergence to expectation values. 
Based on the results with path-weight truncation values of 20 (dotted) and 25 (dashed) in Fig.~\ref{fig:dynamics}, which diverge at time $t=0.4$ and $t=0.6$, respectively, we believe that, given the exponential convergence present in Fig.~\ref{fig:ampdamp_error}, the simulation of the least-noisy system with $\gamma=0.05$ with path-weight truncation of 30 can be trusted to approximately time $t=0.8$. Our peak memory usage was around 12GB, and the full simulation took approximately one hour. However, systems affected by more noise ($\gamma=0.1$ and $0.2$) appear to require lower truncation values and are likely more accurate until the final time of $t=0.92$. 

We stress that this setup strongly deviates from the theoretical assumptions needed to derive our bounds. Namely, we are considering a highly structured circuit with correlated angles. So far, no theoretical accuracy guarantees have been given in this regime. But we show that for systems of interest, even at a quantum critical point, Pauli propagation equipped with suitable truncations can provide accurate results.

\begin{figure}
    \centering
    \includegraphics[width=0.4\linewidth]{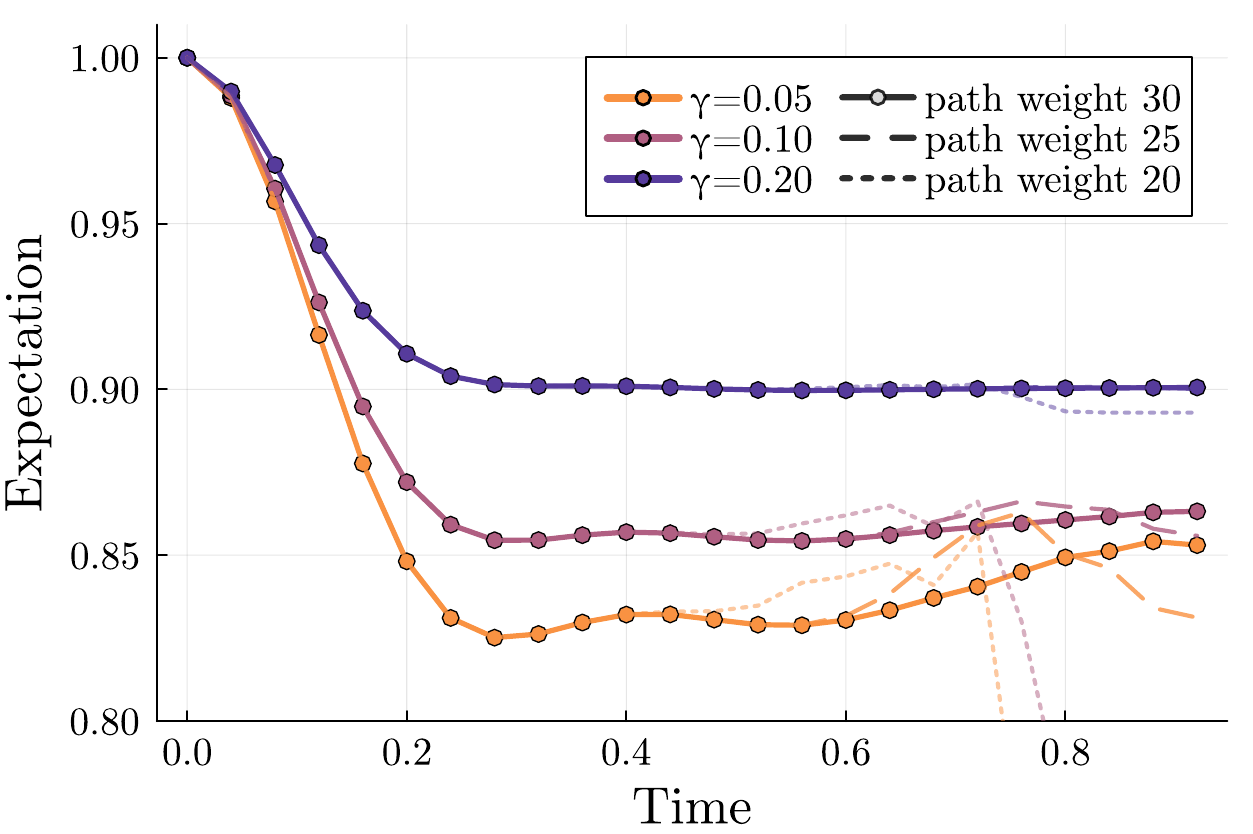}
    \caption{\textbf{Simulation of noisy real-time quantum dynamics}. We simulate the expectation value of a Pauli-Z observable in the middle of an $11\times11$ arrangement of 121 qubits evolving under a rotated transverse-field Ising model at a quantum critical point. In addition to the path-weight truncations 20, 25 and 30, we truncate small coefficients below $2^{-23}$ and propagating Pauli operators with more than 5 X or Y Paulis. This setup was used in Ref.~\cite{beguvsic2024real} to accurately simulate the noise-free case with significantly more computational resources.}
    \label{fig:dynamics}
\end{figure}

\section{Methods}\label{sec:methods}

In this section, we outline the primary methods and proof strategies of this work. \\

\noindent\textbf{Notation.} In order to illustrate our methods and proof ideas, we first introduce some essential notation.  Let $\calP_n \coloneqq \{I,X,Y,Z\}^{\otimes n}$ be the $n$-qubit Pauli basis.
Given a Pauli operator $P = P^{(1)}\otimes P^{(2)} \otimes \dots \otimes P^{(n)} \in \calP_n$, we define its support, denoted as $\mathrm{supp}(P)$,  as the subset of qubits upon which $P$ acts non-trivially, and we define its Pauli weight, denoted as $ \abs{P} $ as the cardinality of its support or, equivalently, as the number of non-identity single-qubit Pauli operators occurring in the tensor decomposition of $P$.
Given an observable $O = \sum_{P\in\calP_n} a_P P $, we define its squared normalized Frobenius norm as
\begin{align}
    \norm{O}^2_{\mathrm{F}}\coloneqq {\sum_{P\in\calP_n} a_P^2}.
\end{align}

\smallskip

\noindent\textbf{Pauli-path analysis.} Our algorithm and guarantees draw heavily on a Pauli-path based analysis. We start by rewriting the noisy circuit model in Eq.~\eqref{eq:noisy-circuit-main} as follows
\begin{align}
    \calC\coloneqq &\underbrace{\calV^{\mathrm{single}} \circ \calN^{\otimes n} \circ \calU_{L}}_{\calC_L} \circ \underbrace{\calN^{\otimes n} \circ \calU_{L-1}}_{\calC_{L-1}} \circ \dots \circ \underbrace{\calN^{\otimes n} \circ \calU_{1}}_{\calC_1}
     = \calC_L \circ \calC_{L-1} \circ \dots \circ \calC_1.
\end{align}
In what follows we trade generality for conciseness by assuming that the circuit layer distributions are invariant under random single-qubit Clifford gates. We emphasize that our results extend to more general circuit ensembles, which require a more detailed and nuanced analysis, as elaborated in the Appendices.

Our technique consists in computing an approximate Heisenberg evolution of the observable via a truncated Pauli-path summation, analogous to similar approaches exploited for simulating noisy and noiseless circuits in the previous literature.
In particular, we truncate Pauli paths according to their path weight (also referred as Hamming weight in the prior literature), thus our estimator coincides with that employed in Refs.~\cite{aharonov2022polynomial, schuster2024polynomial, gonzalez2024pauli}.

As the Pauli operators forms a basis for the Hermitian operators, we can express the Heisenberg-evolved observable $\calC^\dag(O)$ as a sum across Pauli paths as follows
\begin{align}
    \calC^\dag(O)  = &\sum_{P_0, P_1,\dots, P_L \in \calP_n} \llangle O| P_{L} \rrangle    \llangle P_L | \widehat{\calC}_L | P_{L-1} \rrangle \llangle P_{L-1} | \widehat{\calC}_{L-1} | P_{L-2} \rrangle \dots  \llangle P_1 | \widehat{\calC}_{1} | P_{0} \rrangle P_0,
    \label{eq:heis-pp}
\end{align}
where we employed the normalized vectorized notation $\llangle O| P_{L} \rrangle  \coloneqq \frac{1}{2^n}{\Tr[OP_L]}$, $\llangle P_j | \widehat{\calC}_j | P_{j-1} \rrangle \coloneqq\frac{1}{2^n}\Tr[P_j \calC_j(P_{j-1})]$.
We can also rearrange Eq.~\eqref{eq:heis-pp} as follows by introducing a more compact notation:
\begin{align}
     \calC^\dag(O) \coloneqq \sum_{\gamma = (P_0, P_1,\dots, P_L) \in \calP_n^{L+1}} \llangle O| P_{L} \rrangle  \, \Phi_\gamma(\calC) P_0,
\end{align}
where $\gamma = (P_0, P_1,\dots, P_L) $ denotes a Pauli path and $\Phi_\gamma(\calC) \coloneqq   \llangle P_L | \widehat{\calC}_L | P_{L-1} \rrangle  \llangle P_{L-1} | \widehat{\calC}_{L-1} | P_{L-2} \rrangle \dots \llangle P_1 | \widehat{\calC}_{1} | P_{0} \rrangle$
denotes the associated Fourier coefficient.

As the above sum involves exponentially many Pauli paths, in order to keep a manageable runtime we consider only those with a small path weight $\abs{\gamma}$, which we define as follows
\begin{align}
    \abs{\gamma} \coloneqq \sum_{j=1}^L \abs{P_j}.
\end{align}
 Thus, given a truncation order $k > 0$, we consider the following weight-truncated estimator:
\begin{align}\label{eq:truncpath}
    \calC_k^\dag(O) \coloneqq \sum_{\substack{\gamma = (P_0, P_1,\dots, P_L) \in \calP_n^{L+1} \\ \abs{\gamma}<k}} \llangle O| P_{L} \rrangle  \, \Phi_\gamma(\calC) P_0.
\end{align}
Thus, we approximate the expectation value $ \Tr[\calC^\dag(O)\rho]$ by computing the overlap between $\calC_k^\dag(O)$ and $\rho$, i.e., $\Tr[\calC_k^\dag(O)\rho]$.

\smallskip

\noindent{\textbf{The Pauli propagation algorithm.}}
The approximate Heisenberg evolved observable $\calC_k^\dag(O)$ can be computed with an iterative process.
Assume without loss of generality that $O$ consists in a single Pauli term with weight less than $k$. If $O$ is a linear combination of multiple Pauli terms, we will compute the contribution of each term separately. Moreover, Pauli terms with weight at least $k$ can be discarded, as they will form Pauli path with weight at least $k$.  
We start by computing $\calC_L^\dag(O)$ in the Pauli basis:
\begin{itemize}
    \item If $\calC_L$ contains only Clifford unitaries and Pauli noise channels, then $\calC_L^\dag(O)$ is also a single Pauli operator.
    \item In general, $\calC_L^\dag(O)$ can be expressed as a linear combination of Pauli operators. In this case, we say that the layer $\calC_L$ \emph{splits} $O$ in different Pauli paths. This is the case if $\calC_L$ contains non-Clifford unitaries or non-unital noise channels. 
\end{itemize}
For each Pauli path  obtained as described above, we can then apply the adjoint channel $\calC_{L-1}^\dag(\cdot)$ and proceed backwards until $\calC_1^\dag(\cdot)$. Thus, after each iteration the Pauli paths may split again, forming a tree-like structure as shown in Fig.~\ref{fig:schematic}c. 

Computing all the Pauli paths would require a runtime exponential in the number of non-Clifford unitaries and non-unital noise channels. Therefore, we combine such iterative method with a carefully chosen \emph{truncation rule}, which consists in discarding all the Pauli paths with path weight larger than a cutoff $k$, as in Eq.~\eqref{eq:truncpath}, ensuring that the number of computed Pauli paths remains tractable. 
In particular, for an observable composed of polynomially many Pauli terms and $k$ scaling logarithmically in system size, the above procedure requires a polynomial runtime.
This truncation rule has also been employed in Ref.~\cite{aharonov2022polynomial, schuster2024polynomial, gonzalez2024pauli} and bears similarities with the strategies adopted in Refs.~\cite{fontana2023classical, shao2023simulating, rudolph2023classical}.

If we further assume that the noise is unital, then the approximate Heisenberg evolved observable $\calC_k^\dag(O)$ can also be computed with an alternative strategy proposed in Ref.\ \cite{aharonov2022polynomial}, which yields a polynomial runtime provided that $k$ scales logarithmically in system size.

\smallskip

\noindent\textbf{Intuitive explanation of theoretical guarantees.} The widespread success of Pauli-path methods under depolarizing noise can be attributed to a key feature: depolarizing noise has a remarkably simple representation in the Pauli basis. For a Pauli operator \( P \in \mathcal{P}_n \), its evolution under the local depolarizing channel \( \mathcal{N}_p^{\mathrm{(depo)}\otimes n} \) with noise rate \( p \) is given by:  
\begin{align}
     \calN_p^{\mathrm{(depo)}\otimes n}(P) = \calN_p^{\mathrm{\dag(depo)}\otimes n}(P) = (1-p)^{|P|} P,
\end{align}
where the first identity follows from the fact that the depolarizing channel and its adjoint are equal.
Therefore, the Frobenius norm of $P$ is damped by a factor $(1-p)^\abs{P}$, as we have that $\norm{\mathcal{N}_p^{\mathrm{(depo)}\dagger \otimes n}(P)}_{\mathrm{F}} = (1-p)^{\abs{P}}$.
Moreover, the Frobenius norm of an observable $O = \sum_{P\in\calP_n} a_P P$ takes the following form after the application of the depolarizing channel
\begin{align}
    \norm{\calN_p^{\dagger(\mathrm{depo})\otimes n}(O)}_\mathrm{F}^2 = \sum_{P\in\calP_n} a_P^2 (1-p)^{2\abs{P}} \label{eq:dep-frob}.
\end{align}

Non-unital noise channels, in stark contrast with the behavior of depolarizing noise, can increase the Frobenius norm of observables under Heisenberg evolution. However, as we prove in Lemma\ \ref{lem:norm-contr-scr}, when averaged over random gates, any incoherent noise channel contracts the Frobenius norm of Pauli operators under Heisenberg evolution, as we have
\begin{align}
    \bbE_{V\sim\calD} \norm{\calN^{\dag\otimes n}(V^\dag P V)}^2_\mathrm{F} \leq \left(\chi^2_\calD(\calN)\right)^{\abs{P}}
\end{align}
where $\chi^2_\calD(\calN)\coloneqq\frac{\norm{\bold{D}}_2^2  + \norm{\bold{t}}_2^2}{3}$ is the mean squared contraction coefficient of $\calN$ with respect to the distribution $\calD$. As this coefficient is strictly smaller than one, the Frobenius norm of Pauli operators is damped exponentially with respect to the Pauli weight $\abs{P}$, on average over the random unitaries.
Similarly, given an observable $O = \sum_{P\in\calP_n} a_P P$, we have that
\begin{align}
    \bbE_{V\sim\calD} \norm{\calN^{\dag\otimes n}(V^\dag O V)}^2_\mathrm{F} \leq \sum_{P\in\calP_n} a^2_P \left(\chi^2_\calD(\calN)\right)^{\abs{P}},
\end{align}
which also matches the effect of local depolarizing noise of Eq.\ \eqref{eq:dep-frob}.

\smallskip

Given the striking resemblance in Heisenberg picture between the action of depolarizing noise and the \emph{average} action of arbitrary noise, it is convenient to rewrite the channel $\calN$ as follows:
\begin{align}
    \calN \coloneqq \calN^{\mathrm{(depo)}}_{p} \circ \tilde\calN_p,
\end{align}
where $\tilde\calN_p$ is a suitable (non-physical) linear map and $p$ is the \emph{effective depolarizing rate} of the channel $\calN$:
\begin{align}
   p \coloneqq 1- \chi_\calD(\calN). 
\end{align}
This strategy is inspired by Ref.\ \cite{schuster2024polynomial}, which proposed a similar decomposition for noisy circuits affected by spontaneous emissions occurring in random directions.

Having established this correspondence between depolarizing noise and arbitrary noise averaged over random gates, we can conduct a Pauli-path analysis akin to that introduced in Ref.\ \cite{aharonov2022polynomial}. As for the depolarizing case, we find that the contribution of different Pauli-paths is exponentially suppressed with respect to their path-weight, which allows us to show that
\begin{align}
    \bbE \,
    \abs{\Tr[(\calC_k^\dag(O) - \calC^\dag(O))\rho]} \in \exp(-\Omega\left(k\right)).
\end{align}
Thus, the average error over the random choice of the circuit can be exponentially suppressed by increasing the truncation order $k$.
Moreover, we also demonstrate that the truncated observable $\calC_k^\dag(O)$ can be computed in $\poly(n)$-time provided that $k \in \calO(\log(n))$ and the observable $O$ is a linear combination of at most $\poly(n)$ Pauli operators.

\section{Open problems}

In this paper we have conducted a systematic study of random quantum circuits interspersed by arbitrary incoherent noise. 
Taken together, our results provide a robust evidence that expectation values of typical noisy circuits can be efficiently estimated classically with Pauli-path methods. 
Nevertheless, several questions remain unaddressed by our results. Here, we outline a series of open problems for future investigation.

\begin{itemize}
    \item \textbf{Estimating ``dense'' Hamiltonians under non-unital noise.}
    Our main result, Theorem\ \ref{thm:non-u-inf}, guarantees that the noisy expectation values of an observable $H$ can be estimated in polynomial time with precision $\poly(n^{-1}) \norm{H}_{\mathrm{F}}$ on most circuits, provided that $H$ contains \emph{polynomially many} Pauli terms. However, this constraint excludes an interesting class of observables representing dense Hamiltonians with long-range interactions.
    For instance, consider the observable
    \begin{align}
        H = \sum_{ \substack{P \in \{I,Z\}^{\otimes n} \\ \abs{P}\leq \ell}} P.
    \end{align}
    If $\ell \in \omega(1)$, then $H$ contains $n^{\omega(1)}$ Pauli terms, which makes the runtime of our algorithm super-polynomial.
    \item \textbf{Classically sampling under arbitrary noise.}
    In their seminal work, the authors of Ref.\ \cite{aharonov2022polynomial} provided a polynomial-time classical algorithm for sampling from the distribution of random circuits under depolarizing noise. 
    It is natural to consider whether this result can be extended to circuits subject to arbitrary noise. Notably, non-unital noise introduces several challenges. For instance, the algorithms proposed in Ref.\ \cite{aharonov2022polynomial} estimate the expectation values of projectors using a truncated path summation, with a truncation order of $\Theta(\log(n))$. This approach creates a scenario analogous to that encountered with dense Hamiltonians, as described earlier. Consequently, the algorithms discussed in this paper would require quasi-polynomial time for such computations.
    On the other hand, the proof technique of Ref.\ \cite{aharonov2022polynomial} requires the ensemble of noiseless quantum circuit to exhibit anti-concentration. This poses an additional challenge, as random circuits with non-unital noise do not exhibit anti-concentration~\cite{fefferman2023effect}.
    
    \item \textbf{Circuits with correlated parameters.} While our results apply to a broad class of random quantum circuits, we have assumed that all circuit layers are sampled independently. This assumption, though common in the literature on random quantum circuits, does not account for scenarios where different layers are correlated. Such correlations arise, for example, in Trotterized circuits for simulating Hamiltonian dynamics and in circuits used for variational quantum algorithms~\cite{cerezo2020variationalreview}. Although variational circuits are often initialized with random parameters, these parameters are updated during training, breaking the independence assumption. Investigating circuits with correlated parameters presents significant theoretical challenges. However, numerical studies such as the one in Fig.~\ref{fig:dynamics}  could offer valuable insights into their classical simulability.  

    \item \textbf{Alternative noise models.} Despite its generality, our noise models also comes with its limitations.
    Namely, we are restricted to \textit{local} incoherent noise models that are independent of the particular gate applied. However, realistic quantum hardware is also subject to coherent errors and errors with long range correlations (such as those arising from cross talk~\cite{zhou2023quantum, tuziemski2023efficient}). The classical simulability of such noisy circuits is an important direction for future research. 
\end{itemize}

\section{Code Availability}
The numerical simulations in this work were performed with the open-source \href{https://github.com/MSRudolph/PauliPropagation.jl}{PauliPropagation.jl} package. The extensions necessary to employ path-weight truncation will soon be publicly available, but can be shared upon request.

\section{Acknowledgments}
The authors thank Lennart Bittel, Sumeet Khatri, Victor Martinez, Thomas Schuster, Daniel Stilck França and Chu Zhao for valuable discussions and feedback.
AA and ZH acknowledge support from the Sandoz Family Foundation-Monique de Meuron program for Academic Promotion.
AAM acknowledges support by the German Federal Ministry for Education and Research (BMBF) under the project FermiQP. AAM was also supported by the U.S. DOE through a quantum computing program sponsored by the Los Alamos National Laboratory (LANL) Information Science \& Technology Institute and by the Laboratory Directed Research and Development (LDRD) program of LANL under project number 20230049DR. MSR acknowledges funding from the 2024 Google PhD Fellowship. MC acknowledges support by LANL ASC Beyond Moore’s Law project and by the U.S. Department of Energy, Office of Science, Office of Advanced Scientific Computing Research through the Accelerated Research in Quantum Computing Program MACH-Q project. MR and ZH acknowledge support of the NCCR MARVEL, a National Centre of Competence in Research, funded by the Swiss National Science Foundation (grant number 205602).

\bibliography{quantum}

 \clearpage
 \newpage
 
\renewcommand\partname{} 
\appendix
{\huge Appendices}
\part{}
\parttoc 


\section{Preliminaries}

\subsection{Notation}
We briefly introduce the notations employed in this paper. Given a positive integer $n$, we denote the set $[n] \coloneqq \{1,2,\dots, n\}$. Given a distribution $\calD$ over the set $\mathcal{X}$ and a measurable function $F: \mathcal{X} \rightarrow \mathbb{R}$, we denote the expectation value of $F(X)$ for $X$ randomly sampled from $\calD$ as $\bbE_{X\sim\calD} [F(X)]$.
Given two distributions $\calD_1, \calD_2$, we denote by $\calD_1 \otimes \calD_2$ the derived distribution obtained by sampling independently $X_1$ from $\calD_1$ and $X_2$ from $\calD_2$, and outputting $X_1 \otimes X_2$.    

\smallskip

\noindent\textbf{Linear operators.} Let \(\mathcal{H}_n\) denote the Hilbert space for \(n\)-qubits, and let \(\mathcal{B}(\mathcal{H}_n)\) represent the space of linear operators acting on \(\mathcal{H}_n\). We introduce the \emph{vectorization} map \(\mathrm{vec} : \mathcal{B}(\mathcal{H}_n) \to \mathcal{H}_n^{\otimes 2}\), which is defined as follows: for any \(i,j \in 
\{0,1\}^n\), we set \(\mathrm{vec}(\ketbra{i}{j}) \coloneqq \ket{i} \otimes \ket{j}\). 
By linearity, if \(A \in \mathcal{B}(\mathcal{H}_n)\) is expressed as \(A = \sum_{i,j\in\{0,1\}^n} A_{i,j} \ketbra{i}{j}\), then the vectorization of \(A\) is given by \(\mathrm{vec}(A) = \sum_{i,j\in\{0,1\}^n} A_{i,j} \ket{i} \otimes \ket{j}\).
We also adopt notation \(\kket{A} \coloneqq \frac{1}{\sqrt{2^n}}\mathrm{vec}(A)\), which includes a normalization factor. Notably, the canonical inner product between two normalized vectorized operators \(A\) and \(B\) is simply their normalized Hilbert-Schmidt inner product: \(\bbrakket{A}{B} = \frac{1}{2^n}\Tr[A^\dagger B]\), where \(\bbrakket{A}{B} \coloneqq \frac{1}{2^n}\mathrm{vec}(A)^\dagger \mathrm{vec}(B)\).
The following \emph{ABC-rule} is easily verified:
\begin{align}
    \kket{ABC} = A \otimes C^T \kket{B},
    \label{eq:ABCtrick}
\end{align}
for all \(A, B, C \in \mathcal{B}(\mathcal{H}_n)\).

\smallskip

\noindent\textbf{Linear maps.} We denote by $\calL_{n\rightarrow n}$ the set of all linear maps of the form \(\Phi : \mathcal{B}(\mathcal{H}_n) \to \mathcal{B}(\mathcal{H}_n)\).
Quantum channels are linear maps which are also completely positive and trace-preserving. Given $\Phi \in \calL_{n\rightarrow n}$, it is important to note that the map \(\kket{\Phi(X)}\) is linear in \(\kket{X}\) for all \(\kket{X} \in \mathcal{H}_n^{\otimes 2}\). This linearity guarantees the existence of a matrix \(\widehat{\Phi} \in \mathcal{B}(\mathcal{H}_n^{\otimes 2})\), that represents this linear transformation. Hence, the action of \(\Phi\) on an operator \(X\) can be written as:
\begin{align}
\label{eq:phikraus}
\kket{\Phi(X)} = \widehat{\Phi} \kket{X} \quad \text{for all } X \in \mathcal{B}(\mathcal{H}_n).
\end{align}
Furthermore, for every linear map \(\Phi : \mathcal{B}(\mathcal{H}_n) \to \mathcal{B}(\mathcal{H}_n)\), there exists~\cite{watrous2018thetheory} a set of matrices \(\{A_i\}_{i=1}^{4^{n}}\), \(\{B_i\}_{i=1}^{4^{n}} \in \mathcal{B}(\mathcal{H}_n)\) such that \(\Phi(X) = \sum_{i=1}^{4^{n}} A_i X B_i^\dagger\).
Using the \emph{ABC-rule} \eqref{eq:ABCtrick}, we can rewrite $\Phi(X)=\sum^{4^{n}}_{i=1}A_i X B^\dagger_i$ as $\kket{\Phi(X)}=\sum^{4^{n}}_{i=1} A_i \otimes B^*_i \kket{X}$, which implies $\widehat{\Phi}=\sum^{4^{n}}_{i=1} A_i \otimes B^*_i$, that we denote as the matrix form of the linear map.

We will always use the \emph{hat} symbol over a linear map ($\widehat{\Phi}$), to indicate the matrix form of the linear map. It is also easy to observe that the matrix form of composition of two linear maps is equal to the product of the associated matrix form: Let $\Phi_1,\Phi_2$ be two linear maps, then we have $\widehat{\Phi_2\circ \Phi_1}= \widehat{\Phi}_2\cdot \widehat{\Phi}_1$.

For every linear map \(\Phi : \mathcal{B}(\mathcal{H}_n) \to \mathcal{B}(\mathcal{H}_n)\), written as \(\Phi(X) = \sum_{i=1}^{4^{n}} A_i X B_i^\dagger\), its adjoint linear map (or Heisenberg representation) can be defined as \(\Phi^{\dagger}(X) \coloneqq \sum_{i=1}^{4^{n}} A^{\dag}_i X B_i\). The matrix form of the adjoint linear map $\Phi^{\dag}$ is equal to the adjoint of the matrix form of the linear map, i.e., $\widehat{\Phi^{\dag}}=\widehat{\Phi}^\dag$ .

\smallskip

\noindent\textbf{Pauli basis.} We define the Pauli basis as $\calP_n \coloneqq  \{I, X, Y, Z\}^{\otimes n}$. Since the Pauli basis forms an orthonormal Hermitian basis with respect to the normalized Hilbert-Schmidt inner product, we have for all $P,Q \in \calP_n$,
\begin{align}
    \llangle P | Q \rrangle = \frac{1}{2^n} \Tr[PQ] = \delta_{PQ}.
\end{align}
Given a linear map $\calC : \mathcal{B}(\mathcal{H}_n) \rightarrow \mathcal{B}(\mathcal{H}_n)$ we denote the transition amplitude associated to two Pauli operators $P,Q\in\calP_n$ as follows
\begin{align}
    \llangle P | \widehat{\mathcal{C}} | Q  \rrangle = \llangle P |  \calC(Q) \rrangle = \llangle  \calC^\dag(P) | Q \rrangle. 
\end{align}
then, note that $\calC$ can be completely described by the the associated transition amplitudes. For any Pauli operator $P\in\calP_n$ we can write
\begin{align}
    \calC(P) = \sum_{Q\in\calP_n}  \llangle Q | \widehat{\mathcal{C}} | P  \rrangle Q.
\end{align}
By linearity, for any operator $O\in\calB(\calH_n)$ we have
\begin{align}
    \calC(O) = \sum_{P,Q\in\calP_n} \llangle O | P \rrangle \llangle Q | \widehat{\mathcal{C}} | P  \rrangle Q.
\end{align}

Equivalently, the transition amplitudes $ \llangle Q | \widehat{\mathcal{C}} | P \rrangle$ can be expressed as a $4^n\times 4^n$ matrix, which is usually referred as the \emph{Pauli Transfer Matrix} (PTM) of the channel.

\smallskip

\noindent\textbf{Support and weight of Pauli operators.} Given a Pauli operator $P = P^{(1)}\otimes P^{(2)} \otimes \dots \otimes P^{(n)} \in \calP_n$, we define its support, denoted as $\mathrm{supp}(P)$,  as the subset of qubits upon which $P$ acts non-trivially, i.e.
\begin{align}
    \mathrm{supp}(P) \coloneqq \left\{ i \, | \, P^{(i)} \neq I \right\},
\end{align} and we use \(|P|\) to denote the Pauli weight of \(P\), which is defined as the number of non-identity single-qubit Pauli operators in the Pauli tensor decomposition of \(P\), i.e., $\abs{P} = \abs{\mathrm{supp}(P)}$.

\smallskip

\noindent\textbf{Matrix norms.} For a matrix \( A\in\mathcal{B}(\mathcal{H}_n) \), the \emph{Schatten \( p \)-norm} is defined as \( \|A\|_p \coloneqq \big(\Tr\big[(\sqrt{A^\dagger A})^p\big]\big)^{1/p} \), representing the \( p \)-norm of the matrix's singular values. Closely related to the Schatten \( p \)-norm are the trace norm (\(\|\cdot\|_1 \)) and the Hilbert-Schmidt norm (\( \|\cdot\|_2 \)).  Then, we recall that the infinity norm (\( \|\cdot\|_\infty \)) of a matrix corresponds to its largest singular value and can also be obtained as the limit of the Schatten \( p \)-norm when \( p \to \infty \).

We remind that an operator $O \in\mathcal{B}(\mathcal{H}_n)$ is Hermitian if and only if it can be expressed as a linear combination of Pauli operators with real coefficients, i.e.,
\begin{align}
    O = \sum_{P \in \calP_n} a_P P  \;\;\;\text{ where } a_P \in \mathbb{R} \text{ for all } P\in\calP_n.
\end{align}
Given a Hermitian operator $O$, we define its squared \emph{normalized} Frobenius norm $\norm{O}_\mathrm{F}^2 $ as
\begin{align}
    \norm{O}_\mathrm{F}^2 \coloneqq \frac{\norm{O}_2^2}{2^n}  = {\frac{1}{4^n} \sum_{P\in\calP_n} \Tr[OP]^2} = {\sum_{P\in\calP_n}a_P^2},
\end{align}
which is also referred as normalized Hilbert-Schmidt norm or squared Pauli 2-norm.
We will also use the following useful relation:
\begin{align}
    \norm{O}^2_2 = \Tr[\mathbb{F}O^{\otimes 2}]
\end{align}
where $\mathbb{F}$ is the \emph{flip operator}, also known as SWAP operator, defined as
\begin{align}
     \mathbb{F} \coloneqq \sum_{i,j \in \{0,1\}^n} \ketbra{i,j}{j,i}.
\end{align}

\subsection{Ensembles of linear maps and states}
\label{sec:dist-maps}
In this section, we define several ensembles of linear maps that will serve as essential tools throughout this work. These ensembles enable us to analyze the limitations of noisy random circuits without relying on the global~\cite{quek2022exponentially}
or local~\cite{mele2024noise} unitary 2-design assumption commonly adopted in previous works on noisy random circuits.

We start by recalling the definition of the Haar measure, which formalizes the notion of uniform distribution over unitaries. 
For a more detailed explanation we refer to
Ref.~\cite{mele2023introduction} .
\begin{definition}[Haar measure and $t$-designs]
The \emph{Haar measure} $\mu_H$ is the (unique) probability distribution over the unitary group $\mathbb{U}(2^n)$ which is left and right-invariant, which means that for any integrable function $f$, we have
\begin{align}
    \bbE_{U\sim\mu_H} \left[f(U)\right]= \bbE_{U\sim\mu_H} \left[f(UV)\right]= \bbE_{U\sim\mu_H} \left[f(VU)\right],
\end{align}
for any $U,V \in \mathbb{U}(2^n)$.
Furthermore, a probability distribution $\calD$ over the unitary group $\mathbb{U}(d)$ is a unitary $t$-design if it matches the moments of the Haar measure up ot the $t$-th order, that is 
\begin{align}
    \bbE_{U\sim \calD} \left[U^{\otimes t} O U^{\dag \otimes t}\right] =  \bbE_{U\sim \mu_H} \left[U^{\otimes t} O U^{\dag \otimes t}\right],
\end{align}
for all linear operators $O\in\mathcal{B}(\calH_n)$. 
\end{definition}
Equivalently, in the vectorized notation, the previous definition reads as follows
\begin{align}
    \bbE_{U\sim \calD} \left[U^{\otimes t} \otimes U^{\emph*\otimes t}\right] =  \bbE_{U\sim \mu_H} \left[U^{\otimes t} \otimes U^{\emph*\otimes t}\right] .
\end{align}

We further provide some examples of unitary 1 and 2-designs.

\begin{example}[Unitary 1 and 2-designs]
\label{ex:1-2-des}
A distribution  $\calD$ over $\mathbb{U}(2^n)$ is a unitary 1-design if, for all linear operators $O \in \calB(\calH_n)$, it satisfies
\begin{align}
    \bbE_{U\sim \calD} \left[U O U^\emph\dag\right]=\bbE_{U\sim \mu_H} \left[U O U^\emph\dag\right] = \frac{I^{\otimes n}}{2^n}\Tr[O].
\end{align}
An example of unitary 1-design is the uniform distribution over the Pauli basis $\calP_n$, as we have
\begin{align}
    \frac{1}{4^n}\sum_{P\in\calP_n} [POP] = \frac{I^{\otimes n}}{2^n}\Tr[O].
\end{align}
The above identity is also commonly referred as \emph{Pauli twirling.}
It is also noteworthy that the Haar measure over the orthogonal group $\mathbb{O}(2^n)$, which we denote as $\mu_H(\mathbb{O}(2^n))$, also forms a unitary 1-design.

\medskip

A distribution  $\calD$ over $\mathbb{U}(2^n)$ is a unitary 2-design if, for all linear operators $O \in \calB(\calH_n^{\otimes 2})$, it satisfies
\begin{align}
    \bbE_{U\sim \calD} \left[U^{\otimes 2} O U^{\emph\dag\otimes 2}\right]=\bbE_{U\sim \mu_H} \left[U^{\otimes 2} O U^{\emph\dag\otimes 2}\right] = c_{\mathbb{I},O}\mathbb{I} + c_{\mathbb{F},O} \mathbb{F}.
\end{align}
    where
    \begin{align}
        c_{\mathbb{I},O}=\frac{\Tr\!\left(O\right)-2^{-n}\Tr\!\left(\mathbb{F}O\right)}{4^n-1}\quad \text{and} \quad c_{\mathbb{F},O}=\frac{\Tr\!\left(\mathbb{F}O\right)-2^{-n}\Tr\!\left(O\right)}{4^n-1}.
    \end{align}
An example of unitary 2-design is the uniform distribution over the Clifford group $\mathrm{Cl}(2^n)$.    
\end{example}

In the following, we define the class of \emph{locally unbiased} distributions over linear maps. In loose terms, a distribution belongs to this class if it remains invariant under single-qubit unitary 1-designs. We refer to these distributions as ``locally unbiased'' because the expectation value of a quantum state evolved under a unitary 1-design is the maximally mixed state. This state lies at the center of the Bloch sphere, thus it is not biased toward any specific direction.

\begin{definition}[Locally unbiased distribution]
\label{def:ls-supops}

A distribution $\calD$ over linear maps $\calL_{n\rightarrow n}$ is locally unbiased if it is invariant under right-multiplication of single-qubit unitary 1-designs up to the second moment, i.e., if there exist 
some distribution $\calD_1, \calD_2,\dots, \calD_n$ over $\mathbb{U}(2)$ such that (i) all $\calD_i$ are unitary 1-designs, and (ii) the following identity holds 
\begin{align}
    \bbE_{\calC\sim\calD} \left[\widehat{\calC^{\otimes 2}} 
    \right] = 
    \bbE_{\calC\sim\calD} \mathbb{E}_{V\sim \bigotimes_{i=1}^n \calD_i}
    \left[\widehat{\calC^{\otimes 2} }   \left( V^{\otimes 2} \otimes  V^{*\otimes 2}\right) \right].
\end{align}
\end{definition}

Locally unbiased distributions generalize the following notion of \emph{Pauli-invariance} which was previously proposed in Ref.~\cite{aharonov2022polynomial} (in the sense that Pauli-invariance implies local unbiasedness).

\begin{definition}[Pauli-invariant distribution]
A distribution $\calD$ over linear maps $\calL_{n\rightarrow n}$ is Pauli-invariant if it is invariant under right-multiplication of random Pauli operators up to the second moment, i.e., if 
\begin{align}
    \bbE_{\calC\sim\calD} \left[\widehat{\calC^{\otimes 2}}\right] = 
    \frac{1}{4^n}\sum_{P\in\calP_n} \bbE_{\calC\sim\calD}  \left[\widehat{\calC^{\otimes 2}} \left(P^{\otimes 2} \otimes P^{\emph*\otimes 2}\right)\right]\,.
\end{align}
\end{definition}

The following technical lemma plays a central role in the analysis of noisy random circuits. 

\begin{lemma}[Orthogonality]
\label{lem:ortho1}
Let $\calD$ be a distribution over linear maps $\calL_{n\rightarrow n}$. 
\begin{enumerate}
    \item If $\calD$ is locally unbiased, then for all $P,Q\in\calP_n$ such that $\mathrm{supp}(P)\neq \mathrm{supp}(Q)$, we have
    \begin{align}
        \bbE_{\calC \sim \calD} \left[{\calC}^{\otimes 2} (P\otimes Q) \right]=0.
    \end{align}
    In particular, for all observables $O$, we have
    \begin{align}
        \bbE_{\calC \sim \calD} \, \calC^{\otimes 2}(O^{\otimes 2}) = \sum_{A \subseteq [n]}  \bbE_{\calC \sim \calD} \left(\sum_{\substack{P\in\calP_n \\ \mathrm{supp}(P) = A}} \llangle O | P \rrangle \, \calC(P) \right)^{\otimes 2}. \label{eq:ortho-from-ls}
    \end{align}
    \item $\calD$ is Pauli-invariant if and only if, for all $P,Q\in\calP_n$ such that $P\neq Q$, we have
    \begin{align}
        \bbE_{\calC \sim \calD} \left[{\calC}^{\otimes 2} (P\otimes Q) \right]=0.
    \end{align}
    In particular, for all observables $O$, we have
    \begin{align}
        \bbE_{\calC \sim \calD} \, \calC^{\otimes 2}(O^{\otimes 2}) = \sum_{P \in \calP_n}  \llangle O | P \rrangle^2  \bbE_{\calC \sim \calD} \, \calC^{\otimes 2}(P^{\otimes 2}). \label{eq:ortho-from-pauli}
    \end{align}
\end{enumerate}
\end{lemma}
\begin{proof}
We start by proving the first statement. 
Recall that, by Definition~\ref{def:ls-supops}, we have 
\begin{align}
    \bbE_{\calC\sim\calD} \left[\widehat{\calC}^{\otimes 2} 
    \right] = 
    \bbE_{\calC\sim\calD} \mathbb{E}_{V \sim \bigotimes_{i=1}^n \calD_i}\left[\widehat{\calC}^{\otimes 2} \,  V^{\otimes 2} \otimes V^{*\otimes 2} 
    \right],
\end{align}
for some unitary 1-designs $\calD_1, \calD_2,\dots, \calD_n$.
Let $P, Q\in\calP_n$ be two Pauli operators which do not share the same support.
Consider their tensor decompositions $P=P_1\otimes P_2\otimes \dots \otimes P_n$ and $Q=Q_1\otimes Q_2\otimes \dots \otimes Q_n$.
As $\mathrm{supp}(P)\neq \mathrm{supp}(Q)$, there exists an index $j \in [n]$ such that $P_j \neq I$ and $Q_j = I$.
From the previous, we have that
\begin{align}
    \bbE_{V_j \sim \calD_j} \left[V_j^{\dag\otimes 2} (P_j \otimes Q_j) V_j^{\otimes 2}\right]
    = \bbE_{V_j \sim \calD_j}\left[ V_j^\dag P_j V_j\right] \otimes I = 0 \label{eq:1des-supp},
\end{align}
where we used the fact that $\calD_j$ is a unitary 1-design and unitaries are trace-preserving. Then, the desired result easily follows
\begin{align}
&\bbE_{\calC \sim \calD} \left[{\calC}^{\otimes 2} (P\otimes Q) \right]
   \\=&\bbE_{\calC \sim \calD}\left[\calC^{\otimes 2}\left(\mathbb{E}_{V \sim \bigotimes_{i=1}^n \calD_i}V^{\dag\otimes 2} (P\otimes Q)  V^{\otimes 2}\right)\right]
   \\=& \bbE_{\calC \sim \calD}\left[\calC^{\otimes 2} \left(\left(\bbE_{V_j \sim \calD_j} V_j^{\otimes 2} (P_j \otimes Q_j) V_j^{\dag\otimes 2}\right)
   \otimes \left(\mathbb{E}_{\substack{V_1 \sim \calD_1, \dots, V_{j-1} \sim \calD_{j-1},\\ V_{j+1} \sim \calD_{j+1}, \dots, V_n \sim \calD_n}}  \bigotimes_{i\neq j}\left( V_i^{\otimes 2} (P_i\otimes Q_i)  V_i^{\dag\otimes 2} \right)\right)\right)\right] = 0,
\end{align}    
where in the last step we plugged Eq.~\eqref{eq:1des-supp}.
Then Eq.~\eqref{eq:ortho-from-ls} readily follows
\begin{align}
    \bbE_{\calC\sim\calD} [\calC^{\otimes 2}(O^{\otimes 2})]  = & 
    \bbE_{\calC\sim\calD} \left(\calC\left( \sum_{P\in\calP_n} \llangle O | P \rrangle P\right)\right)^{\otimes 2}
     = \sum_{P,Q\in\calP_n} \llangle O | P \rrangle\llangle O | Q \rrangle \underbrace{\bbE_{\calC\sim\calD} [\calC^{\otimes 2} (P\otimes Q)]}_{\text{$=0$ if $\mathrm{supp}(P) \neq \mathrm{supp}(Q)$}}
    \\ = & \sum_{A \subseteq [n]} \left(\sum_{\substack{P\in\calP_n \\ \mathrm{supp}(P) = A}} \bbE_{\calC\sim\calD} \llangle O | P \rrangle \, \calC(P) \right)^{\otimes 2}. 
\end{align}

Regarding the second statement, Lemma 2 in Ref.~\cite{aharonov2022polynomial} established the result for the case where $\calD$ is a distribution over $\mathbb{U}(4)$. Here, we demonstrate that this argument can be naturally extended to linear maps $\calL_{n\rightarrow n}$. First, let us use the invariance under right-multiplication of random Pauli
\begin{align}
\mathbb{E}_{\calC\sim\mathcal{D}}\left[\calC^{\otimes 2}( P  \otimes Q )\right] = 
\frac{1}{4^n} \mathbb{E}_{ \calC\sim\mathcal{D}} 
\left[\calC^{\otimes 2}\left(\sum_{R\in \calP_n}R P R   \otimes   R  Q R \right) \right] .
\end{align}
From here we can see that it suffices to show that
\begin{align}
    \sum_{R\in\calP_n}R P R   \otimes   R  Q R =0\,.
\end{align}
Let $\langle P, Q \rangle :=1[P\text{ and }Q\text{ anticommute}]$, i.e., $\langle P, Q \rangle $ is the indicator function which equals $1$ if $\{P,Q\}=0$ and $0$ if $[P,Q]=0$. We have
\begin{align}
 \sum_{R\in\calP_n} RPR   \otimes   RQR  =  \sum_{R\in\calP_n}(-1)^{\langle P, R \rangle  + \langle Q, R \rangle }P\otimes Q = \sum_{R\in\calP_n}(-1)^{\langle PQ, R\rangle}P\otimes Q= 0,   
\end{align}
where the last line follows from the fact that $PQ$ is not identity, and therefore commutes with half Paulis and anticommutes with the other half.

We now prove the other direction, i.e., that $\mathbb{E}_{\calC\sim\mathcal{D}}\left[\calC^{\otimes 2}( P  \otimes Q )\right] = 0 $ implies the Pauli invariance property. 
Let $O$ be an arbitrary Hermitian operator. We have
\begin{align}
    \frac{1}{4^n}\sum_{R\in\calP_n}\bbE_{\calC\sim \calD} \, \calC^{\otimes 2}\left( R^{\otimes 2} O R^{\otimes 2}\right) =  &\frac{1}{4^n}\sum_{R\in\calP_n}\sum_{P\in\calP_{2n}} \llangle O | P\rrangle \underbrace{\bbE_{\calC\sim \calD}  \left[\calC^{\otimes 2}(R^{\otimes 2}P R^{\otimes 2})\right]}_{ =0 \text{ unless } P = Q^{\otimes 2} \text{ with } Q\in\calP_n}
    \\ = &\frac{1}{4^n}\sum_{Q, R\in\calP_n} \llangle O | Q^{\otimes 2}\rrangle {\bbE_{\calC\sim \calD}  \left[\calC^{\otimes 2}(Q^{\otimes 2})\right]} (-1)^{2\langle Q, R \rangle}
    \\ = &\bbE_{\calC\sim \calD} \, \calC^{\otimes 2}\left(  O \right) ,
\end{align}
where in the last step we used that $(-1)^{2\langle Q, R \rangle} = 1$. 
Equation~\eqref{eq:ortho-from-pauli} can also be proven analogously
\begin{align}
    \bbE_{\calC\sim\calD} [\calC^{\otimes 2}(O^{\otimes 2})]  = & 
    \bbE_{\calC\sim\calD} \left(\calC\left( \sum_{P\in\calP_n} \llangle O | P \rrangle P\right)\right)^{\otimes 2}
    \\ = &\sum_{P,Q\in\calP_n} \llangle O | P \rrangle\llangle O | Q \rrangle \underbrace{\bbE_{\calC\sim\calD} [\calC^{\otimes 2} (P\otimes Q)]}_{\text{$=0$ if $P \neq Q$}}
    \\ = & \sum_{P \in \calP_n }  \bbE_{\calC\sim\calD} \llangle O | P \rrangle^2 \, \calC^{\otimes 2}(P^{\otimes 2}). 
\end{align}
\end{proof}

In the following, we provide an example of a distribution over $\mathbb{U}(2)$ that constitutes a unitary 1-design but does not exhibit Pauli invariance.

\begin{example}[A class of unitary 1-designs which are not Pauli-invariant]
Given a parameter $\theta \in [0,2\pi)$, let $\calE$ be the uniform distribution over the following ensemble 
\begin{align}
    \left\{X,Y,\exp(-i\frac{\theta}{2}Z), \exp(-i\left(\frac{\theta + \pi}{2}\right)Z) \right\}\,.
\end{align}
For $\theta = 0, \pi$, this ensemble coincides with the single-qubit Pauli basis.
We can easily verify that $\calE$ forms a 1-design.
However, $\calE$ is not Pauli invariant, as we have
\begin{align}
    &\bbE_{U\sim\calE} \left[U^{\otimes 2} (X\otimes Y) U^{\emph\dag\otimes 2}\right]
    \\ = &-\frac{1}{2}\left(X\otimes Y\right) + \frac{1}{2}\left( \cos^2(\theta) X\otimes Y + \sin^2(\theta) Y\otimes X +\cos(\theta)\sin(\theta)(X\otimes X -  Y\otimes Y)\right),
\end{align}
which equals 0 if and only if $\theta = 0$ or $\pi$.
\end{example}

We also introduce the class of approximately locally scrambling unitary distributions, which will be particularly useful for analyzing noise within the dephasing class. 

\begin{definition}[Approximate local scrambling]
\label{def:mixing}
Let $\eta\in[0,1)$. A distribution $\mathcal{D}$ over $\mathbb{U}(2)$ is an $\eta$-approximate scrambler if it satisfies the following properties
\begin{align}
    \forall P,Q \in \{I,X,Y,Z\}, P\neq Q : \quad \bbE_{U \sim \calD} \left[U^{\emph{\dag}\otimes 2 } (P\otimes Q) U^{\otimes 2}\right]=0  & \quad \text{\emph{(orthogonality)}}\label{eq:apx-ls1}
    \\\max_{P,Q\in\{X,Y,Z\}} \quad \bbE_{U\sim\calD} \llangle P \,|\,  U\otimes U^*  \,|\, Q\rrangle^2 \leq \frac{1}{3}(1 + 2\eta)   &  \quad \text{\emph{(approximate Pauli-mixing)}}.\label{eq:apx-ls2}
\end{align}

A distribution $\mathcal{D}$ over $\mathbb{U}(2^n)$ is an $\eta$-approximate local scrambler if there exists a distribution $\calD'$ over $\mathbb{U}(2^n)$ and some distribution $\calD_1, \calD_2,\dots, \calD_n$ over $\mathbb{U}(2)$ such that (i) all $\calD_i$ are $\eta$-approximate scrambler, and (ii) the following identity holds 
\begin{align}
    \bbE_{U\sim\calD} \left[U^{\otimes 2} \otimes U^{\emph*\otimes 2}\right] = 
    \bbE_{U\sim\calD'} \mathbb{E}_{V_1 \sim \calD_1, V_2 \sim \calD_2, \dots, V_n \sim \calD_n}\left[U^{\otimes 2} \bigotimes_{i=1}^n V_i^{\otimes 2} \otimes U^{\emph*\otimes 2}  \bigotimes_{i=1}^n V_i^{\emph*\otimes 2} \right],
\end{align}
where $V_1, V_2, \dots, V_n$ are sampled independently.
\end{definition}
This notion generalizes that of locally scrambling distributions~\cite{caro2022outofdistribution, huang2023learning}, which have found fruitful application in previous literature on the classical simulation of random noiseless circuits~\cite{angrisani2024classically}.
In particular, setting $\eta=0$ and $\calD = \calD'$, we recover the definition of locally scrambling distribution (up to the second moment) proposed in the previous literature~\cite{kuo2020markovian, hu2021classical, caro2022outofdistribution, huang2022learning}. 

\smallskip

We further note that the notion of approximate local scrambling is closely related to that of max-relative entropy, a connection we will leverage in our subsequent technical analysis.
{
\begin{observation}[Entropic interpretation of approximate scrambling]
\label{rem:ent}
Let $\calD$ be a single-qubit $\eta$-approximate scrambler and let $P\in\{X,Y,Z\}$. By Eq.~\eqref{eq:apx-ls1}, we have
\begin{align}
    \bbE_{U\sim\calD} \left[U^{\emph{\dag}\otimes 2} P U^{\otimes 2} \right] = \sum_{Q\in\{X,Y,Z\}}  a_P^2 P^{\otimes 2}.
\end{align}
Moreover, the vector $\boldsymbol{a}\in \mathbb{R}^3$ satisfies
\begin{align}
    \norm{\boldsymbol{a}}_2^2 = 1
    \quad \text{and} \quad \norm{\boldsymbol{a}}_\infty^2 \leq \frac{1+2\eta}{3},
\end{align}
where the identity follows from the unitarily invariance of the Frobenius norm and the inequality follows from Eq.~\eqref{eq:apx-ls2}.
Therefore, the vector $p_{\boldsymbol{a}}=(a_X^2, a_Y^2, a_Z^2)$ can be interpreted as a probability distribution with min-entropy bounded from below as
\begin{align}
    H_{\min} (p_{\boldsymbol{a}}) \coloneqq -\log \left(\max_{P\in\{X,Y,Z\}} a_P^2\right) \geq \log(3) - \log({1 + 2\eta}).
\end{align}
Equivalently, the min-entropy can be expressed in terms of the max-relative entropy between $p_{\boldsymbol{a}}$ and the uniform distribution $p_{\boldsymbol{u}}\coloneqq (1/3,1/3,1/3)$
\begin{align}
    D_{\max}\left( p_{\boldsymbol{a}} \| p_{\boldsymbol{u}} \right) \coloneqq -  H_{\min} (p_{\boldsymbol{a}}) + \log(3) \leq \log\left({1+2\eta}\right).
\end{align}
\end{observation}

Hinging on connection between approximate scrambling and the max-relative entropy, we can demonstrate that the approximate scrambling property is preserved under unitary evolution. Specifically, if \(\mathcal{D}\) is a single-qubit \(\eta\)-approximate scrambler, then the derived distribution obtained by sampling a unitary operator \(U\) from \(\mathcal{D}\) and applying the transformation \(UV\), is also a single-qubit \(\eta\)-approximate scrambler, as formalized in the following Lemma.

\begin{lemma}[Monotonicity of approximate scrambling]
\label{le:monotone}
Let $\calD$ be a single-qubit $\eta$-approximate scrambler. Then for all $V\in \mathbb{U}(2)$ we have
\begin{align}
 \forall P,Q \in \calP_1, P\neq Q : \, &\bbE_{U \sim \calD} \left[(VU)^{\emph{\dag}\otimes 2 } (P\otimes Q) (UV)^{\otimes 2}\right]=0 \label{eq:pauli-inv-mon},
  \\  \max_{P,Q\in\{X,Y,Z\}} &\bbE_{U\sim\calD}  \llangle P \,|\,  UV \otimes U^*V^* \,|\, Q\rrangle^2 \leq \frac{1}{3}(1 + 2\eta). \label{eq:scr-inv}
\end{align}
\end{lemma}
\begin{proof}
Eq.~\eqref{eq:pauli-inv-mon} follows from the linearity of the expectation:
\begin{align}
    \bbE_{U \sim \calD} \left[(UV)^{{\dag}\otimes 2 } (P\otimes Q) (UV)^{\otimes 2}\right]
    = V^{{\dag}\otimes 2 } \bbE_{U \sim \calD} \left[ U^{{\dag}\otimes 2 } (P\otimes Q) U^{\otimes 2}\right]V^{\otimes 2},
\end{align}
which equals to zero if $P\neq Q$ by Eq.~\eqref{eq:apx-ls1}.
The second part of the Lemma can be proven exploiting the connection between approximate scrambling and the max-relative entropy described in Observation~\ref{rem:ent}.
We have
\begin{align}
    &\bbE_{U\sim\calD} \left[U^{{\dag} \otimes 2} P^{\otimes 2} U^{\otimes 2}\right] =     \sum_{Q\in \{X,Y,Z\}} a_{Q}^2 Q^{\otimes 2}  \quad \text{ where } \norm{\boldsymbol{a}}_2^2=1 
    \\&\bbE_{U\sim\calD}\left[ V^{{\dag} \otimes 2}U^{{\dag} \otimes 2} P^{\otimes 2} U^{\otimes 2} V^{\otimes 2} \right]= \sum_{Q,R\in \{X,Y,Z\}} a_{Q}^2 \llangle Q | V \otimes V^* | R \rrangle^2 R^{\otimes 2} 
\end{align}

We further observe that, for all $Q\in\{X,Y,Z\}$ we have 
\begin{align}
   \sum_{R\in\{X,Y,Z\}} 
    \llangle Q| V \otimes V^* | R \rrangle^2 = \norm{V^\dag QV}_\mathrm{F}^2 =\norm{Q}_\mathrm{F}^2 =  1,
\end{align}
where in the first step we used the fact that $ \llangle Q| V \otimes V^* | I \rrangle = 0$ as $V$ is trace-preserving , and in the second step we used the unitarily invariance of the Frobenius norm.
Thus, the squared amplitudes $\llangle Q| V | R \rrangle^2 $ sum up to 1 and can be interpreted as transition probabilities. Specifically, we introduce the stochastic channel $\calK : \{X,Y,Z\} \rightarrow \{X,Y,Z\}$ defined as follows
\begin{align}
   \Pr[\calK(Q) =R] = \llangle Q| V | R \rrangle^2.
\end{align}
We can represent the action of $\calK$ on the distribution $p_{\boldsymbol{a}}=(a_X^2,a_Y^2,a_Z^2)$ as
\begin{align}
    &\calK(p_{\boldsymbol{a}}) = p_{\boldsymbol{b}} \coloneqq  (b_{X}^2, b_{Y}^2, b_{Z}^2),
    \\& \text{where}  \;\;\; b_R^2 \coloneqq\sum_{Q\in \{X,Y,Z\}} a_{Q}^2 \llangle Q | V | R \rrangle^2.
\end{align}
We can see that the uniform distribution $p_{\boldsymbol{u}}=(1/3,1/3,1/3)$ is invariant under $\calK$, i.e.,
\begin{align}
    \calK(p_{\boldsymbol{u}})  = p_{\boldsymbol{u}} . 
\end{align}
This is a consequence of the fact that the operator $\frac{1}{3}(X^{\otimes 2} + Y^{\otimes 2} + Z^{\otimes 2})$ is invariant under the unitary evolution $V^{\dag\otimes 2}(\cdot)V^{\otimes 2}$.
Then, by the data-processing inequality, we have
\begin{align}
    D_{\max}\left(p_{\boldsymbol{b}}\| p_{\boldsymbol{u}} \right)=&D_{\max}(\calK(p_{\boldsymbol{a}}) \| \calK\left(p_{\boldsymbol{u}} \right))
    \\\leq  &D_{\max}\left(p_{\boldsymbol{a}}\| p_{\boldsymbol{u}} \right) \leq \log\left({1+2\eta}\right),
\end{align}
which implies Eq.~\eqref{eq:scr-inv}.
\end{proof}

It is easy to see that any 2-designs over $\mathbb{U}(2)$ is also a single-qubit $0$-approximate scrambler. Notably, this include the uniform distribution over the single-qubit Clifford group discussed in Example~\ref{ex:1-2-des}.
However, we will demonstrate below that there are other practically relevant distributions over $\mathbb{U}(2)$ that are $\eta$-approximate scramblers with $\eta < 1$, highlighting the broad applicability of our analysis.

\begin{example}[Pauli rotations along 2 orthogonal axes]
\label{ex:scr-pauli-rot}
For simplicity, we consider two rotations along the $X$ and $Y$ axes, but our conclusions hold for any choice of the orthogonal axes.
Assume that $\phi$ and $\theta$ are sampled independently from the uniform distribution over $[0,2\pi)$.
We can easily verify that this distribution is Pauli invariant, as we have for all $P,Q\in\calP_1$ such that $P\neq Q$,
\begin{align}
   \bbE_{\theta, \phi} (R_X(\phi) R_Z(\theta))^{\otimes 2}(P\otimes Q)(R_Z(\theta)^\emph\dag R_X(\phi)^\emph\dag)^{\otimes 2} =0.
\end{align}
Moreover, this distribution is also an $\eta$-approximate scrambler with $\eta=1/4$, as we have 
\begin{align}
    &\bbE_{\theta, \phi}(R_Z(\theta)^\emph\dag R_X(\phi)^\emph\dag XR_X(\phi) R_Z(\theta))^{\otimes 2} = \frac{1}{2}\left(X^{\otimes 2} +  Y^{\otimes 2}\right),
    \\&\bbE_{\theta, \phi} (R_Z(\theta)^\emph\dag R_X(\phi)^\emph\dag YR_X(\phi) R_Z(\theta))^{\otimes 2} = \frac{1}{4} \left(X^{\otimes 2}+  Y^{\otimes 2} \right) + \frac{1}{2} Z^{\otimes 2},
    \\& \bbE_{\theta, \phi} (R_Z(\theta)^\emph\dag R_X(\phi)^\emph\dag Z R_X(\phi) R_Z(\theta))^{\otimes 2} = \frac{1}{4} \left(X^{\otimes 2} +  Y^{\otimes 2} \right)+ \frac{1}{2} Z^{\otimes 2}.
\end{align}
In all the three cases above, the largest coefficient is $1/2$, which implies that the considered distribution is a (1/4)-approximate scrambler.
\end{example}

On the other hand, we also note that uniform Pauli rotations along a single axis are neither $\eta$-approximate scramblers for $\eta < 1$, nor do they form a unitary 1-design.

\begin{example}[Rotations along a single axis]
We can further consider the case of a uniformly random Pauli rotation along a single axis (e.g., $R_Z(\theta)$). This distribution fails to be a 1-design or approximately locally scrambling, as we have $R_Z(\theta) ZR_Z(\theta)^\emph\dag = Z$ for all $\theta$. In particular, this implies
\begin{align}
&\bbE_{\theta}  R_Z(\theta) ZR_Z(\theta)^\emph\dag = Z,
 \\  & \bbE_\theta (R_Z(\theta)^\emph\dag Z R_Z(\theta))^{\otimes 2} = Z^{\otimes 2}.
\end{align}
\end{example}

We conclude this section by defining a class of distributions over quantum states, previously introduced in Ref.~\cite{schuster2024polynomial}.
This definition generalizes that of state 1-design.
\begin{definition}[Low-average ensemble, \cite{schuster2024polynomial}]
A distribution $\calD$ over quantum states is a low-average distribution with purity $c$ if it satisfies
\begin{align}
   \norm{\bbE_{\rho\sim\calD} [\rho]}_\infty \leq \frac{c}{2^n}.
\end{align}
\end{definition}
In particular, if $\calD$ is a distribution over pure states and $c=1$, then we recover the definition of state 1-design.
More generally, if $\calD$ is unitary 1-design, then we have for any state $\rho$ 
\begin{align}
    \norm{\bbE_{U\sim\calD} [U\rho U^\dag]}_\infty = \norm{\frac{I^{\otimes n}}{2^n}}_\infty = \frac{1}{2^n},
\end{align}
i.e., $U\rho U^\dag$ is sampled from a low-average ensemble with purity 1.

The second moment of the expectation value of an observable $O$, when measured on a state sampled from a low-average ensemble, can be upper bounded as follows.
\begin{lemma}[Lemma 3 in Ref.~\cite{schuster2024polynomial}]
\label{lem:schuster}
 Let $O$ be an observable and $\calD$ be a low-average ensemble over quantum states with purity $c$, i.e. $   \norm{\bbE_{\rho\sim\calD} [\rho]}_\infty \leq c/2^n.$ We have
 \begin{align}
     \bbE_{\rho\sim\calD} \Tr[O \rho]^2 \leq c\, \norm{O}_\mathrm{F}^2. 
 \end{align}
\end{lemma}
\begin{proof}
We provide an elementary proof of this fact. We have:
\begin{align}
    \mathbb{E}_{\rho\sim\mathcal{D}} \Tr[O \rho]^2 \leq \mathbb{E}_{\rho\sim\mathcal{D}} \Tr[O^2 \rho] \leq \norm{O^2}_1 \norm{\mathbb{E}_{\rho\sim\mathcal{D}} \rho}_\infty = \norm{O}_2^2 \norm{\mathbb{E}_{\rho\sim\mathcal{D}} \rho}_\infty \leq \frac{c}{2^n} \norm{O}_2^2.
\end{align}
In the first step, we used the well-known fact that the variance of a quantum observable is non-negative.  
In the second step, we applied Hölder's inequality.  
In the third step, we used the property that the Schatten \( p \)-norm of an observable corresponds to the \( p \)-norm of its eigenvalue vector. 
Finally, in the last step, we used the assumption on the ensemble \( \mathcal{D} \).  

\end{proof}

}
\subsection{Noise channels}
In this work, we study quantum circuits interspersed by local noise. To this end, we model noise with channels of form $\calN^{\otimes n}$, where $\calN$ is an arbitrary single-qubit channel.
In general, a single-qubit channel \(\mathcal{N}\) can be fully characterized by its action on the Pauli matrices $\calP_1 = \{I,X,Y,Z\}$. Thus, $\calN$ can be identified by \(16\) transition amplitudes of the form \(\llangle P | \widehat{\mathcal{N}} | Q \rrangle\), where \(P, Q \in \mathcal{P}_1\). Since \(\mathcal{N}\) is trace-preserving, it satisfies \(\llangle I | \widehat{\mathcal{N}} | Q \rrangle = 0\) for all \(Q \in \{X, Y, Z\}\) and \(\llangle I | \widehat{\mathcal{N}} | I \rrangle = 1\). Consequently, there are \(12\) remaining transition amplitudes whose values must be determined. 

Moreover, as shown in Refs.~\cite{ king2001minimal,mele2024noise}, any single-qubit channel can be decomposed into unitary and non-unitary components. For the purposes of the present analysis, we focus exclusively on the non-unitary component, which can be fully described using only 6 parameters. In this context, we state a useful lemma from Ref.~\cite{mele2024noise}.
\begin{lemma}[Normal form of a quantum channel~\cite{king2001minimal,beth2002analysis}]
\label{le:normal}
    Any single-qubit quantum channel $\mathcal{N}$ can be written in the so called `normal' form:
    \begin{align}
        \mathcal{N}(\cdot)=U\mathcal{N}^{\prime}(V^{\emph\dag}(\cdot)V)U^{\emph\dag},
    \end{align}
    where $U$, $V$ are unitaries and $\mathcal{N}^{\prime}(\cdot)$ is a quantum channel 
    defined by the following transition amplitudes:
    \begin{align}
        \llangle P | \widehat{\mathcal{N}^{\prime}}| Q \rrangle \coloneqq
        \begin{cases}
            1 & \text{if $P=Q=I$},\\
            D_P & \text{if $P=Q \neq I$}, \\
            t_P & \text{if $P\neq I$ and  $Q = I$}, \\
            0 & \text{otherwise.}
        \end{cases}
    \end{align}

where $\bold{t} \coloneqq (t_X,t_Y,t_Z)$ and $\bold{D} \coloneqq (D_X,D_Y,D_Z) \in \mathbb{R}^3$, such that the entries of $\bold{D}$ have all the same sign. 
\end{lemma}

In the following, we present an exhaustive classification of single-qubit (noise) channels.
\begin{itemize}
    \item \textbf{Unitary noise.} A noise channel $\calN$ is unitary, or coherent, if $\calN = U(\cdot) U^\dag$ for some unitary $U\in\mathbb{U}(2)$. As our analysis hinges on the contraction properties of non-unitary channels, unitary noise is beyond the scope of the present work.
    \item \textbf{Non-unitary unital noise.} 
    A quantum channel $\calN$ is \emph{unital} if it preserves the identity, i.e., if it satisfies $\calN(I) = I$. Let the normal form of the channel be $\mathcal{N}(\cdot)=U\mathcal{N}^{\prime}(V^{\dag}(\cdot)V)U^{\dag}$. As $U(\cdot)U^\dag$ and $V(\cdot)V^\dag$ are unital, then $\calN'$ must be unital as well. 
    By unitality, we have that $\bold{t} = (0,0,0)$.
    Moreover, we observe that $\bold{D} \neq (1,1,1)$, otherwise the $\calN'$ would be the identity channel and $\calN$ would be unitary. 
    We can further classify non-unitary unital noise in two classes:
    \begin{itemize}
        \item \textbf{Depolarizing class.}  A channel $\calN$ belong to the depolarizing class if all the entries $D_X, D_Y$ and $D_Z$ are strictly smaller than 1. 
    As a consequence, the limit of repeatedly applying $\calN$ on any input state is the maximally mixed state, which is also the only fixed point of the noise. 
Notably, the depolarizing class includes the depolarizing channel, which corresponds to the case $D_X=D_Y=D_Z=1-p$, where $p\in (0,1]$ is the noise rate of the channel.
\item \textbf{Dephasing class. } A channel $\calN$ belongs to the dephasing class if exactly one entry of the vector $\bold{D}$ is one and the other two are strictly smaller than one.
This implies that the limit of repeatedly applying $\calN$ on any input state lies on a diameter of the Bloch sphere.
The dephasing class includes the dephasing noise, which satisfy $D_Z = 1$ and $D_X=D_Y=(1-2p)$, where $p \in (0,1]$ is the noise rate of the channel. 

All state $\rho$ of the form $\rho = ({I + aZ})/{2}$ for $a \in [-1,1]$ is a fixed point of the dephasing noise.
    \end{itemize}
Furthermore, since there is no single-qubit channel with exactly two entries in the vector $\bold{D}$ strictly less than 0, then any non-unitary, unital noise must belong to either the depolarizing class or the dephasing class. 
 
\item  \textbf{Non-unital noise. } A channel $\calN$ is non-unital if $\calN(I)\neq I$. Consequently, we have that also $\calN'(I)\neq I$ and $\bold{t} \neq (0,0,0)$. This implies that channel $\calN$ is also non-unitary.
The limit of repeatedly applying $\calN$ on any input state leads to a state $\rho$ which is not the maximally mixed state.
Crucially, unlike noise channels in the other two classes, non-unital noise can increase the purity of an operator. 
\end{itemize}

Our analysis extensively relies on the normal form of the channels. As a preliminary step, we derive the following constraint on the parameters of this normal form.

\begin{lemma}[Constraint on normal form parameters]
\label{lem:constraint}
Let $\calN(\cdot) = U\mathcal{N}^{\prime}(V(\cdot)V^\dag)U^{\dag}$ 
be an arbitrary single-qubit channel with normal form parameters $\bold{D}, \bold{t}$. 
Let $\Upsilon(\cdot, \cdot)$ be the function defined as follows
\begin{align}
    \Upsilon(\bold{D}, \bold{t}) \coloneqq 
    \max_{\substack{
    \norm{\boldsymbol{a}}_2^2 = 1
    }} \left\{\sum_{Q\in\{X,Y,Z\}} a_Q^2 D_Q^2 + \left(\sum_{Q\in\{X,Y,Z\}} a_Q t_Q\right)^2 \right\} \label{eq:ups}.
\end{align}
We have that
\begin{align}
    \Upsilon(\bold{D}, \bold{t})  \leq 1.
\end{align}
Moreover the inequality is strict provided that $\norm{\bold{D}}_\infty^2 \in (0,1)$ or $\norm{\bold{t}}_2^2 \in (0,1)$.
\end{lemma}
\begin{proof}Let $(b_X,b_Y,b_Z)$ be the coefficients maximizing Eq.~\eqref{eq:ups}, i.e.
\begin{align}
    (b_X,b_Y,b_Z) \coloneqq \underset{\substack{\boldsymbol{a} \in \mathbb{R}^3\\ \norm{\boldsymbol{a}}_2^2 = 1}}{\arg\max} \, \sum_{Q\in\{X,Y,Z\}} a_Q^2 D_Q^2 + \left(\sum_{Q\in\{X,Y,Z\}} {a_Q} t_Q\right)^2.
    \label{eq: contr-ub}
\end{align}
We define the observable $O$ and the state $\rho$ as follows:
\begin{align}
    O = \sum_{P\in \{X,Y,Z\}}\abs{b_P} \cdot \mathrm{sgn}(t_P \cdot D_P) P \qquad  \text{ and }
    \qquad \rho = \frac{I + O}{2},
\end{align}
where $\mathrm{sgn}(x)$ is the sign of $x$.
We observe that the operator norm of $O$ equals 1. First, we can upper bound it as follows
\begin{align}
    \norm{O}_\infty = \max_{\sigma} \abs{\Tr[O\sigma]} = \max_{\substack{\boldsymbol{r} \in \mathbb{R}^3\\ \norm{\boldsymbol{r}}_2^2 = 1}}\abs{\sum_{P\in\{X,Y,Z\}} \abs{b_P} r_P \cdot \mathrm{sgn}(t_P\cdot D_P)}
    \leq  \max_{\substack{\boldsymbol{r} \in \mathbb{R}^3\\ \norm{\boldsymbol{r}}_2^2 = 1}}\norm{\boldsymbol{b}}_2\norm{\boldsymbol{r}}_2=1,
\end{align}
where we used the Cauchy-Schwarz inequality. Moreover, the inequality is saturated choosing $\sigma = \rho$, so $\norm{O}_\infty =1$.
We can exploit this fact to upper bound Eq.~\eqref{eq: contr-ub}. Recall that the channel $\calN'$ acts as
\begin{align}
    \calN'\left(\frac{1}{2}\left(I + \sum_{P\in\{X,Y,Z\}} r_PP\right)\right) = \frac{I}{2} + \frac{1}{2}\sum_{P\in\{X,Y,Z\}} (r_PD_P + t_P)P.
\end{align} 
Replacing the coefficients $r_P$ with the appropriate values for the state $\rho$, we obtain that
\begin{align}
    {\Tr[O\calN'(\rho)]}= &\frac{1}{2}{\Tr\left[\left(\sum_{P\in \{X,Y,Z\}}\abs{b_P} \cdot \mathrm{sgn}(t_P\cdot D_P)P\right) \calN\left(I + \sum_{P\in \{X,Y,Z\}}\abs{b_P} \cdot \mathrm{sgn}(t_P\cdot D_P) P\right)\right]}
    \\ = &{\sum_{P\in \{X,Y,Z\}}b_P^2 D_P + \abs{b_P t_P} \mathrm{sgn}(D_P)}\,.
\end{align}
By Hölder's inequality, we have that
\begin{align}
   1= \norm{O}_\infty\norm{\rho}_1\geq \abs{\Tr[O\calN'(\rho)]}
    = &\abs{\sum_{P\in \{X,Y,Z\}}b_P^2 D_P + \abs{b_P t_P} \mathrm{sgn}(D_P)} \\ = &\sum_{P\in \{X,Y,Z\}}b_P^2 \abs{D_P} + \sum_{P\in \{X,Y,Z\}}\abs{b_P t_P} 
    \\\geq  &\sum_{P\in \{X,Y,Z\}}b_P^2 D_P^2 +  \left(\sum_{P\in \{X,Y,Z\}} b_P t_P \right)^2,
\end{align}
where the second-to-last step follows from the fact that the coefficients $D_P$ have all the same sign (cf. Lemma\ \ref{le:normal}) and in the last step we used the fact that $x^2 \leq x$ if $\abs{x}\leq 1$.

It remains to determine under which circumstances the last inequality is strict.
We observe that
\begin{align}
 \norm{\bold{D}}_\infty^2 \in (0,1) & \implies \forall {P\in\{X,Y,Z\}} : D_P^2 < \abs{D_P}
  \\& \implies \sum_{P\in \{X,Y,Z\}}b_P^2 D_P^2 < \sum_{P\in \{X,Y,Z\}}b_P^2 \abs{D_P},
\end{align}
where in the first step we used the fact that $x^2 < x$ if $\abs{x} < 1$ and in the second step we used the fact that $\norm{\boldsymbol{b}}_2^2=1$.
Moreover, by Cauchy-Schwarz inequality, we have
\begin{align}
    \left(\sum_{P\in \{X,Y,Z\}} b_P t_P \right)^2 \leq \norm{\bold{t}}_2^2\norm{\boldsymbol{b}}_2^2 =  \norm{\bold{t}}_2^2.
\end{align}
and therefore
\begin{align}
\norm{\bold{t}}_2^2 \in (0,1) \implies \left(\sum_{P\in \{X,Y,Z\}} b_P t_P \right)^2 < 
  \sum_{P\in \{X,Y,Z\}}\abs{b_P t_P},
\end{align}
where we used again the fact that $x^2 < x$ if $\abs{x} < 1$.
Putting all together, we obtain that
\begin{align}
     \sum_{P\in \{X,Y,Z\}}b_P^2 D_P^2 +  \left(\sum_{P\in \{X,Y,Z\}} b_P t_P \right)^2 < 1,
\end{align}
provided that $\norm{\bold{D}}_\infty^2 \in (0,1)$ or $\norm{\bold{t}}_2^2 \in (0,1) $. 
    
\end{proof}

\subsection{The Pauli propagation method}
We provide a brief introduction to the Pauli Propagation framework for classically simulating quantum circuits.
Let a circuit be represented by an $L$-layered quantum channel $\calC$
\begin{align}
    \calC = \calC_L \circ \calC_{L-1} \dots \circ\calC_1.
\end{align}
{Then, its matrix (vectorized) form is:}
\begin{align}
    \widehat{\calC} = \widehat{\calC}_L \widehat{\calC}_{L-1} \dots \widehat{\calC}_1.
\end{align}
Given a Pauli path $\gamma = (P_0, P_1, P_2,\dots, P_L) \in \calP_n^{L+1}$, we denote the associated Fourier coefficients as
\begin{align}
    \Phi_\gamma(\calC) \coloneqq 
    \llangle P_L | \widehat{\calC}_L | P_{L-1} \rrangle \times \llangle P_{L-1} | \widehat{\calC}_{L-1} | P_{L-2} \rrangle \times \dots \times \llangle P_1 | \widehat{\calC}_{1} | P_{0} \rrangle
\end{align}
Thus, we can express the Schrödinger-evolved state $\calC(\rho)$ and the Heisenberg-evolved observable $\calC^{\dag}(O)$ as sums over Pauli paths 
\begin{align}
&\calC(\rho) = \sum_{\gamma =(P_0,P_1,\dots, P_L) \in \calP_n^{L+1}} \llangle \rho| P_0\rrangle \Phi_\gamma(\calC) P_L,
  \\&\calC^{\dag}(O) = \sum_{\gamma =(P_0,P_1,\dots, P_L) \in \calP_n^{L+1}} \llangle O| P_L\rrangle \Phi_\gamma(\calC) P_0.
\end{align}

In order to compute such evolutions approximately, we consider a suitable subset of Pauli paths $\calA \subseteq \calP_n^{L+1}$ and the associated ``truncated linear map''  defined as
\begin{align}
    \calC_{\calA}(\cdot) \coloneqq 
    \sum_{\gamma =(P_0,P_1,\dots, P_L) \in \calA} \llangle \,\cdot\,| P_0\rrangle \Phi_\gamma(\calC) P_L\,.
\end{align}
Note that the adjoint of $\calC_{\calA}$ is
\begin{align}
        \calC^{\dag}_\calA(\cdot) =
    \sum_{\gamma =(P_0,P_1,\dots, P_L) \in \calA} \llangle \,\cdot\,| P_L\rrangle \Phi_\gamma(\calC) P_0\,.
\end{align}
Then we have
\begin{align}
&\calC_{\calA}(\rho) = \sum_{\gamma =(P_0,P_1,\dots, P_L) \in \calA} \llangle \rho| P_0\rrangle \Phi_\gamma(\calC) P_L,
  \\&\calC^{\dag}_\calA(O) = \sum_{\gamma =(P_0,P_1,\dots, P_L) \in \calA} \llangle O| P_L\rrangle \Phi_\gamma(\calC) P_0.
\end{align}

\subsubsection{Orthogonality of Pauli paths}

In this work, we categorize Pauli paths based on their support and weight, defined as follows.

\begin{definition}[Path support and path weight]
Let $\gamma = (P_0, P_1, \ldots, P_L) \in \calP_n^{L+1}$ be a Pauli path.
\begin{itemize}
    \item The support of $\gamma$ is the vector of the supports of its components:
    \begin{align}
        \mathrm{supp}(\gamma) = (\mathrm{supp}(P_0), \mathrm{supp}(P_1),\dots, \mathrm{supp}(P_L)).
    \end{align}
    \item The weight of $\gamma$ is the sum of the Pauli weight of its components:
    \begin{align}
        \abs{\gamma} = \sum_{i=1}^L \abs{P_i}.
    \end{align}
\end{itemize}
\end{definition}

We will also make extensive use the following orthogonality relationships. 

\begin{lemma}[Orthogonality of Pauli paths]
\label{lem:ortho2}
Let $\calC = \calC_L \circ \calC_{L-1}\circ \dots \circ \calC_1$ be an $L$-layered quantum channel sampled from a distribution $\calD_{\mathrm{circ}}$. 
\begin{enumerate}
    \item Assume that all layers $\calC_j$ are sampled independently from locally unbiased distributions $\calD_j$, and moreover $\calC_L^\dag$ is also sampled from a locally unbiased distribution. Then the Fourier coefficients of paths with different supports are uncorrelated:
    \begin{align}
         \mathrm{supp}(\gamma)\neq \mathrm{supp}(\gamma')\implies \bbE_{\calC\sim\calD_{\mathrm{circ}}} \left[\Phi_\gamma(\calC)\Phi_{\gamma'}(\calC) \right]= 0.
    \end{align}
    \item Assume that all layers $\calC_j$ are sampled independently from Pauli invariant distributions $\calD_j$, and moreover $\calC_L^\dag$ is also sampled from a Pauli invariant distribution. Then the Fourier coefficients of different paths are uncorrelated:
    \begin{align}
\gamma  \neq \gamma' \implies  \bbE_{\calC\sim\calD_{\mathrm{circ}}} \left[ \Phi_\gamma(\calC) \Phi_{\gamma'}(\calC)\right]
=0.
\end{align}
\end{enumerate}
\end{lemma}
\begin{proof}
 We can express the product of the two coefficients $\Phi_\gamma(\calC)  \Phi_{\gamma'}(\calC)$ as follows:  
  \begin{align}
     &\Phi_\gamma(\calC)  \Phi_{\gamma'}(\calC) 
     \\=&\llangle P_L | \widehat\calC_L| P_{L-1} \rrangle
      \llangle P'_L | \widehat\calC_L| P'_{L-1} \rrangle
      \llangle P_{L-1} | \widehat\calC_{L-1} | P_{L-2} \rrangle 
      \llangle P'_{L-1} |\widehat\calC_{L-1} | P'_{L-2} \rrangle\dots
       \dots \llangle P_1 |\widehat\calC_1 | P_{0} \rrangle \llangle P'_1 |  \widehat\calC_1 | P'_{0} \rrangle
      \\=&\llangle P_L\otimes P_L' | \widehat\calC_L^{\otimes 2}| P_{L-1}\otimes P_{L-1}' \rrangle
      \llangle P_{L-1} \otimes P_{L-1}' |\widehat\calC_{L-1}^{\otimes 2} | P_{L-2}\otimes P'_{L-2} \rrangle 
       \dots \llangle P_1\otimes P_1' | \widehat{\calC}_1^{\otimes 2} | P_{0}\otimes P_0' \rrangle.
 \end{align}
 
 \medskip
 
We will prove the two statements separately. Consider two paths $\gamma,\gamma'$ with different supports. Then there exists an index $j$ such that $\mathrm{supp}(P_j) \neq  \mathrm{supp}(P_j')$. If $j \neq L$, we can use the fact that $\calC_{j+1}$ is sampled from a locally unbiased distribution and therefore by Lemma~\ref{lem:ortho1}
\begin{align}
    \mathbb{E}_{\calC_{j+1} \sim \calD_{j+1}} \left[\calC_{j+1}^{\otimes 2}(P_j\otimes P_j')\right] = 0\,.
\end{align}
Therefore, by linearity of the expectation and the trace we have:
\begin{align}
    \mathbb{E}_{\calC_{j+1} \sim \calD_{j+1}} \llangle P_{j+1}\otimes P_{j+1}' | \widehat\calC_{j+1}^{\otimes 2} | P_{j}\otimes P_{j}' \rrangle
    = & \mathbb{E}_{\calC_{j+1} \sim \calD_{j+1}}  \frac{1}{4^n}\Tr\left\{\calC_{j+1}^{\otimes 2}(P_j\otimes P_j') (P_{j+1}\otimes P_{j+1}')\right\}
    \\ = & \frac{1}{4^n}\Tr\left\{ \mathbb{E}_{\calC_{j+1} \sim \calD_{j+1}} [\calC_{j+1}^{\otimes 2}(P_j\otimes P_j')](P_{j+1}\otimes P_{j+1}')\right\} = 0,
\end{align}
which implies that $ \bbE_{\calC\sim\calD_{\mathrm{circ}}} [\Phi_\gamma(\calC)\Phi_{\gamma'}(\calC)] = 0$.
Similarly, if $j = L$, we can instead use the fact that $\calC_L^\dag$ is sampled from a locally unbiased distribution and therefore
\begin{align}
        \mathbb{E}_{\calC_L \sim \calD_L} \left[\calC_L^{\dag \otimes 2}(P_L\otimes P_L')\right] = 0. 
\end{align}
Hence, we can conclude that $ \bbE_{\calC\sim\calD_{\mathrm{circ}}} [\Phi_\gamma(\calC)\Phi_{\gamma'}(\calC)] = 0$.

\medskip

The proof of the second statement is analogous. 
 Consider two different paths $\gamma,\gamma'$. Then there exists an index $j$ such that $P_j \neq  P_j'$. If $j \neq L$, we can use the fact that $\calC_{j+1}$ is sampled from a Pauli invariant distribution and therefore by Lemma~\ref{lem:ortho1}
\begin{align}
    \mathbb{E}_{\calC_{j+1} \sim \calD_{j+1}} \left[\calC_{j+1}^{\otimes 2}(P_j\otimes P_j')\right] = 0\,.
\end{align}
Thus, employing the linearity of the expectation and the trace as above we can prove that $ \bbE_{\calC\sim\calD_{\mathrm{circ}}} [\Phi_\gamma(\calC)\Phi_{\gamma'}(\calC)] = 0$.
Similarly, if $j = L$, we can instead use the fact that $\calC_L^\dag$ is sampled from a Pauli invariant distribution and therefore
\begin{align}
        \mathbb{E}_{\calC_L \sim \calD_L} \left[\calC_L^{\dag \otimes 2}(P_L\otimes P_L')\right] = 0, 
\end{align}
which also allows us to conclude that $ \bbE_{\calC\sim\calD_{\mathrm{circ}}} [\Phi_\gamma(\calC)\Phi_{\gamma'}(\calC)] = 0$.
\end{proof}

\medskip

\subsubsection{Estimating second moments by sampling Pauli paths}

Let $\calD_{\mathrm{circ}}$ be a distribution over $L$-layered quantum channels $\calC_L \circ \calC_{L-1} \circ \dots \circ \calC_1$ such that all the channels $\calC_j$ are sampled from Pauli invariant distributions, which by Lemma\ \ref{lem:ortho2} implies that the Fourier coefficients of any two distinct Pauli paths are uncorrelated, i.e.
    \begin{align}
\gamma  \neq \gamma' \implies  \bbE_{\calC\sim\calD_{\mathrm{circ}}} \left[ \Phi_\gamma(\calC) \Phi_{\gamma'}(\calC)\right]
=0. \label{eq:ortho-sampling}
\end{align}
In this section we will demonstrate that several second moment quantities associated to the distribution $\calD_{\mathrm{circ}}$ can be efficiently estimated with a simple classical randomized algorithm.

Specifically, given an observable $O$, we will consider the following general expression: 
\begin{align}
    F(\calD_{\mathrm{circ}}) \coloneqq \sum_{\gamma =(P_0,P_1,\dots, P_L)\in\calP_n^{L+1}} \llangle O| P_L\rrangle^2 \bbE_{\calC\sim\calD_{\mathrm{circ}}} \left[\Phi_\gamma^2(\calC)\right] f(\gamma) \label{eq:second-moment},
\end{align}
where $f : \calP_n^{L+1} \rightarrow [0,1]$ is a function of the Pauli paths.
It is easy to see that many relevant quantities associated to $\calD_{\mathrm{circ}}$ can be written in this form:
\begin{itemize}
    \item \textbf{Variances of expectation values.} Given an observable $O$ and a state $\rho$, the associated variance can be written as
    \begin{align}
        \bbE_{\calC \sim \calD_{\mathrm{circ}}} \Tr[O\calC(\rho)]^2 =  
        &\bbE_{\calC \sim \calD_{\mathrm{circ}}} \left(\sum_{\gamma =(P_0,P_1,\dots, P_L)\in\calP_n^{L+1}} \Phi_\gamma(\calC)\llangle O| P_L\rrangle \Tr[P_0\rho] \right)^2 \\ = &\sum_{\gamma =(P_0,P_1,\dots, P_L)\in\calP_n^{L+1}}  \bbE_{\calC\sim\calD_{\mathrm{circ}}} \left[\Phi_\gamma^2(\calC)\right] \llangle O| P_L\rrangle^2 \Tr[P_0\rho]^2  ,
    \end{align}
where the second identity follows from the orthogonality relation in Eq.\ \ref{eq:ortho-sampling}. Setting $f_{\mathrm{var}}(\gamma) = \Tr[P_0\rho]^2  $, we find that the variance can be rewritten as in Eq.\ \ref{eq:second-moment}.
    \item \textbf{Mean squared errors.} Given an observable $O$, a state $\rho$, a subset of Pauli paths $\calA \subseteq \calP_{n}^{L+1}$ and its complement $\overline{\calA} = \calP_n^{L+1} \setminus \calA$, we consider the following mean squared error:
    \begin{align}
        \bbE_{\calC\sim\calD_{\mathrm{circ}}} \Tr[\left( \calC^\dag(O) - \calC_{\calA}^\dag(O)\right)\rho]^2 =  
        &\bbE_{\calC \sim \calD_{\mathrm{circ}}} \left(\sum_{\gamma =(P_0,P_1,\dots, P_L)\in\overline{\calA}} \Phi_\gamma(\calC)\llangle O| P_L\rrangle \Tr[P_0\rho] \right)^2 ,
        \\ = &\sum_{\gamma =(P_0,P_1,\dots, P_L)\in\overline{\calA}}  \bbE_{\calC\sim\calD_{\mathrm{circ}}} \left[\Phi_\gamma^2(\calC)\right] \llangle O| P_L\rrangle^2 \Tr[P_0\rho]^2  
        \\ = & \sum_{\gamma =(P_0,P_1,\dots, P_L)\in\calP_n^{L+1}}  \bbE_{\calC\sim\calD_{\mathrm{circ}}} \left[\Phi_\gamma^2(\calC)\right] f_{\mathrm{\calA}}(\gamma),
    \end{align}
    where we used again orthogonality relation of Pauli paths in Eq.\ \ref{eq:ortho-sampling} and we defined the function 
    \begin{align}
        f_{\mathrm{\calA}}(\gamma)= \begin{cases}
            \Tr[P_0\rho]^2 & \text{if $\gamma \not\in \calA$,}\\
            0 & \text{if $\gamma \in\calA$}.
        \end{cases}
    \end{align}
    Alternatively, one can evaluate the mean squared error with respect to the normalized Frobenius norm \\ $\norm{\calC^\dag(O) - \calC^\dag_\calA(O)}_{\mathrm{F}}$:
    \begin{align}
        \bbE_{\calC\sim\calD_{\mathrm{circ}}} \norm{\calC^\dag(O) - \calC^\dag_\calA(O)}_{\mathrm{F}}^2
        = &\frac{1}{4^n}  \bbE_{\calC\sim\calD_{\mathrm{circ}}} \Tr[\left(\sum_{\gamma =(P_0,P_1,\dots, P_L)\in\overline{\calA}} \Phi_\gamma(\calC)\llangle O| P_L\rrangle P_0  \right)^2] 
        \\ = &\sum_{\gamma =(P_0,P_1,\dots, P_L)\in\overline{\calA}}  \bbE_{\calC\sim\calD_{\mathrm{circ}}} [\Phi_\gamma^2(\calC)] \llangle O| P_L\rrangle^2 
        \\ = &\sum_{\gamma =(P_0,P_1,\dots, P_L)\in\calP_n^{L+1}}\bbE_{\calC\sim\calD_{\mathrm{circ}}} [\Phi_\gamma^2(\calC)] f_{\calA, {\mathrm{Frob}}} (\gamma),
    \end{align}
        where we used again orthogonality relation of Pauli paths in Eq.\ \ref{eq:ortho-sampling}  and we defined the function 
    \begin{align}
        f_{\mathrm{\calA},\mathrm{Frob}}(\gamma) = \begin{cases}
            1  & \text{if $\gamma \not\in \calA$,}\\
            0 & \text{if $\gamma \in\calA$}.
        \end{cases}
    \end{align}
\end{itemize}
We are now ready to state the main result of this section.

\begin{theorem}[Monte Carlo estimates]
\label{thm:num-estimate}
Let $\calD_{\mathrm{circ}}$ be a distribution over $L$-layered quantum channels $\calC_L \circ \calC_{L-1} \circ \dots \circ \calC_1$ such that all the channels $\calC_j$ are sampled from Pauli invariant distributions. 
Let $f : \calP_n^{L+1} \rightarrow [0,1]$ be a function of the Pauli paths computable in time $T$.
Given an observable $O$, let $F(\calD_{\mathrm{circ}})$ be defined as follows
\begin{align}
    F(\calD_{\mathrm{circ}}) \coloneqq \sum_{\gamma =(P_0,P_1,\dots, P_L)\in\calP_n^{L+1}} \llangle O| P_L\rrangle^2  \bbE_{\calC\sim\calD_{\mathrm{circ}}} \left[\Phi_\gamma^2(\calC)\right] f(\gamma).
\end{align}
Assume that we can sample $P\in\calP_n$ with probability proportional to $\llangle O|  P_L \rrangle ^2$ in time $T$. Moreover, assume that for all $j\in[L]$ and for all $P \in \calP_n$, we can also sample $Q\in\calP_n$ with probability proportional to  ${\bbE_{\calC_j \sim \calD_j}\llangle P| \widehat{\calC}_j| Q \rrangle ^2}$ in time $T$.
Then for any $\epsilon,\delta \in (0,1]$, there exists a classical randomized algorithm that runs in time
\begin{align}
    \calO(TL \epsilon^{-2} \log(1/\delta))
\end{align}
and outputs a value $\Lambda$ such that
\begin{align}
    \abs{\Lambda - F(\calD_{\mathrm{circ}})}  \leq \epsilon\norm{O}_\mathrm{F}^2,
\end{align}
with probability at least $1-\delta$.
\end{theorem}

\begin{proof}
We can estimate $F(\mathcal{D}_{\mathrm{circ}})$ using the following randomized algorithm.

Let $M \coloneqq \lceil 2\log(2/\delta)/\epsilon^2 \rceil$. For $j = 1, 2, \dots, M$, repeat the following steps:
    \begin{enumerate}
    \item Sample $P_L$ with probability 
    \begin{align}
        \frac{\llangle O|  P_L \rrangle ^2}{\sum_{P\in\calP_n} \llangle O|  P \rrangle ^2} = \frac{\llangle O|  P_L \rrangle ^2}{\norm{O}_{\mathrm{F}}^2} 
    \end{align}
    \item For $i = L,\dots, 1$, given $P_i$, sample $P_{i-1}$ with probability
    \begin{align}
        \frac{\bbE_{\calC_i\sim\calD_i}\llangle P_i| \widehat{\calC}_i | P_{i-1} \rrangle ^2}{\sum_{P \in \calP_n} \bbE_{\calC_i\sim\calD_i}\llangle P_i| \widehat{\calC}_i | P \rrangle ^2} = \frac{\bbE_{\calC_i\sim\calD_i}\llangle P_i| \widehat{\calC}_i | P_{i-1} \rrangle ^2}{\bbE_{\calC_i\sim\calD_i}\norm{\calC_i^\dag(P_i)}^2_{\mathrm{F}}} 
    \end{align}
    \item Given $\gamma = (P_1, P_2, \dots, P_{L})$, output $ \lambda_j = K(\gamma) \cdot f(\gamma)$, where the normalization factor $K(\gamma)$ is defined as
    \begin{align}
        K(\gamma) = \norm{O}^2_{\mathrm{F}} \prod_{i=1}^L \left(\bbE_{\calC_i\sim\calD_i}\norm{\calC_i^\dag(P_i)}^2_{\mathrm{F}}\right).
    \end{align}
\end{enumerate}
Finally, we compute the empirical mean $\Lambda = \frac{1}{M}\sum_{j=1}^M \lambda_j$.

\smallskip

\noindent We notice that, at each iteration, a path $\gamma$ is sampled with probability
\begin{align}
    \Pr[\text{$\gamma = (P_0, P_1,\dots, P_L)$ is sampled}]  = \frac{\llangle O|  P_L \rrangle ^2}{\norm{O}_{\mathrm{F}}^2}
    \prod_{i=1}^L \left( \frac{\bbE_{\calC_i\sim\calD_i}\llangle P_i| \widehat{\calC}_i | P_{i-1} \rrangle ^2}{\bbE_{\calC_i\sim\calD_i}\norm{\calC_i^\dag(P_i)}^2_{\mathrm{F}}}\right).
\end{align}
Therefore procedure described above yields an unbiased estimator for the target value $F(\calD_{\mathrm{circ}})$, as the expected value of $\lambda_j$ is
\begin{align}
    \bbE [\lambda_j] = &\sum_{\gamma=(P_0,P_1,\dots, P_L)} \Pr[\text{$\gamma$ is sampled}] K(\gamma) f(\gamma)
    = {\llangle O|  P_L \rrangle ^2}
    \prod_{i=1}^L \left({\bbE_{\calC_i\sim\calD_i}\llangle P_i| \widehat{\calC}_i | P_{i-1} \rrangle ^2}\right) f(\gamma) 
    \\ = &\sum_{\gamma =(P_0,P_1,\dots, P_L)\in\calP_n^{L+1}} \llangle O| P_L\rrangle^2  \bbE_{\calC\sim\calD_{\mathrm{circ}}} \left[\Phi_\gamma^2(\calC)\right] f(\gamma) = F(\calD_{\mathrm{circ}}).
\end{align}
Moreover, we observe that 
\begin{align}
     \norm{\calC_i^\dag(P_i)}^2_{\mathrm{F}} \leq \norm{\calC_i^\dag(P_i)}^2_\infty \leq \norm{P_i}^2_\infty = 1 .
\end{align}
where the first inequality follows from the relation $2^n \norm{A}_F^2 = \norm{A}_2^2 \leq 2^n \norm{A}_\infty^2$ and the second inequality follows from the Russo-Dye theorem~\cite{bhatia2015positive} (which states \(\|\calC^{\dagger}(A)\|_\infty \le \|A\|_\infty\) for any quantum channel \(\calC\)), and in the final step we noted that the operator norm of any Pauli matrix is 1.
Thus, the random variable $\lambda_j$ can be bounded as follows
\begin{align}
    {\lambda_j} = K(\gamma) {f(\gamma)} \leq \norm{O}^2_{\mathrm{F}} \prod_{i=1}^L \left(\bbE_{\calC_i\sim\calD_i}\norm{\calC_i^\dag(P_i)}^2_{\mathrm{F}}\right) \leq \norm{O}^2_{\mathrm{F}},
\end{align}
where we used the fact that $f(\gamma)\leq 1$ and $\norm{\calC_i^\dag(P_i)}^2_{\mathrm{F}} \leq 1$. Moreover, $\lambda_j$ is always positive since both $K(\gamma) {f(\gamma)}$ are positive for all $\gamma$. Therefore, $\lambda_j$ lies in the interval $[0,1]$.

We can conclude the proof by invoking the Chernoff-Hoeffding bound, which guarantees the final outcome approximates the mean squared error with precision $\epsilon \norm{O}^2_{\mathrm{F}}$ with probability at least $1-\delta$.
\end{proof}

\section{Technical Lemmas for noisy random channels}
In this section we introduce several technical tools for studying the effect of noise on random circuits.
\subsection{Motivation: bounding the Frobenius norm under Heisenberg evolution}
\label{sec:frobenius}
While our work encompasses arbitrary local noise channels, most of the previous research focuses on the ``depolarizing class''.
A widely employed feature of the depolarizing class is that it contracts the Frobenius norm of any traceless observable.
In particular, let $\calN_p^{(\mathrm{depo})}$ be the single-qubit depolarizing channel with noise rate $p\in[0,1]$ and $O = \sum_{P\in\calP_n} a_P P$ be an observable with Frobenius norm $\norm{O}^2_\mathrm{F} =\sum_{P\in\calP_n} a_P^2$.
We have
\begin{align}
    \calN_p^{(\mathrm{depo})\dag\otimes n}(O)= \calN_p^{(\mathrm{depo})\otimes n}(O) = \sum_{P\in\calP_n} a_P (1-p)^{\abs{P}} P\,,
\end{align}
and therefore 
\begin{align}
    \norm{\calN_p^{(\mathrm{depo})\dag\otimes n}(O)}_\mathrm{F}^2 = \sum_{P\in\calP_n} a_P^2 (1-p)^{2\abs{P}}. 
\end{align}
Thus $\calN_p^{(\mathrm{depo})\dag\otimes n}$ strictly contracts the Frobenius norm if $O$ is traceless, as we have
\begin{align}
    \norm{\calN_p^{(\mathrm{depo})\dag\otimes n}(O)}^2_\mathrm{F} \leq (1-p)^2 \norm{O}^2_\mathrm{F}.
\end{align}
However, this is not true in general for other kinds of noise. For instance, observables supported on $\{I,Z\}^{\otimes n}$ are invariant under the dephasing channel. 
Moreover, non-unital noise can \emph{increase} the Frobenius norm under Heisenberg evolution. Let $\calA_q$ be the amplitude-damping channel with noise rate $q$. We have

\begin{align}
    \calA_q^{\dag \otimes 2}(ZZ+IZ + ZI) = (1-q)^2 ZZ+ (1-q^2)(IZ + ZI) + (2q + q^2)II \,.
\end{align}

Therefore, the corresponding Frobenius norm is
\begin{align}
  \norm{\calA_q^{\dag \otimes 2}(ZZ+IZ + ZI)}^2_\mathrm{F}  =q^2(2 + q)^2 + (1-q)^4 + 2(1-q^2)^2.
\end{align}
In particular, setting $q = 1$, we find that the squared Frobenius norm strictly increases:
\begin{align}
    \frac{\norm{\calA_q^{\dag \otimes 2}(ZZ+IZ + ZI)}^2_\mathrm{F}}{\norm{ZZ+IZ + ZI}^2_\mathrm{F}} = {3}.
\end{align}

\medskip

This preliminary observation suggests that techniques relying on norm truncation may only be applicable when the noise channel belongs to the depolarizing class. Despite this intuition, we will demonstrate that the adjoint of any local noise channel -- potentially non-unital or of a dephasing nature -- contracts the Frobenius norm of any traceless observable under averaging over random unitaries.
Our results generalize the bounds presented in previous works such as Refs.\ \cite{mele2024noise, quek2022exponentially}. While these previous results require the presence of local and global unitary 2-designs respectively, our approach encompasses a wider class of circuit ensembles.

\subsection{Average-case norm contraction}
We provide the following definition of contraction coefficients for single-qubit channels, which will serve as a useful technical tool for studying the average-case contraction of the Frobenius norm under local noise channels.
\begin{definition}[Contraction coefficients]
\label{def:coeff}
Given an observable $O = \sum_{P\in\calP_n}a_PP$ and a subset of qubits $A\subseteq[n]$, denote by $O_A$ the restriction to its Pauli terms being non-identity on $A$ and identity elsewhere:
\begin{align}
    O_A \coloneqq \sum_{\substack{P\in\calP_n\\\mathrm{supp}(P)=A}} a_P P\,.
\end{align}
Given a single-qubit quantum channel $\calN$, we define the contraction coefficient of $\calN$ as
\begin{align}
    \chi(\calN) \coloneqq  \max_{A\subseteq [n]} \max_{\substack{O\\ \norm{O_A}_\mathrm{F}\neq 0}}\left(\frac{\norm{\calN^{\emph\dag\otimes n}(O_A)}_\mathrm{F}}{\norm{O_A}_\mathrm{F}}\right)^{1/\abs{A}}.
\end{align}
Moreover, given a distribution $\calD$ over $\mathbb{U}(2)$, we define the \emph{mean squared} contraction coefficient of $\calN$ with respect to the $\calD$ as
\begin{align}
    \chi_\calD^2(\calN) \coloneqq   \max_{A\subseteq [n]} \max_{\substack{O\\ \norm{O_A}_\mathrm{F}\neq 0}}\left(\bbE_{W\sim\calD^{\otimes n}}\frac{\norm{\calN^{\emph\dag\otimes n}(W^\emph\dag O_A W)}^2_\mathrm{F}}{\norm{O_A}_\mathrm{F}^2}\right)^{1/\abs{A}}.
\end{align}
\end{definition}

We also remark that the contraction coefficient of a channel $\calN(\cdot) = U\calN' V(\cdot) V^\dag U^\dag $ can be expressed in terms of $\calN'$.

\begin{observation}
Given a noise channel expressed in normal form as $\calN(\cdot) = U\calN' V(\cdot) V^\dag U^\dag $, then $\calN$ and $\calN'$ have the same contraction coefficient $\chi(\calN)$. This follows from the following two facts: first, for any observable $O$ and subset of qubits $A$, we have that $(U^{\dag\otimes n} O U^{\otimes n})_A = U^{\dag\otimes n} O_A U^{\otimes n}$, and second, the Frobenius norm is unitarily invariant. Thus, we obtain that
\begin{align}
    \chi(\calN) = &\max_{A\subseteq [n]} \max_{\substack{O\\ \norm{O_A}_\mathrm{F}\neq 0}}\left(\frac{\norm{V^{\dag\otimes n}\calN'^{\dag\otimes n}(U^{\dag\otimes n}O_A U^{\otimes n})V^{\otimes n}}_\mathrm{F}}{\norm{O_A}_\mathrm{F}}\right)^{1/\abs{A}}
    \\= &\max_{A\subseteq [n]} \max_{\substack{O\\ \norm{O_A}_\mathrm{F}\neq 0}}\left(\frac{\norm{\calN'^{\dag\otimes n}(O_A)}_\mathrm{F}}{\norm{O_A}_\mathrm{F}}\right)^{1/\abs{A}}
    = \chi(\calN').
\end{align}
\end{observation}

Applying iteratively Lemma\ \ref{lem:constraint}, we can prove that the contraction coefficient $\chi(\calN)$ is strictly smaller than 1 for noise channels within the depolarizing class or the non-unital class.

\begin{lemma}[General upper bound on the contraction coefficients]
\label{lem:contr-general}
Let $\calN (\cdot) = U \calN'(V(\cdot)V^\dag)U^\dag$ be a single-qubit channel with normal form parameters $\bold{D}$ and $\bold{t}$. We have

\begin{align}
    \chi^2(\calN) \leq {\Upsilon(\bold{D},\bold{t})} \coloneqq     \max_{\substack{
    \norm{\boldsymbol{a}}_2^2 = 1
    }} \sum_{Q\in\{X,Y,Z\}} a_Q^2 D_Q^2 + \left(\sum_{Q\in\{X,Y,Z\}} a_Q t_Q\right)^2   \,.
\end{align}
In particular, by Lemma~\ref{lem:constraint}, this implies that the contraction coefficient of any single-qubit channel $\calN$ is at most 1, and it is strictly smaller than 1 if $\calN$ satisfies $\norm{\bold{D}}_\infty^2 \in (0,1)$ or $\norm{\bold{t}}_2^2 \in (0,1)$. 
\end{lemma}
\begin{proof}
As a preliminary step, we prove the following claim.
\begin{claim}
\label{claim:contr}
Let $H = \sum_{P\in\{X,Y,Z\}}\sum_{Q \in \{I,X,Y,Z\}^{\otimes (k-1)}} b_{P\otimes Q} P\otimes Q $ be an observable. 
We have
\begin{align}
    \norm{\calN'^{\dag}\otimes \mathcal{I}_{k-1}(H)}_\mathrm{F}^2 \leq c \norm{H}_\mathrm{F}^2,
\end{align}
where $\mathcal{I}_{k-1}(\cdot)$ is the identity channel acting on $k-1$ qubits and $c\coloneqq \Upsilon(\bold{D},\bold{t})$.
\end{claim}
\begin{proof}[Proof of the Claim]
We start by writing the expression of $\calN'^{\dag}\otimes \mathcal{I}_{k-1}(H)$
\begin{align}
    \calN'^{\dag}\otimes \mathcal{I}_{k-1}(H)=
   \sum_{Q \in \{I,X,Y,Z\}^{\otimes (k-1)}}\left( \sum_{P\in\{X,Y,Z\}} D_P b_{P\otimes Q} P\otimes Q
    + \left(\sum_{P\in\{X,Y,Z\}} b_{P\otimes Q}  t_P\right) I\otimes Q\right).
\end{align}
Therefore the normalized Frobenius norm of $\calN'^{\dag}\otimes \mathcal{I}_{k-1}(H)$ satisfies
\begin{align}
    \norm{\calN'^{\dag}\otimes \mathcal{I}_{k-1}(H)}^2_\mathrm{F}
    &=  \sum_{Q \in \{I,X,Y,Z\}^{\otimes (k-1)}} \left(\sum_{P\in\{X,Y,Z\}} D_P^2 b^2_{P\otimes Q} 
    + \left(\sum_{P\in\{X,Y,Z\}} b_{P\otimes Q}  t_P\right)^2\right)
    \\& =  \sum_{Q \in \{I,X,Y,Z\}^{\otimes (k-1)}} \left(\sum_{P\in\{X,Y,Z\}} D_P^2 \frac{b^2_{P\otimes Q}}{\norm{\boldsymbol{b}_Q}^2_2}
    + \left(\sum_{P\in\{X,Y,Z\}}\frac{b_{P\otimes Q}}{\norm{\boldsymbol{b}_Q}_2} t_P\right)^2\right) {\norm{\boldsymbol{b}_Q}^2_2}
    \\ &\leq \Upsilon(\bold{D},\bold{t}) \sum_{Q \in \{I,X,Y,Z\}^{\otimes (k-1)}} {\norm{\boldsymbol{b}_Q}^2_2}
     \coloneqq c \, \norm{H}_\mathrm{F}^2,
\end{align}
where we denoted ${\norm{\boldsymbol{b}_Q}^2_2}\coloneqq\sum_{P\in\{X,Y,Z\}} b^2_{P\otimes Q} $ and the only inequality follows from Lemma~\ref{lem:constraint}. This concludes the proof.
\end{proof}
Given a subset of qubits $A\subseteq [n]$ with size $\abs{A} = k$, let $O_A = \sum_{\substack{P \in \calP_n \\ \mathrm{supp}(P) = A}} a_P P$ be a linear combination of Pauli operators supported on $A$. Given $j \in A$, there exist some suitable coefficients $b_{P\otimes Q}$ such that
    \begin{align}
    \bigotimes_{\substack{i \in A \setminus \{j\} }} \calN_{\{i\}}'^\dag(O_A) 
    =  \sum_{P\in\{X,Y,Z\}}\sum_{Q \in \{I,X,Y,Z\}^{\otimes (k-1)}} b_{P\otimes Q} P\otimes Q \otimes I^{\otimes (n-k)}
\end{align}
where the channel $\calN_{\{i\}}'^\dag$ act on the $i$-th qubit and we have omitted the identity channel acting on the remaining $n +1 - k$ qubits in order to simplify the notation. Hence the above observable satisfies the assumptions of Claim\ \ref{claim:contr}.
We can thus prove the desired result by applying iteratively Claim~\ref{claim:contr}. 
\begin{align}
    \norm{\calN'^{\dag\otimes n}(O_A)}_\mathrm{F}^2 
    = &\norm{\bigotimes_{\substack{i \in A}} \calN_{\{i\}}'^\dag (O_A)}_\mathrm{F}^2 
     = \norm{  \calN_{\{j_1\}}'^\dag \left(\bigotimes_{\substack{i \in A \setminus \{j_1\}}} \calN_{\{i\}}'^\dag (O_A)\right)}_\mathrm{F}^2 
    \\\leq  c &\norm{ \bigotimes_{\substack{i \in A \setminus \{j_1\}}} \calN_{\{i\}}'^\dag (O_A)}_\mathrm{F}^2 
     =   c\norm{  \calN_{\{j_2\}}'^\dag \left(\bigotimes_{\substack{i \in A \setminus \{j_1, j_2\}}} \calN_{\{i\}}'^\dag (O_A)\right)}_\mathrm{F}^2 
    \\  \leq  c^2&\norm{ \bigotimes_{\substack{i \in A \setminus \{j_1, j_2\}}} \calN_{\{i\}}'^\dag (O_A)}_\mathrm{F}^2 
    \leq \dots \leq c^k \norm{O_A}^2_\mathrm{F},
\end{align}
where we have again omitted the identity channels to simplify the notation. Since $c = \Upsilon(\bold{D}, \bold{t})$, the statement of the Lemma follows.
\end{proof}

Since $\Upsilon(\bold{D}, \bold{t}) < 1$ only when the channel $\mathcal{N}$ is either depolarizing-like or non-unital, a different approach is required to analyze the contraction properties of dephasing-like noise. In particular, we will demonstrate that such noise still induces norm contraction when the initial Pauli observable is scrambled -- i.e. it spreads into at least two distinct Pauli components -- under the action of a unitary transformation. This idea is captured by the notion of \emph{approximate scrambling}, as defined in Definition~\ref{def:mixing}. We formalize this insight in the following lemma.

\begin{lemma}[Contraction coefficients under approximately locally scrambling evolution]
\label{lem:norm-contr-scr}
Let $\calN$ be a single-qubit channel with normal form parameters $\bold{D}$ and $\bold{t}$ and let $\calD$ be an $\eta$-approximate scrambler over $\mathbb{U}(2)$ for some $\eta \in [0,1]$.
\begin{enumerate}
    \item If $\eta = 0$ (i.e., if $\calD$ is a single-qubit unitary $2$-design), we have
    \begin{align}
        \chi_\calD^2(\calN) = \frac{\norm{\bold{D}}_2^2 + \norm{\bold{t}}_2^2}{3} .
    \end{align}
    \item If $\calN$ belongs to the dephasing class, we have
    \begin{align}
        \chi_\calD^2(\calN) \leq \eta + \left({1-\eta}\right)\frac{\norm{\bold{D}}_2^2}{3}.
    \end{align}
\end{enumerate}
\end{lemma}
\begin{proof}
Let $\calN(\cdot) = U\calN'(V(\cdot)V^\dag)U^\dag $ be an arbitrary single-qubit channel. 
   Given $P\in\{X,Y,Z\}$, let $V^\dag W^\dag P W V = \sum_{Q\in \{X,Y,Z \}} a_Q Q$.
By the definition of  $\eta$-approximate scrambler and Lemma~\ref{le:monotone}, we have that
\begin{align}
    \max_{Q\in \{X,Y,Z \}}\bbE_{W\sim\calD^{\otimes n}} \, a_Q^2 =  \max_{Q\in \{X,Y,Z \}} \bbE_{W\sim\calD^{\otimes n}} \llangle P | WU \otimes W^* U^*| Q\rrangle^2 \leq  \frac{1+2\eta}{3},
\end{align}
where for $\eta = 0$ the inequality is in fact an equality.
We obtain that
\begin{align}
    &\bbE_{W\sim\calD^{\otimes n}}\norm{\calN^\dag(W^\dag P W)}^2_{\mathrm{F}} 
 = \bbE_{W\sim\calD^{\otimes n}}\sum_{Q\in\{X,Y,Z\}} a_Q^2 D_Q^2 + \left(\sum_{Q\in\{X,Y,Z\}} a_Q t_Q\right)^2 . 
\end{align}

\noindent\underline{Case 1 : $\eta = 0$, $\calN$ arbitrary.} We have
\begin{align}
    &\bbE_{W\sim\calD^{\otimes n}}\norm{\calN^\dag(W^\dag P W)}^2_{\mathrm{F}} 
    = \frac{1}{3} \left(\norm{\bold{D}}_2^2+ \norm{\bold{t}}_2^2 \right). 
\end{align}
Moreover, given $A\subseteq [n]$ a $O = \underset{\substack{P\in\calP_n \\ \mathrm{supp}(P) = A}}{\sum} b_P P$, we have that 
\begin{align}
    &\bbE_{W\sim\calD^{\otimes n}}\norm{\calN^\dag(W^\dag O W)}^2_{\mathrm{F}} = \underset{\substack{P\in\calP_n \\ \mathrm{supp}(P) = A}}{\sum} b_P^2 \left\{\frac{\norm{\bold{D}}_2^2+ \norm{\bold{t}}_2^2}{3}\right\}^{\abs{A}} = \left\{\frac{\norm{\bold{D}}_2^2+ \norm{\bold{t}}_2^2}{3}\right\}^{\abs{A}} \norm{O}^2_{\mathrm{F}}.
\end{align}

\smallskip

\noindent \underline{Case 2: $\eta \neq 0$, $\mathcal{N}$ within the dephasing class.}
Assume without loss of generality that $D_X = D_Y =  D < 1$ and $D_Z = 1$.
We need to upper bound the following quantity:
\begin{align}
    &\bbE_{W\sim\calD^{\otimes n}}\norm{\calN^\dag(W^\dag P W)}^2_{\mathrm{F}} 
 = \bbE_{W\sim\calD^{\otimes n}}\sum_{Q\in\{X,Y,Z\}} a_Q^2 D_Q^2 =  
 \bbE_{W\sim\calD^{\otimes n}} \left(a_X^2 + a_Y^2\right) D + a_Z^2,
\end{align}
under the constraint that $ \max_{Q\in \{X,Y,Z\}}\bbE_{W\sim\calD} a_Q^2 \leq \frac{1}{3}(1+2\eta)$. When $\bbE_{W\sim\calD}  a_Z^2 = \frac{1}{3}(1+2\eta)$ the expression is maximized and takes value
\begin{align}
\left(1- \frac{1}{3}(1+2\eta)\right) D + \frac{1}{3}(1+2\eta)= \eta + \left({1-\eta}\right)\frac{\norm{\bold{D}}_2^2}{3}.
\end{align}
Moreover, given $A\subseteq [n]$ a $O = \underset{\substack{P\in\calP_n \\ \mathrm{supp}(P) = A}}{\sum} b_P P$, we have that 
\begin{align}
    &\bbE_{W\sim\calD^{\otimes n}}\norm{\calN^\dag(W^\dag O W)}^2_{\mathrm{F}} = \underset{\substack{P\in\calP_n \\ \mathrm{supp}(P) = A}}{\sum} b_P^2 \left\{\eta + \left({1-\eta}\right)\frac{\norm{\bold{D}}_2^2}{3}\right\}^{\abs{A}} = \left\{\eta + \left({1-\eta}\right)\frac{\norm{\bold{D}}_2^2}{3}\right\}^{\abs{A}} \norm{O}^2_{\mathrm{F}}.
\end{align}
\end{proof}

The contraction coefficients proposed in Def.~\ref{def:coeff} can be related to the average-case contraction of the Frobenius norm.

\begin{lemma}[Average norm contraction]
Let $\calN$ be a single-qubit channel and let $\calD$ be a 1-design over $\mathbb{U}(2)$. For all observables $O=\sum_{P\in\calP_n}a_PP$, we have
\begin{align}
    \bbE_{W\sim\calD^{\otimes n}}\norm{\calN^{\dag\otimes n}(W^{\dag}OW)}_\mathrm{F}^2 \leq \sum_{P\in\calP_n} \left\{ \chi^2_\calD(\calN)\right\}^{\abs{P}}a_P^2.
\end{align}
In particular, we also have
\begin{align}
    \bbE_{W\sim\calD^{\otimes n}} \, \norm{\calN^{\emph\dag\otimes n}(W^{\emph\dag}OW)}_\mathrm{F}^2 - 4^{-n}{\Tr[O]^2}
     \leq \left\{ \chi^2_\calD(\calN)\right\}\left(\norm{O}_\mathrm{F}^2 - 4^{-n}{\Tr[O]^2}\right)\,.
\end{align}
\end{lemma}
\begin{proof}
    In general, we can write
\begin{align}
    O = \sum_{w=1}^n \, \sum_{\substack{A \subseteq [n]\\\abs{A}=w}}\, \sum_{\substack{P\in\calP_n\\ \mathrm{supp}(P)=A}}
    a_P P.
\end{align}
We have
\begin{align}
    \bbE_{U\sim\calD} (U^\dag OU)^{\otimes 2} = \bbE_{U\sim\calD} \sum_{w=1}^n \, \sum_{\substack{A \subseteq [n]\\\abs{A}=w}} \left(\sum_{\mathrm{P}\in A}a_P U^\dag PU\right)^{\otimes 2}\,.
\end{align}

Fix a subset $A$ of size $k$,
\begin{align}
    \norm{\calN^{\otimes n\dag}\left(\sum_{P\in A}a_P U^{\dag} P U\right)}_\mathrm{F}^2
    = \norm{\calN^{\otimes k\dag}\left(\sum_{P\in\{X,Y,Z\}^{\otimes k}}b_P P \right)}_\mathrm{F}^2 \leq \left\{ \chi^2_\calD(\calN)\right\}^k \norm{O}_\mathrm{F}^2,
\end{align}
for some suitable coefficients $b_P$.
Putting all together
\begin{align}
     &\bbE_{U\sim\calD}  \norm{\calN^{\dag\otimes n}(U^\dag O U)}^2_\mathrm{F}
     =  \sum_{w=1}^n \, \sum_{\substack{A \subseteq [n]\\\abs{A}=w}}\, \sum_{\substack{P\in\calP_n\\ \mathrm{supp}(P)=A}}
     \bbE_{U\sim\calD} \norm{\calN^{\otimes n\dag}\left(\sum_{P\in A}a_P U^{\dag} P U\right)}_\mathrm{F}^2
     \\ \leq  &\sum_{w=1}^n \, \sum_{\substack{A \subseteq [n]\\\abs{A}=w}}\, \sum_{\substack{P\in\calP_n\\ \mathrm{supp}(P)=A}}  \left\{ \chi^2_\calD(\calN)\right\}^{w} a_P^2 = \sum_{P\in\calP_n} \left\{ \chi^2_\calD(\calN)\right\}^{\abs{P}} a_P^2.
\end{align}

\end{proof}

We conclude this section by explicitly providing the values of $\chi(\calN)$ and $\chi_\calD(\calN)$ for amplitude damping and dephasing noise.

\begin{example}[Contraction coefficients]
\label{obs:coeff-noise}
Let $\calN^{(\mathrm{amp})}_\gamma$ be the amplitude damping channel with noise strength $\gamma$, characterized by the normal form parameters $\bold{D} = (\sqrt{1-\gamma}, \sqrt{1-\gamma}, 1-\gamma)$  and $\bold{t}=(0,0,\gamma)$.

Then we obtain
\begin{align}
    \chi^2\left(\calN^{(\mathrm{amp})}_\gamma\right) \leq &\max_{\substack{
    \norm{\boldsymbol{a}}_2^2 = 1
    }} \sum_{Q\in\{X,Y,Z\}} a_Q^2 D_Q^2 + \left(\sum_{Q\in\{X,Y,Z\}} a_Q t_Q\right)^2 
   \\ \leq & \norm{\bold{D}}^2_\infty + \norm{\bold{t}}_2^2 = 1 - \gamma + \gamma^2,
\end{align}
where the first step follows from  Lemma~\ref{lem:contr-general}  and the second step follows from 
Hölder's inequality.

\smallskip

Furthermore, we also consider the dephasing channel $\calN^{(\mathrm{deph})}_p$ with noise strength $p$, characterized by the normal form parameters $\bold{D} = (1-2p, 1-2p, 1)$  and $\bold{t}=(0,0,0)$. Let $\calD$ be the distribution over $\mathbb{U}(2)$ obtained by performing two Pauli rotations along orthogonal directions with angles sampled independently and uniformly from $[0,2\pi)$.

By Example~\ref{ex:scr-pauli-rot}, we know that the $\calD$ is an $\eta$-approximate scrambler with $\eta = 1/4$. Therefore, Lemma~\ref{lem:norm-contr-scr} yields the following mean squared contraction coefficient:
\begin{align}
    \chi^2_\calD\left(\calN^{(\mathrm{deph})}_p\right) = \frac{1 + (1-2p)^2}{2}.
\end{align}
\end{example}

\subsection{Effective depolarizing rate of noisy random channels}
As previously observed in Ref.~\cite{schuster2024polynomial}, every quantum channel can be re-expressed as a sequential composition of the depolarizing channel and a suitable linear map, which is in general non-physical.
In particular, given a value $p\in[0,1]$ , we decompose the noisy channel $\calN$ as
\begin{align}
    \calN =   \calN_p^{(\mathrm{depo})} \circ \tilde\calN_p, \label{eq:depol-decomp}
\end{align}
where $\calN_p^{(\mathrm{depo})}$ is the depolarizing channel with noise rate $p$, i.e., $\calN_p^{(\mathrm{depo})}(\cdot) = (1-p)(\cdot) + p \frac{I}{2} \Tr[\cdot]$, and $\tilde\calN_p$ is the a linear map given by
\begin{align}
    &\tilde\calN_p(I) =  \calN(I)/({1-p}) - I p/(1-p)\,,
    \\&\tilde\calN_p(X) =\calN(X)/ (1-p),
    \\&\tilde\calN_p(Y) = \calN(Y)/ (1-p),
    \\&\tilde\calN_p(Z) = \calN(Z)/ (1-p).
\end{align}
Moreover, the adjoint channel $\tilde\calN_p^\dag$ is given by
\begin{align}
    &\tilde\calN_p^\dag(I) = I
    \\&\tilde\calN_p^\dag(X) =\calN^\dag(X)/ (1-p),
    \\&\tilde\calN_p^\dag(Y) = \calN^\dag(Y)/ (1-p),
    \\&\tilde\calN_p^\dag(Z) = \calN^\dag(Z)/ (1-p).
\end{align}
In particular, in Ref.~\cite{schuster2024polynomial} the above decomposition is applied to a specific class of non-unital noise, which models a spontaneous emission in a random direction with rate $\gamma_s$. Intuitively, such randomization is performed implicitly in random circuits, via the random gates before
and after each noise channel. 

Crucially, the authors of Ref.~\cite{schuster2024polynomial} observed that the associated non-physical map, $\tilde{\calN}_p^\dag$, does not increase the Frobenius norm of any operator on average over random noise channel directions, provided that \( p \leq \gamma_s / 2 \).  

\medskip  

We extend this result by demonstrating that a similar phenomenon holds for \emph{any} noise channel \(\calN\) when averaged over a tensor product of single-qubit gates sampled from unitary 1-designs. Specifically, we demonstrate that the Frobenius norm does not increase on average over \(\calD\), as long as value of \(p\) in Eq.\ \ref{eq:depol-decomp} is appropriately chosen according to the contraction coefficient \(\chi_\calD(\calN)\).

\begin{lemma}
\label{lem:norm-not-incr}
Let $\calD$ be a 1-design over $\mathbb{U}(2)$ and let $p\coloneqq {1-\chi_\calD(\calN)}$. 
For all observables $O$ we have
 \begin{align}
     \bbE_{V\sim\calD^{\otimes n}} \norm{\tilde\calN^{{\dag}\otimes n }_p (V^{\dag} OV)}_\mathrm{F}^2 \leq \norm{O}_\mathrm{F}^2,
 \end{align}
 that is, the linear map $\tilde\calN^{\emph\dag\otimes n}_p$ does not increase the Frobenius norm in expectation over a randomly sampled $V$.
\end{lemma}
\begin{proof}
We denote $ O   = \sum_{P\in\calP_n} a_PP$. We have
\begin{align}
    \bbE_{V\sim\calD^{\otimes n}} \norm{\tilde\calN^{{\dag}\otimes n }_p (V^\dag O V)}_\mathrm{F}^2
    = &\bbE_{V\sim\calD^{\otimes n}} \sum_{k=1}^n \, \sum_{\substack{A\subseteq [n]\\ \abs{A}=k}} \,  \norm{ \sum_{\substack{P\in\calP_n\\ \mathrm{supp}(P)=A}} a_P \tilde\calN^{{\dag}\otimes n }_p(V^\dag P V)}_\mathrm{F}^2
    \\ =&\bbE_{V\sim\calD^{\otimes n}} \sum_{k=1}^n \, \sum_{\substack{A\subseteq [n]\\ \abs{A}=k}} \,  \frac{1}{(1-p)^{2k}}\norm{ \sum_{\substack{P\in\calP_n\\ \mathrm{supp}(P)=A}} a_P \calN^{{\dag}\otimes n }_p(V^\dag P V)}_\mathrm{F}^2
    \\ \leq &\bbE_{V\sim\calD^{\otimes n}} \sum_{k=1}^n \, \sum_{\substack{A\subseteq [n]\\ \abs{A}=k}} \,  \norm{ \sum_{\substack{P\in\calP_n\\ \mathrm{supp}(P)=A}} a_P \calN^{{\dag}\otimes n }_p(V^\dag P V)}_\mathrm{F}^2
    \\ = &\norm{V^\dag O V}_{\mathrm{F}}^2 = \norm{ O }_{\mathrm{F}}^2,
\end{align}    
where we used the orthogonality property of Pauli terms with different supports (Lemma\ \ref{lem:ortho1}) and the inequality follows from the definition of mean squared contraction coefficient (Definition~\ref{def:coeff}).
\end{proof}

\section{Noisy circuit model}

We introduce our noisy circuit model.

\begin{definition}[Noisy Circuit Model]
\label{def:circuitmod}
A noisy $L$-layered quantum circuit \(\calC\) consists of \(n\)-qubit unitary layers interleaved with local (single-qubit) noise, concluding with a final layer of single-qubit gates.
Formally, we express the circuit as:
\begin{align}
\calC \coloneqq \mathcal{V}^{\mathrm{single}} \circ \mathcal{N}^{\otimes n} \circ \mathcal{U}_{L} \circ \cdots \circ \mathcal{N}^{\otimes n} \circ \mathcal{U}_{1} 
\label{eq:randcirc_main}
\end{align}
Let $\calD_{\mathrm{single}}$ a 1-design over $\mathbb{U}(2)$. We make the following assumptions:
\begin{enumerate}
    \item \(\mathcal{N}\) is a single-qubit quantum channel.
    \item \(\mathcal{V}^{\mathrm{single}} \coloneqq V(\cdot)V^{\dagger}\), with \(V \coloneqq \bigotimes_{i=1}^{n} V_i\), represents a layer of single-qubit gates. We assume that $V$ is sampled from a distribution $\calD_{L+1}.$
    \item For \( i = 1, 2, \dots, L \), \(\mathcal{U}_i = U_i(\cdot)U_i^\dag\) represents the quantum channel associated with the \(i\)-th unitary circuit layer. We assume that \({U}_i\) is sampled from a distribution \(\mathcal{D}_i\) and that it is composed of non-overlapping gates, each acting on \(\mathcal{O}(1)\)-qubits.
    \item For all $i =1,2,\dots, L+1$, the distribution $\calD_i$ satisfies the following
    \begin{align}
        \bbE_{U\sim \calD_i} \left[U\otimes U^*\right] = \bbE_{U\sim\calD_i'} \, \bbE_{W\sim\calD_{\mathrm{single}}^{\otimes n}}\left[U W\otimes U^*W^* \right] 
        \label{eq:d-single}
    \end{align}
\end{enumerate}

We denote by \(\calD_{\mathrm{circ}}\) the overall distribution of the noisy circuits described above.

\end{definition}

It is worth noting that the final unitary layer in the circuit is not essential. Moreover, if the circuit were to end with an additional layer of local noise \(\mathcal{N}^{\otimes n}\), this noise could be absorbed into the measured observable, effectively replacing \(O\) with \(\mathcal{N}^{\dagger \otimes n}(O)\).

Moreover, our results remain valid if we replace Eq.~\eqref{eq:d-single} with the following more general expression:  
\begin{align}
    \mathbb{E}_{U \sim \mathcal{D}_i} \left[ U \otimes U^* \right] = \mathbb{E}_{U \sim \mathcal{D}_i'} \, \mathbb{E}_{W \sim \bigotimes_{j=1}^n \mathcal{D}_{\mathrm{single}, i}^{(j)}} \left[ U W \otimes U^* W^* \right],
\end{align}
where \(\mathcal{D}_{\mathrm{single}, i}^{(j)}\) represents 1-designs over \(\mathbb{U}(2)\) for all \(i \in [L+1]\) and \(j \in [n]\). For simplicity and to avoid overly complex notation, we assume a single-qubit distribution \(\mathcal{D}_{\mathrm{single}}\) across all qubits and circuit layers.

Similarly, our results can be easily extended to circuits influenced by noise channels that vary across qubits. Specifically, this involves replacing \(\mathcal{N}^{\otimes n}\) with the more general expression \(\bigotimes_{j=1}^n \mathcal{N}^{(j)}\).

Lemma\ \ref{lem:norm-not-incr} hints to the fact that all local noise channels behave similarly to the depolarizing noise on typical instances of random circuits. Thus, given a distribution over noisy circuits such as the one in Definition~\ref{def:circuitmod}, it is natural to assign to it an \emph{effective depolarizing rate}.

\begin{definition}[Effective depolarizing rate of noisy random circuits]
    Let $\calD_{\mathrm{circ}}$ be the circuit distribution introduced in Definition\ \ref{def:circuitmod}. We define the associated effective depolarizing rate  of $\calD_{\mathrm{circ}}$ as
\begin{align}
    p \coloneqq 1-  {\chi_{\calD_{\mathrm{single}}}(\calN)}.
\end{align}
\end{definition}

\section{Classical simulation of noisy random circuits}
We will provide some additional notation to facilitate the proof of our main result.
For convenience, we also rewrite the noisy circuit as 
\begin{align}
\calC \coloneqq &\mathcal{V}^{\mathrm{single}} \circ \mathcal{N}^{\otimes n} \circ \mathcal{U}_{L} \circ \dots \circ \mathcal{N}^{\otimes n} \circ \mathcal{U}_{1} 
\\=
&\mathcal{V}^{\mathrm{single}} \circ \calN_p^{(\mathrm{depo})\otimes n} \circ \tilde\calN^{\otimes n} \circ \mathcal{U}_{L} \circ \cdots \circ \calN_p^{(\mathrm{depo})\otimes n}\circ  \tilde\calN^{\otimes n} \circ \mathcal{U}_{1} 
\\ =
& \calC_L \circ \calC_{L-1} \circ \dots \circ \calC_1,
\end{align}
where we the layers $\calC_j$ are given by
\begin{align}
    \calC_j\coloneqq \begin{cases}
    \mathcal{V}^{\mathrm{single}} \circ \calN_p^{(\mathrm{depo})\otimes n} \circ \tilde\calN^{\otimes n} \circ \mathcal{U}_{L}  & \text{if $j=L$},
    \\ \calN_p^{(\mathrm{depo})\otimes n} \circ \tilde\calN^{\otimes n} \circ \mathcal{U}_{j}  & \text{if $j\neq L$}.
    \end{cases}
\end{align}

We also consider the (non-physical) linear map obtained by ``taking out'' the depolarizing channels:
\begin{align}
\tilde\calC\coloneqq &
\mathcal{V}^{\mathrm{single}} \circ \tilde\calN^{\otimes n}_p \circ \mathcal{U}_{L} \circ \cdots  \circ  \tilde\calN^{\otimes n}_p \circ \mathcal{U}_{1} \label{def:non-phys}
\\ =
& \tilde\calC_L \circ \tilde\calC_{L-1} \circ \dots \circ \tilde\calC_1,
\end{align}
where we defined the layers $\tilde\calC_j$ as
\begin{align}
    \tilde\calC_j\coloneqq \begin{cases}
    \mathcal{V}^{\mathrm{single}} \circ \tilde\calN^{\otimes n}_p \circ \mathcal{U}_{L}  & \text{if $j=L$},
    \\ \tilde\calN^{\otimes n} \circ \mathcal{U}_{j}  & \text{if $j\neq L$}.
    \end{cases}
\end{align}
Given an observable $O$ and an initial state $\rho$, we aim to approximate the expectation value 
\begin{align}
    \Tr[O \calC(\rho)] = \Tr[\calC^\dag(O)\rho].
\end{align}
To this end, we fix a positive value $k$ and we approximate the Heisenberg evolved observable $\calC^\dag(O)$ by computing only the Pauli paths with total weight smaller than $k$:
\begin{align}
    O_{\calC}^{(k)} \coloneqq \sum_{\substack{\gamma = (P_0, P_1, \dots, P_L) \in \calP_n^{L+1}:\\\abs{\gamma}<k}} \, \llangle O | P_L \rrangle \Phi_\gamma(\calC) P_0,
\end{align}
where we recall that the Fourier coefficients $\Phi_\gamma(\calC)$ are defined as
\begin{align}
    \Phi_\gamma(\calC) = \llangle P_L| \widehat\calC_L | P_{L-1}\rrangle
    \llangle P_{L-1}| \widehat\calC_{L-1} | P_{L-2}\rrangle \dots \llangle P_1| \widehat\calC_1 | P_{0}\rrangle.
\end{align}
We also consider the observable $\tilde\calC^\dag(O)$ obtained by Heisenberg-evolving $O$ with the linear map $\tilde\calC$:
\begin{align}
    \tilde\calC^\dag(O) = \sum_{\substack{\gamma = (P_0, P_1, \dots, P_L) \in \calP_n^{L+1}}} \, \llangle O | P_L \rrangle \Phi_\gamma(\tilde\calC) P_0,
\end{align}
where the Fourier coefficients $ \Phi_\gamma(\tilde\calC)$ are defined as
\begin{align}
    \Phi_\gamma(\tilde\calC) = \llangle P_L| \widehat{\tilde\calC}_L | P_{L-1}\rrangle
    \llangle P_{L-1}| \widehat{\tilde\calC}_{L-1} | P_{L-2}\rrangle \dots \llangle P_1| \widehat{\tilde\calC}_1 | P_{0}\rrangle.
\end{align}
\subsection{Average contraction implies efficient classical simulation}
We start by proving that the Fourier coefficients of $\calC$ and $\tilde\calC$ are connected by a proportionality factor, which is exponential in the path weight.
\begin{lemma}[Proportionality of Fourier coefficients]
\label{lem:path-supp-dep}
Let $\gamma \in \calP_n^{L+1}$ be a Pauli path. Then, the associated Fourier coefficients $\Phi_{\gamma}(\calC)$ and $\Phi_{\gamma}(\tilde\calC)$ satisfy the  identity
\begin{align}
    \Phi_{\gamma}(\calC) = (1-p)^{\abs{\gamma}} \Phi_{\gamma}(\tilde\calC).
\end{align}

\end{lemma}
\begin{proof}
    The lemma follows from the fact that $\calC_j = \calN_p^{(\mathrm{depo})\otimes n}\circ \tilde\calC_j$ for all $j\in [L]$. Therefore,
    \begin{align}
        \llangle P_j| \widehat{\calC}_j | P_{j-1}\rrangle =
         \llangle \calN_p^{(\mathrm{depo})\otimes n}(P_j)| \widehat{\tilde\calC}_j | P_{j-1}\rrangle  =(1-p)^{\abs{P_j}} \llangle P_j| \widehat{\tilde\calC}_j | P_{j-1}\rrangle. 
    \end{align}
    Iterating over al $j\in[L]$ we obtain the desired result:
    \begin{align}
        \Phi_{\gamma}(\calC) = &\llangle P_L| \widehat\calC_L | P_{L-1}\rrangle
    \llangle P_{L-1}| \widehat\calC_{L-1} | P_{L-2}\rrangle \dots \llangle P_1| \widehat\calC_1 | P_{0}\rrangle \\=&
     (1-p)^{\abs{P_L}} \llangle P_L| \widehat{\tilde\calC}_L | P_{L-1}\rrangle
    (1-p)^{\abs{P_{L-1}}} \llangle P_{L-1}| \widehat{\tilde\calC}_{L-1} | P_{L-2}\rrangle \dots (1-p)^{\abs{P_1}} \llangle P_1| \widehat{\tilde\calC}_1 | P_{0}\rrangle
    \\ =&
    (1-p)^{\abs{P_L}+\abs{P_{L-1}} + \dots + \abs{P_1}} \Phi_{\gamma}(\tilde\calC)
        \coloneqq (1-p)^{\abs{\gamma}} \Phi_{\gamma}(\tilde\calC).
    \end{align}
\end{proof}
Leveraging the definition of effective depolarizing rate of random noisy circuits, we show that, on average over the choice of random unitaries, the linear map $\tilde{\calC}$ does not increase the Frobenius norm under Heisenberg evolution.
\begin{corollary}[Non-increasing Frobenius norm]
\label{cor:norm-not-incr}
It holds that
 \begin{align}
     \bbE_{\calC\sim\calD_{\mathrm{circ}}}\norm{\tilde\calC^\dag(O)}^2_{\mathrm{F}}\leq \norm{O}^2_{\mathrm{F}}.
 \end{align}
\end{corollary}
\begin{proof}
We start by making the following preliminary observation:
\begin{align}
\bbE_{V\sim\calD_{\mathrm{single} }^{\otimes n}} \norm{\calU_j\circ \tilde\calN_p^{\otimes n}(V^\dag OV)}^2_{\mathrm{F}}  
    = \bbE_{V\sim\calD_{\mathrm{single} }^{\otimes n}} \norm{\tilde\calN_p^{\otimes n}(V^\dag OV)}^2_{\mathrm{F}} \leq \norm{O}^2_{\mathrm{F}}, \label{eq:contr-layer}
\end{align}
where the first step follows from the unitarily invariance of Schatten norms, and the second step follows applying Lemma\ \ref{lem:norm-not-incr} and noting that the effective depolarizing rate is defined as $p\coloneqq 1 - \chi_{\calD_{\mathrm{single}}}(\calN)$.
Employing Eq.\ \ref{def:non-phys}, we can write
\begin{align}
&\bbE_{\calC \sim \calD_{\mathrm{circ}}} \norm{\tilde\calC^\dag(O)}^2_{\mathrm{F}} =
    \bbE_{\calC \sim \calD_{\mathrm{circ}}} \norm{\tilde\calC_{1}^\dag\circ \dots \circ \tilde\calC_{L-1}^\dag\circ \tilde\calC_L^\dag(O)}^2_{\mathrm{F}}
   \\=
    &\bbE_{U_1 \sim \calD_1,\dots, U_L \sim \calD_L, V \sim \calD_{\mathrm{circ}}^{\otimes n}}
    \norm{\mathcal{U}_{1}^\dag\circ \tilde\calN_p^{\dag\otimes n} \circ \cdots \circ  \tilde\calN_p^{\dag\otimes n} \circ \mathcal{U}_{L}^\dag \circ \tilde\calN_p^{\dag\otimes n} \circ \mathcal{V}^{\mathrm{single}\dag}(O)}^2_{\mathrm{F}},
\end{align}
where we denoted $\calU_j \coloneqq U_j(\cdot)U_j^\dag$ and $\mathcal{V}^{\mathrm{single}} \coloneqq V(\cdot)V^\dag$.
Using Eq.\ \ref{eq:d-single}, we can insert independently sampled layers of random single-qubit gates within the above expression:
\begin{align}
    &\bbE_{U_1 \sim \calD_1,\dots, U_L \sim \calD_L, V \sim \calD_{\mathrm{circ}}^{\otimes n}}
    \norm{\mathcal{U}_{1}^\dag\circ \tilde\calN_p^{\dag\otimes n} \circ \cdots \circ  \tilde\calN_p^{\dag\otimes n} \circ \mathcal{U}_{L}^\dag \circ \tilde\calN_p^{\dag\otimes n} \circ \mathcal{V}^{\mathrm{single}\dag}(O)}^2_{\mathrm{F}}
    \\ = &\bbE_{U_1 \sim \calD'_1,\dots, U_L \sim \calD'_L} \mathbb{E}_{V_1, V_2, \dots, V_L \sim \calD_{\mathrm{circ}}^{\otimes n}}
    \norm{\mathcal{U}_{1}^\dag\circ \tilde\calN_p^{\dag\otimes n} \circ  \mathcal{V}_1^{\mathrm{single}\dag} \circ \cdots \circ  \tilde\calN_p^{\dag\otimes n} \circ \mathcal{V}_{L-1}^{\mathrm{single}\dag} \circ \mathcal{U}_{L}^\dag \circ \tilde\calN_p^{\dag\otimes n} \circ \mathcal{V}_L^{\mathrm{single}\dag}(O)}^2_{\mathrm{F}} \\\leq &\norm{O}^2_{\mathrm{F}},
\end{align}
where the final inequality follows by applying iteratively Eq.\ \ref{eq:contr-layer}.
\end{proof}
The core result follows by combining Lemma\ \ref{lem:path-supp-dep} and Corollary\ \ref{cor:norm-not-incr} with the orthogonality properties of different Pauli paths.

\begin{theorem}[Mean squared error]
\label{thm:mse-core}
Let $\calD_{\mathrm{circ}}$ be an $L$-layered locally unbiased distribution over noisy circuits, and let $p$ be the effective depolarizing rate of $\calD_{\mathrm{circ}}$.
Let $O_{\calC}^{(k)} \coloneqq \sum_{\substack{\gamma :\abs{\gamma}<k}} \, \llangle O | P_L \rrangle \Phi_\gamma(\calC) P_0$ be the observable obtained by computing the approximate evolution of $O$ under $\calC^\dag$ restricted to the Pauli paths of weight less than $k$.
We have
\begin{align}
   \bbE_{\calC\sim\calD_{\mathrm{circ}}} \norm{O_{\calC}^{(k)}  - \calC^\emph\dag(O)}_{\mathrm{F}}^2 \leq (1-p)^{2k} 
   \norm{O}_{\mathrm{F}}^2.
\end{align}

Moreover, for any initial state $\rho$ we have
\begin{align}
   \bbE_{\calC\sim\calD_{\mathrm{circ}}} \Tr[\left(O_{\calC}^{(k)} - \calC^\emph\dag(O)\right)\rho]^2
   \leq (1-p)^{2k} \norm{O}_{\mathrm{F}}^2.
\end{align}
\end{theorem}
\begin{proof}
We start by proving the following technical claim, which implies that the contribution of paths with different weight to the mean squared error can be evaluated separately.
\begin{claim}
\label{claim:mse-weight}
For all subsets of Pauli paths $\calS \subseteq \calP_n^{L+1}$, we have
\begin{align}
    \bbE_{\calC\sim\calD_{\mathrm{circ}}} \left(\sum_{\substack{\gamma \in \calS}} \llangle O | P_{L+1}\rrangle  \Phi_\gamma(\tilde\calC) P_0 \right)^{\otimes 2}
    = \bbE_{\calC\sim\calD_{\mathrm{circ}}} \sum_{w=1}^n \left(\sum_{\substack{\gamma \in \calS:\\ \abs{\gamma}= w}} \llangle O | P_{L+1}\rrangle  \Phi_\gamma(\tilde\calC) P_0 \right)^{\otimes 2}\,.
\end{align}

\end{claim}

\begin{proof}[Proof of the Claim]
We can prove the result by exploiting the fact that paths with different support are orthogonal, i.e., their Fourier coefficients are uncorrelated. As such,
\begin{align}
        &\bbE_{\calC\sim\calD_{\mathrm{circ}}} \left(\sum_{\substack{\gamma \in \calS}} \llangle O | P_{L+1}\rrangle  \Phi_\gamma(\tilde\calC) P_0 \right)^{\otimes 2}
    = \bbE_{\calC\sim\calD_{\mathrm{circ}}}  \left(\sum_{w=1}^n\sum_{\substack{\gamma \in \calS:\\ \abs{\gamma}= w}} \llangle O | P_{L+1}\rrangle  \Phi_\gamma(\tilde\calC) P_0 \right)^{\otimes 2}
    \\ = & \sum_{w,w'=1}^n \sum_{\gamma: \abs{\gamma}= w} \sum_{\gamma': \abs{\gamma'}= w'} \llangle O | P_{L+1}\rrangle \llangle O | P'_{L+1}\rrangle \underbrace{\bbE_{\calC\sim\calD_{\mathrm{circ}}} \Phi_\gamma(\tilde\calC)  \Phi_{\gamma'}(\tilde\calC)}_{=0 \text{ if } w\neq w'} P_0 \otimes P_0'
    \\ =  &\sum_{w=1}^n \sum_{\gamma: \abs{\gamma}= w} \sum_{\gamma': \abs{\gamma'}= w} \llangle O | P_{L+1}\rrangle \llangle O | P'_{L+1}\rrangle {\bbE_{\calC\sim\calD_{\mathrm{circ}}} \Phi_\gamma(\tilde\calC)  \Phi_{\gamma'}(\tilde\calC)} P_0 \otimes P_0'
    \\ =&\bbE_{\calC\sim\calD_{\mathrm{circ}}} \sum_{w=1}^n \left(\sum_{\substack{\gamma \in \calS:\\ \abs{\gamma}= w}} \llangle O | P_{L+1}\rrangle  \Phi_\gamma(\tilde\calC) P_0 \right)^{\otimes 2},
\end{align}
where in the third step we invoked Lemma~\ref{lem:ortho2}.

\end{proof}

The desired results can be proven by evaluating separately the contribution of paths with different weight and using the fact that the Fourier coefficients of $\calC$ and those of $\tilde\calC$ coincide, up to a proportionality factor which is exponential in the path weight.
As for the first statement, we have
\begin{align}
   &\bbE_{\calC\sim\calD_{\mathrm{circ}}} \norm{O_{\calC}^{(k)}  - \calC^\dag(O)}_{\mathrm{F}}^2
   = \bbE_{\calC \sim\calD_{\mathrm{circ}} } \norm{\sum_{\gamma: \abs{\gamma}\geq k} \llangle O | P_{L+1}\rrangle  \Phi_\gamma(\calC) P_0}_{\mathrm{F}}^2
   \\=&\bbE_{\calC \sim\calD_{\mathrm{circ}}} \norm{\sum_{\gamma: \abs{\gamma}\geq k} (1-p)^{\abs{\gamma}}\llangle O | P_{L+1}\rrangle  \Phi_\gamma(\tilde\calC) P_0}_{\mathrm{F}}^2
   \\=&\frac{1}{2^n}\bbE_{\calC \sim\calD_{\mathrm{circ}} }  \Tr[\mathbb{F}\left(\sum_{\gamma: \abs{\gamma}\geq k} (1-p)^{\abs{\gamma}}\llangle O | P_{L+1}\rrangle  \Phi_\gamma(\tilde\calC) P_\gamma\right)^{\otimes 2}]
   \\=&\frac{1}{2^n} \sum_{w\geq k}(1-p)^{2w}\bbE_{\calC \sim\calD_{\mathrm{circ}} } \Tr[\mathbb{F}\left(\sum_{\gamma: \abs{\gamma} = w} \llangle O | P_{L+1}\rrangle  \Phi_\gamma(\tilde\calC) P_\gamma\right)^{\otimes 2}]
   \\\leq &(1-p)^{2k } \frac{1}{2^n} \sum_{w \geq k}\bbE_{\calC \sim\calD_{\mathrm{circ}} } \Tr[\mathbb{F}\left(\sum_{\gamma: \abs{\gamma} = w} \llangle O | P_{L+1}\rrangle  \Phi_\gamma(\tilde\calC) P_\gamma\right)^{\otimes 2}]
   \\\leq &(1-p)^{2k } \frac{1}{2^n} \sum_{w \geq 0}\bbE_{\calC \sim\calD_{\mathrm{circ}} } \Tr[\mathbb{F}\left(\sum_{\gamma: \abs{\gamma} = w} \llangle O | P_{L+1}\rrangle  \Phi_\gamma(\tilde\calC) P_\gamma\right)^{\otimes 2}]
   \\= & (1-p)^{2k} \bbE_{\calC \sim\calD_{\mathrm{circ}} } \norm{\sum_{\gamma: \abs{\gamma}\in \calP_n^{L+2}} \llangle O | P_{L+1}\rrangle  \Phi_\gamma(\tilde\calC) P_\gamma}_{\mathrm{F}}^2
   = (1-p)^{2k} \bbE_{\calC \sim\calD_{\mathrm{circ}} }  \norm{\tilde\calC^\dag(O)}_{\mathrm{F}}^2 \leq (1-p)^{2k} \norm{O}_{\mathrm{F}}^2.
\end{align}
In the second step, we applied Lemma~\ref{lem:path-supp-dep}, while in the fourth step, we invoked Claim~\ref{claim:mse-weight}. The first inequality follows from the observation that the sum is taken over indices \( w \geq k \), implying \( (1-p)^w \leq (1-p)^k \). The second inequality holds because all the additional terms included in the sum are positive. Finally, the last inequality is a direct consequence of Corollary~\ref{cor:norm-not-incr}.

\medskip

The second statement can be shown as follows
\begin{align}
   &\bbE_{\calC\sim\calD_{\mathrm{circ}}} \Tr[\left(O_{\calC}^{(k)} -\calC^\dag(O)\right)\rho] ^2 = 
   \bbE_{\calC\sim\calD_{\mathrm{circ}} , V\sim \mathcal{D}^{\otimes n}_{\mathrm{single}}}\Tr[\left(O_{\calC}^{(k)} -\calC^\dag(O)\right) (V\rho V^\dag)] ^2 
    \\= &\bbE_{\calC\sim\calD_{\mathrm{circ}}}  \norm{O_{\calC}^{(k)} -\calC^\dag(O)}_\mathrm{F}^2 \leq (1-p)^{2k}\norm{O}_\mathrm{F}^2,
\end{align}
where the inequality follows from Lemma~\ref{lem:schuster} and from the fact that $\calD^{\otimes n}_{\mathrm{single}}$ is a 1-design, which implies that $V\rho V^\dag$ is sampled from a low-average ensemble with purity 1.
\end{proof}

It remains to show that the weight-truncated Pauli path summation can be evaluated efficiently. This has been already proven in Ref.~\cite{aharonov2022polynomial} for circuits interspersed by depolarizing noise. 
While for unital noise the same proof can be adapted straightforwardly, the upper bound found in Ref.~\cite{aharonov2022polynomial} does not hold anymore for the case of non-unital noise, due to the additional Pauli paths created by the noisy layers.
Here, we provide a modified argument which holds under any local noise, provided that the observable $O$ is a linear combination of polynomially many Pauli operators.

We start by providing the definition of \emph{legal} Pauli path, which are the Fourier coefficients that bear non-zero contributions.

\begin{definition}
We say that a Pauli path $\gamma = (P_0, P_1,\dots, P_{L+1})$ is \emph{legal} if its Fourier coefficient gives a non-zero contribution, i.e., if it satisfies $\llangle O | P_{L+1}\rrangle\bbE_{\calC\sim\calD_{\mathrm{circ}}}\left[\Phi_\gamma^2(\calC)\right] \neq 0$.
\end{definition}
If the noise channel $\calN$ is unital, then the Fourier coefficients of all the legal Pauli paths can be efficiently enumerated.
In particular, the proof of the unital case coincides with that of Lemma 8 in Ref.\ \cite{aharonov2022polynomial} for circuits with 2-qubit gates interleaved by local depolarizing noise. To see this, it suffices to notice that any unital single-qubit noise channel can be rewritten as $\calN = U \calN'(V(\cdot) V^\dag)U^\dag$, where 
\begin{align}
    P\neq Q \implies \llangle P | \widehat{\calN'}| Q  \rrangle = 0, 
\end{align}
for $P,Q\in\calP_1$, i.e., the Pauli Transfer Matrix of $\calN'$ is diagonal. Therefore, $\calN'$ does not create additional Pauli paths, and thus it is covered by same proof of Ref.~\cite{aharonov2022polynomial}.

\begin{lemma}[Runtime for the unital case, adapted from Lemma 8 in Ref.\ \cite{aharonov2022polynomial}]
\label{lem:unital-runtime}
If the noise is unital, then the number of legal Pauli paths with weight at most $k$ is at most $n^{k/L} \cdot \exp(\calO(k))$. Furthermore, there is an efficient algorithm to enumerate the legal paths in time $n^{k/L} \cdot \exp(\calO(k))$. 
\end{lemma}

However, the result discussed above does not hold if the channel $\calN$ is non-unital. The following observation provides an easy counter-example.

\begin{observation}
Let $\calC$ be a quantum circuit interspersed by amplitude damping noise and let $O$ be a $k$-local Hamiltonian such as
\begin{align}
    O = \sum_{\substack{P\in\{I,Z\}^{\otimes n} : \abs{P}<k}}P.
\end{align} 
Assume for the sake of simplicity that the circuit ends with a layer of Pauli gates. 
Then there are at least  $\Omega(n^{k-1})$ legal Pauli paths.
\end{observation}

Crucially, projectors over computational basis states, e.g., $\ketbra{0^n}{0^n} = 2^{-n} \sum_{P\in\{I,Z\}^{\otimes n}}P$ constitute another possible counter-example.
These observations motivate the adoption of a different strategy for circuits with non-unital noise, which we formalize in the following Lemma.  

\begin{lemma}
[Runtime for arbitrary noise]
\label{lem:runtime-non-u}
Assume that each unitary layer $U_j$ consists of non-overlapping 1 and 2-qubit gates. 
For any noise, if the observable $O$ contains $M$ different Pauli terms, then the number of legal Pauli paths with weight at most $k$ is at most $M \exp(\calO(k))$. Furthermore, there is an efficient algorithm to enumerate the legal paths in time $M n \exp(\calO(k))$.
\end{lemma}

\begin{proof}

We will enumerate all the legal Pauli paths of the form $\gamma = (P_0, P_1,\dots, P_{L})$ with weight smaller than $k$ with an iterative algorithm.
First, we notice that these paths contain at most $k$ non-identity terms otherwise their path weights would exceed $k$ (we recall that  $P_0$ does not contribute to the Pauli weight).
Moreover, such non-identity terms need to be consecutive, since the adjoint of any quantum channel is unital.
Thus, if $L>k$, then the legal paths $\gamma$ with weight smaller than $k$ it holds that
\begin{align}
    P_0 = P_1 = \dots = P_{L-k+1} = I^{\otimes n}.
\end{align}
Therefore, we can assume without loss of generality that $L\leq k$. 

\smallskip

The enumeration algorithm corresponds to a Breadth-First Search (BFS) of the sub-tree constituted by the legal paths with weight less than $k$. Specifically, it works by iterative forming some \emph{partial} Pauli paths, which are vectors of the form
\begin{align}
    (\star, \star,\dots, \star, P_j, P_{j+1},\dots, P_{L}),
\end{align}
where the symbol $\star$ denotes an empty entry.
\begin{enumerate}
    \item Let $\mathbb{S}_t$ be a set of partial Pauli paths at time $t$. We set $t = L$ at the beginning of the algorithm, and we decrease the counter by one at each iteration until $0$. 
    \item We initialize the set $\mathbb{S}_{L}$ as follows
    \begin{align}
        \mathbb{S}_{L} \leftarrow \left\{(\star, \star,\dots, \star, P_{L}) \} : \llangle P_{L} | O \rrangle \neq 0 \text{ and } \abs{P_{L}} < k \right\}
    \end{align}
    \item At iteration $0 < j < L$, for all partial paths $(\star, \star,\dots, \star, P_{j+1},\dots, P_{L}) \in \mathbb{S}_{j+1}$ we compute the Heisenberg evolved observable $\calC_j^\dag(P_{j+1})$ and we add the following paths to the set $\mathbb{S}_j$
    \begin{align}
        \mathbb{S}_j \leftarrow \left\{(\star, \dots,  \star,P_{j}, P_{j+1}, \dots, P_{L}) \} : \llangle P_{j+1} | \widehat{\calC_j} |P_j \rrangle \neq 0 \text{ and } \abs{P_j} + \abs{P_{j+1}} +\dots \abs{P_{L}} < k \right\}.
    \end{align}
    \item At iteration $j = 0$, for all partial paths $(\star, P_2,\dots, P_{L}) \in \mathbb{S}_{1}$ we compute the Heisenberg evolved observable $\calC_1^\dag(P_{1})$ and we add the following paths to the set $\mathbb{S}_0$
    \begin{align}
        \mathbb{S}_0 \leftarrow \left\{(P_0, P_1, \dots, P_{L}) \} : \llangle P_{1} | \widehat{\calC_1} |P_0 \rrangle \neq 0 \right\}.
    \end{align}
\end{enumerate}

For any fixed index $j$, we can upper bound the number of partial Pauli paths $(\star, \dots,\star, P_j,\dots, P_L)$ contained in $\mathbb{S}_j$ with the following counting argument.

\begin{enumerate}
    \item We assign to the Pauli operators $P_j, P_{j+1},\dots, P_L$ their Pauli weights $\ell_j, \ell_{j+1},\dots, \ell_L$.
    The number of possible choices is at most
    \begin{align}
         &\sum_{w=0}^{k} \left\{\text{ number of solutions to the equation } \sum_{i=j}^L  \abs{\ell_{i}} =w  \right\}
         \\= &\sum_{w=0}^{k}  \binom{w+L-j}{L-j} \leq (k+1) 2^{k+L} \in \exp(\calO(k)).
    \end{align}
    where in the first step we used the fact that the number of solutions to the equation $\sum_{i=1}^t x_i = m$ under the constraint $x_i \leq 0$ equals $\binom{m+t-1}{t-1} \leq 2^{m+t-1}$, and the last step follows because we can assume $L\leq k$ without loss of generality.

    \item We choose the support (positions of identities and non-identities) for each Pauli operators $P_j, P_{j+1},\dots, P_L$. 
    We recall that $P_L$ satisfies $\llangle P_{L} | O \rrangle \neq 0$ and $\abs{P_L}<k$, i.e., $P_L$ is in the Pauli decomposition $O$ and has Pauli weight less than $k$.
     Thus, $P_{L}$ can take $\calO(\min \{ n^k, M\})$ different values. 
    For any given value of $P_{i+1}$, we can upper bound the number of possible choices of $P_i$ by light-cone argument. As $U_{i+1}$ contains only non-overlapping gates, each acting on $c\in\calO(1)$ qubits, and $P_{i+1}$ has weight $\ell_{i+1}$, then $P_i$ is supported on a set of $c^{\ell_{i+1}}$  qubits. Those qubits can take values $\{I,X,Y,Z\}$, which yields $\left(c^{\ell_{i+1}}\right)^4 = c^{4\ell_{i+1}}$ different choices.
    Overall, the number of possible choices for fixed values of $\ell_j, \ell_{j+1},\dots, \ell_{L}$ is at most
    \begin{align}
        \prod_{i=j}^L c^{4\ell_{i+1}}= c^{4\sum_{i=j}^L \ell_{i+1}} \leq c^{4k} \in \exp(\calO(k)).
    \end{align}
    We emphasize that this argument relies solely on the fact that a layer is comprised of non-overlapping $\calO(1)$-qubit gates, and does not require the circuit to be geometrically local.
    \item For each fixed choice of the supports, we assign a value a value among $X,Y,Z$ to each non-identity single-qubit operator. The number of possible choices is at most  
    \begin{align}
        \prod_{i=j}^L 3^{\ell_{i+1}} = 3^{\sum_{i=j}^L \ell_{i+1}} \leq 3^{k} \in \exp(\calO(k)).
    \end{align}
\end{enumerate}
Thus $\mathbb{S}_j$ contains at most $M\,\mathrm{exp}(\calO(k))$ elements. As $\calC_j$ contains only non-overlapping $\calO(1)$-qubit channels,  all the required  transition amplitudes $\llangle P_{j+1} | \widehat{\calC_j} |P_j \rrangle$ for $\abs{P_{j+1}} \leq k$ can be computed in time $\exp(\calO(k))$.
Moreover, the algorithm needs to store in memory each $n$-qubit Pauli operator, leading to an additional factor $\calO(n)$.
Putting all together, we find that the runtime is at most
\begin{align}
    M\, n\,  \exp(\calO(k)),
\end{align}

\end{proof}

Combining the above results, we obtain the following upper bound on the time complexity.

\begin{theorem}[Time complexity]
Let $\calC$ be a noisy quantum circuit sampled from the distribution $\calD_{\mathrm{circ}}$ with effective depolarizing rate $p \in \Theta(1)$
\begin{itemize}
    \item If the measured observable $O$ is expressed as a linear combination of $M$ Pauli operators, then there exists a classical algorithm running in time $M\cdot\mathrm{poly}(n, 1/\epsilon, 1/\delta)$ that approximates $\Tr[O\calC(\rho)]$ within additive error $\epsilon\norm{O}_\mathrm{F}$ with probability at least $1-\delta$ over the choice of $\calC$.
    \item If the noise channel $\calN$ is unital and the depth $L$ is at least logarithmic in the number of qubits $n$, then there exists a classical algorithm running in time $\mathrm{poly}(n)$ that approximates $\Tr[O\calC(\rho)]$ within additive error $\epsilon\norm{O}_\mathrm{F}$ with probability at least $1-\delta$ over the choice of $\calC$,  for arbitrary $\epsilon,\delta$ decaying inverse-polynomially in $n$.
\end{itemize}  
\end{theorem}

\begin{proof}
By Theorem~\ref{thm:mse-core}, the approximate Heisenberg evolved observable $O_\calC^{(k)}$ satisfies the following:
 \begin{align}
     {\mathbb{E}_{\calC \sim \calD_{\mathrm{circ}}} \Tr[\left(O^{(k)}_\calC - \calC^\dag(O)\right)\rho]^2} \leq (1-p)^{2k} \norm{O}^2_{\mathrm{F}}.
 \end{align}
Thus, Markov's inequality yields
 \begin{align}
    &\Pr_{\calC\sim\calD_{\mathrm{circ}}}\left\{\left\lvert\Tr[\left(O^{(k)}_\calC - \calC^\dag(O)\right)\rho]\right\rvert \geq \epsilon\norm{O}_{\mathrm{F}}\right\}
    =
     \Pr_{\calC\sim\calD_{\mathrm{circ}}}\left\{\Tr[\left(O^{(k)}_\calC - \calC^\dag(O)\right)\rho]^2 \geq \epsilon^2\norm{O}^2_{\mathrm{F}}\right\}
     \\\leq &\frac{{\mathbb{E}_{\calC \sim \calD_{\mathrm{circ}}} \Tr[\left(O^{(k)}_\calC - \calC^\dag(O)\right)\rho]^2} }{\norm{O}^2_{\mathrm{F}}\epsilon^2}\leq  \frac{(1-p)^{2k} }{\epsilon^2}
 \end{align}
Choosing $k = \lceil \log\left(\delta^{1/2} \epsilon\right)/\log(1-p)\rceil  \in \calO(\log\left(\delta^{-1}\epsilon^{-1})\right)$, we obtain the desired failure probability $\delta$, i.e.
 \begin{align}
    &\Pr_{\calC\sim\calD_{\mathrm{circ}}}\left\{\left\lvert\Tr[\left(O^{(k)}_\calC - \calC^\dag(O)\right)\rho]\right\rvert \geq \epsilon\right\} \leq \delta.
 \end{align}
 By Lemma\ \ref{lem:runtime-non-u}, the runtime of the Pauli Propagation algorithm with path-weight cutoff $k \in \calO(\log\left(\delta^{-1}\epsilon^{-1})\right)$ is
 \begin{align}
     Mn \exp\left(\calO(k)\right) = M n \mathrm{poly}(\epsilon^{-1}, \delta^{-1}).
 \end{align}
 Moreover, assuming that $L \in \Omega(\log(n))$ and the noise channel $\calN$ is unital, then by Lemma\ \ref{lem:unital-runtime}  we can achieve $\epsilon \norm{O}_\mathrm{F}$ precision and with probability $\delta$ in polynomial time, for arbitrary $\epsilon,\delta$ decaying inverse-polynomially in $n$.
\end{proof}

\section{Effective depth of noisy random circuits beyond local $2$-designs}
We conclude our theoretical investigation by extending the main result of Ref.~\cite{mele2024noise} to the broad classes of circuits studied in this work. Specifically, we demonstrate that, on average over the choice of random unitaries, the Frobenius norm of the traceless components of an observable undergoes exponential suppression with circuit depth under noisy Heisenberg evolution.

\begin{theorem}
\label{thm:ed}
For all observables $O$ it holds that
    \begin{align}
        \bbE_{\calC\sim\calD_{\mathrm{circ}}}\left(\norm{\calC^\dag(O)}_\mathrm{F}^2 - 4^{-n}\Tr[\calC^\dag(O)]\right) \leq (1-p)^{2L} \norm{O}_\mathrm{F}^2.
    \end{align}
    Moreover, given two quantum states $\rho,\sigma$, we have
    \begin{align}
        \bbE_{\calC\sim\calD_{\mathrm{circ}}} \Tr[O\calC(\rho-\sigma)]^2 \leq 4(1-p)^{2L} \norm{O}_\mathrm{F}^2.
    \end{align}
\end{theorem}

\begin{proof}
We will exploit the fact that $\calC  = \bigcirc_{j=1}^L \calC_j = \bigcirc_{j=1}^L \calN_p^{(\mathrm{depo})\otimes n}\circ \tilde\calC_j $.
Therefore, for all observables $O=\sum_{P\in\calP_n} a_P P$ we have 
\begin{align}
\bbE_{\calC_j}\left(\norm{\calC_j^\dag(O)}_\mathrm{F}^2 - 4^{-n}\Tr[\calC_j^\dag(O)]\right)
= &\bbE_{\calC_j} \sum_{w=1}^n\norm{ \tilde\calC^\dag_j\circ\calN_p^{(\mathrm{depo})\otimes n}\left(\sum_{P:\abs{P}=w} a_P P\right) }^2_\mathrm{F}
\\ =&\bbE_{\calC_j} \sum_{w=1}^n (1-p)^{2w} \norm{ \tilde\calC^\dag_j\left(\sum_{P:\abs{P}=w} a_P P\right) }^2_\mathrm{F}
\\ \leq & (1-p)^{2}\left(\norm{O}_\mathrm{F}^2 - 4^{-n}\Tr[O]\right),
\end{align}
where the inequality follows from Lemma~\ref{lem:norm-not-incr}. Iterating over all $j$ we obtain the desired result.

The we can prove the second part of the Theorem as follows:
\begin{align}
    &\bbE_{\calC\sim\calD_{\mathrm{circ}}} \Tr[O\calC(\rho-\sigma)]^2 
    = \bbE_{\calC\sim\calD_{\mathrm{circ}}} \Tr\left[\left(O - \Tr[O]\frac{I^{\otimes n}}{2^n}\right)\calC(\rho-\sigma)\right]^2 
    \\&\leq 2 \bbE_{\calC\sim\calD_{\mathrm{circ}}} \Tr\left[\left(O - \Tr[O]\frac{I^{\otimes n}}{2^n}\right)\calC(\rho)\right]^2 
    + 2 \bbE_{\calC\sim\calD_{\mathrm{circ}}} \Tr\left[\left(O - \Tr[O]\frac{I^{\otimes n}}{2^n}\right)\calC(\sigma)\right]^2 
    \\& \leq 4\bbE_{\calC\sim\calD_{\mathrm{circ}}}\left(\norm{\calC^\dag(O)}_\mathrm{F}^2 - 4^{-n}\Tr[\calC^\dag(O)]\right) 
    \\& \leq 4(1-p)^{2}\left(\norm{O}_\mathrm{F}^2 - 4^{-n}\Tr[O]\right),
\end{align}
where the second inequality follows by Lemma~\ref{lem:schuster} 
\end{proof}

We also obtain the following corollary for the trace distance.
\begin{corollary}
Given two quantum states $\rho,\sigma$, we have  
\begin{align}
    \bbE_{\calC\sim\calD_{\mathrm{circ}}} \norm{\calC(\rho-\sigma)}_1^2 \leq 4^{n+1}(1-p)^{2L}.
\end{align}
\end{corollary}
\begin{proof}
We will exploit the inequality $\norm{\cdots}_1^2 \leq 2^n \norm{\cdots}_2^2$ and the identity $\norm{H}_2^2 = \frac{1}{2^n}\sum_{P\in\calP_n}\Tr[HP]^2$, which holds for all Hermitian operators $H$. We obtain
    \begin{align}
      \bbE_{\calC\sim\calD_{\mathrm{circ}}} \norm{\calC(\rho-\sigma)}_1^2 
      \leq & 2^n \bbE_{\calC\sim\calD_{\mathrm{circ}}} \norm{\calC(\rho-\sigma)}_2^2 
      \\=   &\bbE_{\calC\sim\calD_{\mathrm{circ}}} \sum_{P\in\calP_n} \Tr[P\calC(\rho-\sigma)]^2
       \leq 4^{n+1}(1-p)^{2L},
    \end{align}
    where  we invoked Theorem~\ref{thm:ed} in the last step.
\end{proof}
A consequence of the above results is that any random noisy circuit can be truncated to an effective logarithmic depth for the task of estimating expectation values. 
We formalize this intuition with the following Corollary.
\begin{corollary}[Noise-induced shallow depth]
\label{thm:supp-eff}
Let $O$  be an observable and $\rho$ be an initial state. Consider the ``truncated'' noisy circuit 
\begin{align}
        \mathcal{C}_{[L, L-j]}  \coloneqq \mathcal{V}^{\mathrm{single}} \circ \mathcal{N}^{\otimes n} \circ \mathcal{U}_{L} \circ \mathcal{N}^{\otimes n} \circ \mathcal{U}_{L-1} \circ \dots \circ \mathcal{N}^{\otimes n} \circ \mathcal{U}_{L-j}, 
\end{align}
for a suitable $j \in \mathcal{O}(\log(\epsilon^{-1}\delta^{-1}))$.
With probability at least $1-\delta$ over the choice of $\mathcal{C}$, it holds that 
\begin{align}
    \abs{\Tr[O\mathcal{C}(\rho)] - \Tr[O\mathcal{C}_{[L, L-j]}(\sigma)]} \leq \epsilon \|O\|,
\end{align}
where $\sigma$ is an arbitrary state.
\end{corollary}
\begin{proof}
   The Corollary follows as an immediate consequence of Theorem\ \ref{thm:ed}. 
   We simply replace the channel $\calC$ in the statement of Theorem\ \ref{thm:ed} with $\mathcal{C}_{[L, L-j]} $. Then we obtain that, for any two quantum states $\sigma, \sigma'$:
    \begin{align}
        \bbE_{\calC\sim\calD_{\mathrm{circ}}} \Tr[O \mathcal{C}_{[L, L-j]} (\sigma-\sigma')]^2 \leq 4(1-p)^{2j} \norm{O}^2.
    \end{align}
    Replacing $\sigma'$ with $\mathcal{N}^{\otimes n} \circ \mathcal{U}_{L-j -1} \circ \dots \circ \mathcal{N}^{\otimes n} \circ \mathcal{U}_{1}(\rho)$, we obtain
        \begin{align}
        \bbE_{\calC\sim\calD_{\mathrm{circ}}} \Tr[O(\mathcal{C}_{[L, L-j]} (\sigma) - \calC(\rho))]^2 \leq 4(1-p)^{2j} \norm{O}^2.
    \end{align}
    By Markov's inequality, we can choose a suitable  $j \in \mathcal{O}(\log(\epsilon^{-1}\delta^{-1}))$ such that, with probability $1-\delta$, it holds that
    \begin{align}
    \abs{\Tr[O\mathcal{C}(\rho)] - \Tr[O\mathcal{C}_{[L, L-j]}(\sigma)]} \leq \epsilon \|O\|.
\end{align}
This concludes the proof of the Corollary.
\end{proof}

\section{Necessity of the random circuit assumption}
\label{app:necessity}
We have stated that we need to require \textit{some} average-case assumption if we want to simulate noisy-circuits for any possibly non-unital noise. However, one might question why this assumption is specifically made for random circuits rather than, for instance, random input states with fixed circuits, as done in Ref.~\cite{schuster2024polynomial}. 
The reason is that we can show that relying solely on the randomness in the input state is impossible, assuming BPP $\neq$ BQP.

This can be seen by still leveraging Ref.~\cite{ben2013quantum}, which shows how to perform fault-tolerant quantum computation in a fully quantum-coherent fashion without mid-circuit measurements and fresh auxiliary qubits through a careful compilation strategy, when the noise is non-unital. The idea is to use such a compilation strategy after waiting for a small time during which the noise drives the system (initially in its random state) towards its fixed point. We provide the detailed proof below.

\begin{theorem}
    If $\mathrm{BPP} \neq \mathrm{BQP}$, there exists no efficient classical algorithm to simulate non-unital noisy quantum circuits using fixed circuits and only a randomness assumption over the input state choice (i.e., with high probability over the input state sampled from some distribution). 
\end{theorem}

\begin{proof}
Consider a local non-unital noise $\mathcal{N}$, specifically a probabilistic reset channel $\mathcal{N}$ that, with some probability $p > 0$, keeps the state invariant, and with probability $1-p$, resets the state in the zero computational basis state. 

Assume the existence of a classical algorithm capable of estimating expectation values of observables for noisy quantum circuits under the assumption of randomness over the input states (with high probability over the input state distribution). We demonstrate how such an algorithm can simulate fault-tolerant quantum computation. 

Let $\mathcal{C}$ be any poly-size noiseless quantum circuit. We show how to classically estimate the expectation value of $Z_1$ with arbitrary accuracy $\varepsilon$ (thus solving a BQP-complete problem) using the following procedure:

\begin{enumerate}
    \item Assume an initial input state $\rho_0$ sampled from a distribution (the specific distribution is irrelevant to the argument).
\item Let the system evolve for $L = O(\log(n/\varepsilon))$ time steps without applying gates, allowing the noise to act independently across $n$ qubits. The resulting state is $\mathcal{N}_L^{\otimes n}(\rho_0)$, where 
   \[
   \mathcal{N}_L \coloneqq \underbrace{\mathcal{N} \circ \mathcal{N} \circ \cdots \circ \mathcal{N}}_{L \text{ times}}.
   \]
\item Apply $\Phi_{\mathcal{C}}$, the fault-tolerant implementation of $\mathcal{C}$ as described in Ref.~\cite{ben2013quantum}. This implementation allows estimation of $\Tr(Z_1 C(\ketbra{0^n}{0^n}))$ with accuracy $\varepsilon/2$. 
\end{enumerate}
The classical algorithm is promised to be able to estimate the $Z_1$ expectation value of the state 
\[
\rho' = \Phi_{\mathcal{C}} \circ \mathcal{N}_L^{\otimes n}(\rho_0),
\] 
with precision $\varepsilon/2$ and high probability over the choice of $\rho_0$. We now prove that this estimate also approximates $\Tr(Z_1 C(\ketbra{0^n}{0^n}))$ to within $\varepsilon$. 
We begin with:
\begin{align}
    |\Tr(Z_1 C(\ketbra{0^n}{0^n})) - \Tr(Z_1 \rho')| 
    &\leq |\Tr(Z_1 \Phi_{\mathcal{C}}(\ketbra{0^n}{0^n})) - \Tr(Z_1 \rho')|  + |\Tr(Z_1 C(\ketbra{0^n}{0^n})) - \Tr(Z_1 \Phi_{\mathcal{C}}(\ketbra{0^n}{0^n}))| \\
    &\leq |\Tr(Z_1 \Phi_{\mathcal{C}}(\ketbra{0^n}{0^n})) - \Tr(Z_1 \rho')| + \varepsilon/2.
\end{align}

It suffices to show that $|\Tr(Z_1 \Phi_{\mathcal{C}}(\ketbra{0^n}{0^n})) - \Tr(Z_1 \rho')| \leq \varepsilon/2$ to complete the proof. Using Hölder's inequality, it is enough to bound:
\begin{align}
    \norm{\Phi_{\mathcal{C}} \circ \mathcal{N}_L^{\otimes n}(\rho_0) - \Phi_{\mathcal{C}}(\ketbra{0^n}{0^n})}_1.
\end{align}
By the data-processing inequality~\cite{wilde2013quantum}, we have:
\begin{align}
    \norm{\Phi_{\mathcal{C}} \circ \mathcal{N}_L^{\otimes n}(\rho_0) - \Phi_{\mathcal{C}}(\ketbra{0^n}{0^n})}_1 
    &\leq \norm{\mathcal{N}_L^{\otimes n}(\rho_0) - \ketbra{0^n}{0^n}}_1.
\end{align}
We can write:
\begin{align}
    \mathcal{N}_L^{\otimes n}(\rho_0) = \mathrm{Prob}^{(L)}(\text{All reset}) \ketbra{0^n}{0^n} + (1-\mathrm{Prob}^{(L)}(\text{All reset})) \rho_{\text{rest}},
\end{align}
 where $ \mathrm{Prob}^{(L)}(\text{All reset}) $ is the probability that every qubit is reset at time step $L$, and  $\rho_{\text{rest}}$ is a quantum state representing the case in which not all qubits are reset. Then, because of triangle-inequality, we have:
\begin{align}
    \norm{\mathcal{N}_L^{\otimes n}(\rho_0) - \ketbra{0^n}{0^n}}_1 
    &\leq 2 (1-\mathrm{Prob}^{(L)}(\text{All reset}))=2\mathrm{Prob}^{(L)}(\text{Exists one not reset}).
\end{align}
The probability of at least one qubit not resetting is:
\begin{align}
    \mathrm{Prob}^{(L)}(\text{Exists one not reset}) 
    &\leq n \mathrm{Prob}^{(L)}(\text{1st qubit not reset}) = n (1-p)^L.
\end{align}
By choosing $L = O(\log(n/\varepsilon))$, we ensure:
\begin{align}
    \norm{\mathcal{N}_L^{\otimes n}(\rho_0) - \ketbra{0^n}{0^n}}_1 \leq \varepsilon/2.
\end{align}
This concludes the proof. 
\end{proof}

\section{Additional numerical results}
\label{sec:addnumer}

This section provides additional numerical results that complement those presented in the main text. As discussed in the main text, simulating deep noisy random circuits with non-unital noise (e.g., amplitude damping) is arguably more interesting than simulating circuits in noiseless or depolarizing scenarios, where simply outputting zero is often sufficient~\cite{mcclean2018barren, wang2020noise}. 

For amplitude damping (or any non-unital noise), this strategy fails due to the absence of barren plateaus, established in Ref.~\cite{mele2024noise} for circuits with local 2-design gates. This leads to observable expectation values exhibiting large variances, making the zero-output approach ineffective. Our large-scale numerical simulations, shown in Fig.~\ref{fig:barren} confirm that amplitude damping induces the absence of barren plateaus in our ansatz, which is the same we used  in Figs.~\ref{fig:ampdamp_error} and~\ref{fig:dephasing_error} and which does not satisfy the local 2-design assumptions used in Ref.~\cite{mele2024noise}.  These results highlight the necessity for non-trivial classical algorithms to accurately simulate such noisy circuits, even at high depths.
\begin{figure}[h]
    \centering
    \includegraphics[width=0.5\linewidth]{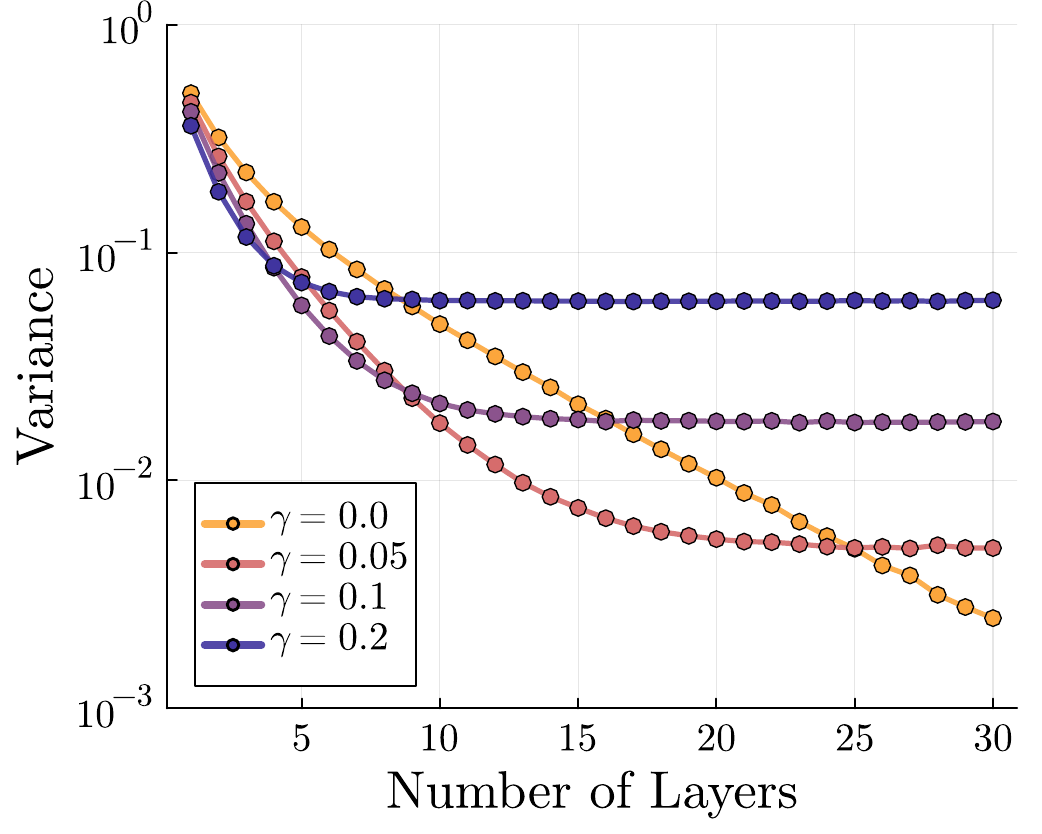}
    \caption{\textbf{Absence of barren plateaus with amplitude damping noise}. The variance, estimated via Theorem\ \ref{thm:num-estimate} with $10^6$ randomly sampled Pauli paths, of a $Z$ expectation value in the middle of a 1D lattice is shown for a system of $60$ qubits. Here, we used the same ansatz as in Figs.~\ref{fig:ampdamp_error} and~\ref{fig:dephasing_error}, specifically a 1D bricklayer topology of RX-RZ-RZZ gates. The parameter $\gamma$ in the plot refers to the amplitude damping rate.}
    \label{fig:barren}
\end{figure}

\end{document}